%% file: main.tex
\documentclass[a4paper,UKenglish,cleveref, autoref, authorcolumns, thm-restate]{lipics-v2021}

\input{header}

\title{Static to Dynamic Correlation Clustering}

\author{Nairen Cao}{New York University, United States of America}{nc1827@nyu.edu}{https://orcid.org/0000-0002-4961-763X}{}
\author{Vincent Cohen-Addad}{Google Research, United States of America}{cohenaddad@google.com}{https://orcid.org/0000-0002-0779-8962}{}
\author{Euiwoong Lee}{University of Michigan, United States of America}{euiwoong@umich.edu}{https://orcid.org/0000-0003-1454-7587}{Supported in part by NSF grant CCF-2236669 and Google.}
\author{Shi Li}{Nanjing University, China}{shili@nju.edu.cn}{https://orcid.org/0000-0001-9140-9415}{Affiliated with the School of Computer Science in Nanjing University, and supported by the State Key Laboratory for Novel Software Technology, and the New Cornerstone Science Laboratory.}
\author{David Rasmussen Lolck}{University of Copenhagen, Denmark}{dalo@di.ku.dk}{https://orcid.org/0000-0001-8835-0926}{Supported by VILLUM Foundation Grant 54451, Basic Algorithms Research Copenhagen (BARC).}
\author{Alantha Newman}{Université Grenoble Alpes, France}{alantha.newman@grenoble-inp.fr}{https://orcid.org/0009-0009-7353-7734}{}
\author{Mikkel Thorup}{University of Copenhagen, Denmark}{mthorup@di.ku.dk}{https://orcid.org/0000-0001-5237-1709}{Supported by VILLUM Foundation Grant 54451, Basic Algorithms Research Copenhagen (BARC).}
\author{Lukas Vogl}{EPFL, Switzerland}{lukas.vogl@epfl.ch}{https://orcid.org/0000-0002-8241-536X}{Supported by the Swiss National Science Foundation project 200021-184656 "Randomness in Problem Instances and Randomized Algorithms".}
\author{Shuyi Yan}{University of Copenhagen, Denmark}{shya@di.ku.dk}{https://orcid.org/0000-0001-9439-8942}{Supported by VILLUM Foundation Grant 54451, Basic Algorithms Research Copenhagen (BARC).}
\author{Hanwen Zhang}{University of Copenhagen, Denmark}{hazh@di.ku.dk}{https://orcid.org/0000-0002-3149-7799}{Supported by VILLUM Foundation Grant 54451, Basic Algorithms Research Copenhagen (BARC),  Independent Research Fund Denmark, grant 1054-00032B, and the Carlsberg Foundation, grant CF24-1929. }

\authorrunning{Cao et al.} %

\Copyright{Nairen Cao, Vincent Cohen-Addad, Euiwoong Lee, Shi Li, David Rasmussen Lolck, Alantha Newman, Mikkel Thorup, Lukas Vogl, Shuyi Yan and Hanwen Zhang} %



\category{Track A: Algorithms, Complexity and Games} %

\relatedversiondetails{Full version}{https://arxiv.org/abs/2504.12060}

\acknowledgements{}%

\nolinenumbers %

\EventEditors{Sayan Bhattacharya, Danupon Nanongkai, Michael Benedikt, and Gabriele Puppis}
\EventNoEds{4}
\EventLongTitle{53rd International Colloquium on Automata, Languages, and Programming (ICALP 2026)}
\EventShortTitle{ICALP 2026}
\EventAcronym{ICALP}
\EventYear{2026}
\EventDate{July 7--10, 2026}
\EventLocation{Royal Holloway, University of London, Egham, United Kingdom}
\EventLogo{}
\SeriesVolume{374}
\ArticleNo{127}

\begin{document}

\maketitle

\input{abstract}

\clearpage

\input{introduction}

\input{dynamic-framework}

\input{dynamicPivot}
\input{preclustering}
\input{dynamicLocalSearch-new}
\input{clusterLP}

\bibliography{reference}

\appendix

\input{appendix}
\input{improvedLocalSearch}

\end{document}

%% file: header.tex
\graphicspath{{figs/}}
\usepackage{subcaption}
\usepackage{amsmath,amssymb,amsfonts, amsthm}
\usepackage{xcolor}
\usepackage{xspace}
\usepackage{mathtools}

\usepackage{xspace}
\usepackage{subfiles}
\usepackage{bm}
\usepackage{tabularx}
\usepackage{multirow}

\usepackage[symbol]{footmisc}

\usepackage{soul}

\usepackage{algorithm}
\usepackage{algpseudocode}
\algtext*{EndWhile}
\algtext*{EndIf}
\algtext*{EndFor}
\algrenewcommand\algorithmicrequire{\textbf{Input:}}
\algrenewcommand\algorithmicensure{\textbf{Output:}}

\usepackage{tikz}
\usetikzlibrary{fit,shadows,positioning,shapes.misc,decorations.pathreplacing,calc}

\newcommand{\cost}{\mathrm{cost}}
\newcommand{\obj}{\mathrm{obj}}

\DeclareMathOperator*{\E}{{\mathbb{E}}}

\newcommand{\calC}{{\mathcal{C}}}

\newcommand{\poly}{{\mathrm{poly}}}

\newcommand{\adm}{\mathrm{adm}}
\newcommand{\dc}{d_{\mathrm{cross}}}
\newcommand{\core}{\mathrm{core}}

\newcommand{\supp}{\mathrm{supp}}
\newcommand{\covers}{\mathrm{cover}}
\newcommand{\opt}{\textsc{Opt}}

\newcommand{\eps}{\varepsilon}
\newcommand{\covereps}{\gamma}

\newcommand{\coverClusterLP}{{covering cluster LP}\xspace}

\newcommand{\dadm}{d_{\mathrm{adm}}}
\newcommand{\Nadm}{N_{\mathrm{adm}}}
\newcommand{\Ncand}{N_{\mathrm{cand}}}

\newcommand{\hatp}{\hat{p}}
\newcommand{\zt}{z^{(t)}}
\newcommand{\pt}{p^{(t)}}

\newcommand{\zmwu}{z^{\star}}
\newcommand{\tmwu}{T_{\mathrm{MW}}}

\newcommand{\pureclusterlpratio}{1.485}

\allowdisplaybreaks

\newcounter{enumcount}

\newcommand\drop[1]{}

%% file: abstract.tex
\begin{abstract}
    Correlation clustering is a well-studied problem, first proposed by Bansal, Blum, and Chawla [Mach. Learn. '04]. 
    The input is an unweighted, undirected graph. The problem is to cluster
    the vertices so as to minimize the number of edges between vertices in different clusters and missing edges between vertices inside the same cluster. This problem has a wide application in data mining and machine learning. 
    We introduce a general framework that transforms  existing static correlation clustering algorithms into fully-dynamic ones that work against an adaptive adversary. 

    We show how to apply our framework to known efficient correlation clustering algorithms, starting from
    the classic 3-approximate Pivot algorithm from Ailon, Charikar and Newman [JACM'08]. Applied to the most recent sublinear 
    $1.485$-approximation algorithm from Cao, Cohen-Addad, Lee, Li, Lolck, Newman, Thorup, Vogl, Yan and Zhang [STOC'25] \footnote[1]{\label{foot:1.485} The conference paper from Cao, Cohen-Addad, Lee, Li, Lolck, Newman, Thorup, Vogl, Yan and Zhang [STOC'25] claimed an approximation factor of $1.437$, based on result from Cao, Cohen-Addad, Lee, Li, Newman and Vogl [STOC'24]. However, the STOC'24 paper has a subtle bug, which was fixed in the arXiv version with the correct ratio of $1.485$. }, we get an $1.485$-approximation fully-dynamic algorithm that works with worst-case constant update time. The original static algorithm gets its approximation factor with constant probability, and we get the same against an adaptive adversary in the sense that for any given update step, not known to our algorithm, our solution is an $1.485$-approximation with constant probability when we reach this update.
     
    Most of previous dynamic algorithms, including the celebrated result from Behnezhad, Charikar, Ma and Tan [FOCS'19], had approximation factors around $3$ in expectation, and they could only handle an oblivious adversary. 
    A recent algorithm by Braverman, Dharangutte, Pai, Shah, and Wang [AISTATS'25] handles an adaptive adversary, but it has a large unspecified constant approximation ratio. 
    This contrasts with our general transformation, which works with all the best approximation factors known for the static case.
\end{abstract}

%% file: introduction.tex
\section{Introduction}

Correlation clustering is a classic clustering problem. Given an undirected unweighted graph $G=(V,E)$, our goal is to compute a clustering of the vertices that minimizes the number of edges between vertices in different clusters and missing edges between vertices inside the same cluster. 

In this work, we address correlation clustering in the dynamic setting, where the graph $G$ is updated by inserting and deleting edges.\footnote{In other works, it is usual to describe the input as being a complete graph with edges of weight either $+1$ or $-1$. Since we are working in a dynamic setting, and care about the running time, it makes more sense to use the convention that $+1$ edges are edges in the input graph and $-1$ edges are the non-edges.} 
Our goal is to maintain a clustering that is good with respect to the correlation clustering problem for the current graph. 
We present a framework that transforms efficient algorithms for the static correlation clustering problem into a dynamic algorithm that can handle edge updates against an adaptive adversary, at almost no cost to the approximation ratio. However, the efficiency has to be relative to a compressed graph representation that we shall introduce later.
We will apply this transformation to all the best static near-linear time algorithms.

\begin{figure}[hb]
    \centering
    \begin{subfigure}{0.45\linewidth}
        \input{figures/cctikzv2_1}
    \end{subfigure}
    \begin{subfigure}{0.45\linewidth}
        \input{figures/cctikzv2_2}
    \end{subfigure}
    \caption{A clustering of a graph and its cost. Red solid lines marks all edges between vertices in different clusters. Red dashed lines mark all the missing edges between vertices in the same cluster. This formulation is equivalent to the signed graph version, where $+$ edges are treated as edges and $-$ edges are treated as non-edges. }
    \label{fig:cctikz}
\end{figure}

\subsection{Prior work} 
The correlation clustering problem was first studied by Bansal, Blum, and Chawla in \cite{BBC04}. 
The model has a number of applications such as clustering ensembles \cite{bonchi2013overlapping}, duplicate detection \cite{arasu2009large}, community mining \cite{chen2012clustering}, disambiguation tasks \cite{kalashnikov2008web}, and automated labeling \cite{agrawal2009generating, chakrabarti2008graph}. Correlation clustering is of fundamental importance for the machine learning and data mining communities and a large body of work for solving correlation clustering in practice keeps on appearing at flagship conferences in these areas~\cite{DMM2024prunedPivot,BravermanDPSW25,ChakrabartyM-NIPS23,PanPORRJ15,ShiDELM21,pmlr-v162-veldt22a}. 

\paragraph*{Polynomial time approximation}
The first (large) constant approximation algorithm is due to \cite{BBC04}. The constant was then improved to $4$ by \cite{CGW05} who also proved that the problem is APX-hard. Later, Ailon, Charikar, and Newman \cite{ACN08} introduced the crucial idea of pivot-based algorithms, where in each round, the algorithm picks a random unclustered vertex and creates a new cluster consisting of this vertex and some of its neighbors. 
They gave a combinatorial 3-approximation algorithm and improved the approximation ratio to 2.5 by rounding a standard linear program (LP) which has an integrality gap of at least 2. \cite{CMSY15} improved the approximation ratio to 2.06 using a more sophisticated rounding scheme of the same LP.

To bypass the integrality gap of 2, Cohen-Addad, Lee, and Newman \cite{CLN22} used the Sherali-Adams hierarchy and achieved a $1.995$-approximation, which was then improved to a $1.73$-approximation by combining the pivot-based rounding with the newly-developed set-based rounding \cite{CLLN23+}. The best approximation algorithm currently known is a $\pureclusterlpratio^{\,\ref{foot:1.485}}$-approximation due to \cite{cao2024understanding} which proposed the Cluster LP to achieve this result. They
also showed that obtaining a 25/24-approximation is NP-hard.

\paragraph*{Deterministic algorithms} 
When it comes to deterministic algorithms, the fastest one with (large) constant factor approximation is Algorithm 1 in \cite{CLMNP21} which has a trivial deterministic implementation is $O(nm)$ time. A deterministic factor $(3+\eps)$- approximation is presented in \cite{fischer2025faster} running in $\tilde O(n^3)$ time. One can obtain a deterministic 2.5-approximation in (large) polynomial time by first solving the $O(n^3)$-sized LP with the Ellipsoid method and rounding it with the work of~\cite{DBLP:journals/mor/ZuylenW09}.
\footnote{In the conference version of \cite[Proposition 2.1]{BravermanDPSW25}, the authors claim that the existence proof of \cite{DBLP:conf/innovations/Assadi022} gives a deterministic $O(m)$ time constant-approximate algorithm. The authors have since retracted their claim in the most recent arXiv version \cite{braverman2024fully}. 
}

\paragraph*{Near-linear time}
For the purpose of transforming static algorithms into dynamic ones, our best hope are those with near-linear running time, that is $\tilde O(m)$-time, where $m$ is the number of edges. 
We remark that while sublinear time algorithms exist, their sublinear time holds only when $m = \omega(n)$. 
As we have to deal with the case $m = O(n)$, these algorithms are not different from those with a near-linear running time for our purpose.
The best deterministic running time of $O(nm)$ from \cite{CLMNP21} is thus far too slow.

Among all near-linear time algorithms, the first one is the classical linear time $3$-approximate Pivot algorithm of \cite{ACN08}. 
The first sublinear time algorithm in $O(n\log^2 n)$ time with constant approximation ratio is from \cite{DBLP:conf/innovations/Assadi022}, which is based on the parallel algorithm of \cite{CLMNP21}. 
We also have the more recent local-search based $1.847$-approximation of \cite{CLMTYZ24} which runs in sublinear time. 
Recently, in \cite{cao2025fastLP}, it was shown that the $\pureclusterlpratio$-approximation algorithm based on the Cluster LP from \cite{cao2024understanding} can be implemented in sublinear time. 
All of these near-linear time algorithms are Monte Carlo algorithms.

\paragraph*{Dynamic clustering with low approximation ratio} 
The dynamic setting has also been explored. 
Much of the work here has been on how to maintain a pivot-based clustering under these changes. One result is that of \cite{Behnezhad2019FullyDM}, which shows how to maintain a Maximal Independent Set with an update time of $O(\log^2 n\log^2\Delta)$ where $\Delta$ is the maximum degree of any vertex. 
This can be used to maintain the $3$-approximate Pivot algorithm for a fixed permutation of pivots. This was since improved in \cite{DMM2024prunedPivot} to maintaining a $(3+\eps)$-approximation with an update time of $O(1/\eps)$. 
Finally, it has recently been shown that the barrier of $3$-approximation indeed can be broken in the dynamic setting in \cite{behnezhad2024fullydynamicCC}. 
Here they give a $2.997$-approximation algorithm, again based on maintaining the pivot under additions and deletions with a polylogarithmic update time. 
\emph{All of these pivot-based results only work with an oblivious adversary}: they fix the random order of the pivots in advance, maintaining it for all updates. 
All of these results only work with an oblivious adversary: they fix the random order of the pivots in advance, maintaining it for all updates. 

\paragraph*{Dynamic against adaptive adversary}
For all the above pivot-based dynamic algorithms, an \emph{adaptive adversary} could easily learn the order of the pivots, and then make them perform very badly afterward by constructing an input graph that is bad for this pivot order. 
Recently, \cite{BravermanDPSW25} presented a dynamic large constant factor approximation algorithm supporting updates in $O(\log^2 n)$ amortized time against an adaptive adversary. 
It is based on the non-pivot-based $O(n\log^2 n)$ time randomized algorithm from \cite{DBLP:conf/innovations/Assadi022}, but at the cost of a large approximation factor. They also mention ``\textit{It is unclear how the pivot-based algorithms could be made to work in the adversarially robust setting}''\cite[2nd page]{BravermanDPSW25}.  In particular, this concerns the classic 3-approximate Pivot algorithms.

The general issue of getting a Monte Carlo randomized dynamic algorithm to work against an adaptive adversary is addressed in \cite{10.1145/3519935.3520064}. 
They show a generic transformation from the oblivious setting to the adaptive setting using tools from differential privacy, but the transformation has an extra polynomial factor on the update time even in the few cases studied in their paper. In this paper, however, we are aiming for update times that are polylogarithmic or even constant.

\paragraph*{Clustering for dynamic streams}
A recent result from \cite{assadi2025cc}
shows that for a fully-dynamic graph, we can maintain a linear sketch of sublinear $\tilde O(n)$ size from which we can derive a correlation clustering in randomized polynomial time. 
The basic idea is that they desparsify the sketch into a graph $G'$ that is similar to the current graph $G$ in the sense that any correlation clustering will have almost the same cost in $G$ and $G'$. 
They can therefore apply any polynomial time correlation clustering algorithm to $G'$ and get an almost as good correlation clustering for $G$. 
This result follows the streaming tradition and is very interesting from an information theoretic perspective. 
However, spending polynomial time whenever we want a clustering of the current graph is not good from our perspective of dynamic graph algorithms, where we want to maintain a concrete clustering in constant or logarithmic time per edge update. 

The interesting aspect from \cite{assadi2025cc} is that their sketch is of $\tilde O(n)$ size. 
A fundamental issue from a fully-dynamic perspective is that a randomized sketch of sublinear size appears to sacrifice the ability to handle a non-oblivious adaptive adversary. 
The issue is that if we use a randomized sketch of the current graph, and if we after each update reveal a clustering based on the random choices in the sketch, then the adversary can learn about these random choices and then make updates that destroy the quality of the sketch. 
This is not so much an issue for the streaming scenario in \cite{assadi2025cc} where they just maintain the sketch, and only in the end construct the clustering in polynomial time. However, it is a major issue if we want to maintain a clustering throughout the updates against an adaptive adversary, and this is the setting we are considering in this paper.

\subsection{Our results and techniques}

In this paper, we introduce a general
framework that transforms efficient static algorithms for the static correlation clustering problem into dynamic algorithms that work against an adaptive adversary. However, the efficiency has to relate to a compressed graph representation that we shall introduce later. In particular, this will work for the pivot-based algorithms that, as mentioned above, have so far failed against adaptive adversaries. 

We apply our framework to all near-linear time correlation clustering algorithms with low approximation ratio mentioned before:
the classical 3-approximate Pivot algorithm from \cite{ACN08}, the $1.847$-approximation local search algorithm from \cite{CLMTYZ24} and the most recent \pureclusterlpratio-approximation Cluster LP algorithm from \cite{cao2025fastLP}. 
The latter gives us our main fully-dynamic result:
\begin{theorem}\label{thm:dyn-clusterLP}
    For any $\delta \in (0, 1)$ we can maintain clustering for a fully-dynamic graph in $O(\log 1/\delta)$ worst-case time per edge update. 
    Against an adaptive adversary, for each $i$ the maintained clustering is a $\pureclusterlpratio$-approximation with probability at least $1 - \delta$ at update $i$. 
\end{theorem}

Another way to explain \cref{thm:dyn-clusterLP} is that the adversary can choose any $i$ hidden to the algorithm before the updates start and adaptively choose the first $i$ updates. 

A particular situation of \cref{thm:dyn-clusterLP} is when $\delta = \frac{1}{\poly(n)}$. 
Then we get polynomially small error probability with logarithmic update time. 
This matches the setting of \cite{BravermanDPSW25}. 

Some modifications we make to the static algorithms are of independent interest. 
Currently, the best implementations of the algorithms from \cite{CLMTYZ24,cao2025fastLP} are sublinear in $\tilde O(n)$ time, but this includes logarithmic factor in the running time even for constant error probability. 
As a result, the algorithm does not run in $O(m)$ time if $m=o(n\log n)$.

We show that these algorithms can be implemented in strictly linear time, that is, $O(m)$ time.
We further show that the static algorithm from \cite{CLMTYZ24} can be implemented in $O(m)$ time, yielding a 1.847-approximation with probability at least $1-\exp(-\sqrt{m}/\log^3m)$. 
For such an exponentially low error probability,
we would normally have used $\Theta(\sqrt{m}/\log^3m)$ repetitions, leading to a corresponding polynomial blow-up in time. 
The same exponentially low error probability can be achieved for the computation of the Cluster LP in \cite{cao2025fastLP}, but the LP rounding brings the error probability up to an arbitrarily small constant if we want $O(m)$ time.

In the rest of this subsection we will introduce the framework and later sketch how the different static algorithms can be modified to work within it.

\paragraph*{Graph representation and compression via correlation clustering}
We will now be more precise about how we want to represent a graph $G=(V,E)$.
Instead of storing the edges $E$, we will store a clustering $\mathcal{C}$ together with the set $D$ of \emph{violated pairs}, that is, unordered vertex pairs $(u,v)$ such that either $(u,v)\in E$ but $u$ and $v$ are in different clusters in $\mathcal{C}$, or $(u,v)\not\in E$ but $u$ and $v$ are in the same cluster in $\mathcal{C}$. 
The correlation clustering cost of $\mathcal{C}$ is exactly $|D|$. 
We shall refer to set $D$ as the \emph{violation} of the clustering $\mathcal C$, and to $(\mathcal C,D)$ as the \emph{cluster representation}, noting that it uniquely defines the edge set $E$. 
We say that the cluster representation $(\mathcal C,D)$ is $c$-approximate if the clustering is $c$-approximate, that is, if $|D|$ is at most $c$ times the minimum correlation clustering cost. In addition, it is simple to reconstruct $E$ from $(\calC,D)$.

We can think of the cluster representation as a graph compression: It is never worse than the actual graph representation by more than the $O(n)$ space required to store $\mathcal{C}$, since we can always choose the clustering to be a singleton cluster for each vertex. 
Then $D$ is just the set $E$ of edges in the graph. 
In the other extreme, if we have all vertices forming a single cluster, then the violation $D$ is exactly the set of non-edges $\binom{V}{2}\setminus E$. 
However, our representation could be much better in both of these extremes if the clustering $\mathcal C$ has small cost. 
We can in fact bound the space required by $O(n + |D|)$, from which we get these properties. 
It is worth mentioning that one can ensure $O(n)$ additional space compared to the normal graph representation, with two global counters on the size of $D$ and $E$. 
However, this is not necessary to achieve our main results.

\paragraph*{Fully-dynamic framework}
To get a fully dynamic correlation clustering algorithm, assume that we have a static algorithm that takes as input an arbitrary cluster representation $(\mathcal{C}, D)$ and outputs a $c$-approximate cluster representation $(\mathcal{C}', D')$ in $O(t|D|)$ time for some parameter $t$.
Since $D$ may be much smaller than the edge set $E$, we can think of using $(\mathcal{C}, D)$ as a warm start. 
We also note that $|D|$ might be much smaller
than the clustering $\mathcal{C}$, which is of size $\Theta(n)$. 
Therefore, the static algorithm can only produce a clustering $\mathcal{C'}$ with limited modification of the input clustering $\mathcal C$.

Now, in correlation clustering, each edge update will only change the cost of a given clustering by $1$. Thus, if at some point, we have a good clustering $\mathcal{C}$ with violation $D$, then the same clustering will remain good for $O(\eps|D|)$ edge updates. 
More precisely, we will show that if $\mathcal{C}$ was a $c$-approximation, then $\mathcal{C}$ will remain a $(1+\eps)c$-approximation after $\mu|D|$ edge updates, where $\mu \le \frac{\eps}{2(1+\eps)c}$. 
Therefore, it suffices to reconstruct a new $c$-approximate clustering after $\mu|D|$ edge updates. 

In the period between reconstructions, we will just store the sequence $U$ of $\mu|D|$ edge updates (so the violation is not updated explicitly as that would involve hashing). 
For a fixed clustering $\mathcal C$, we will show how to update our violation $D$ according to the edge updates in $U$ deterministically in $O(|D|+|U|)=O(|D|)$ time. 
Next, we apply the assumed static algorithm to our updated cluster representation $(\mathcal{C}, D)$ to reconstruct a new $c$-approximate cluster representation $(\mathcal{C}', D')$ in $O(t|D|)$ time. 
The new clustering $\mathcal{C}'$ will again remain good for the next $\mu|D'|$ edge updates.

The result is a $(1+\eps)c$-approximate fully-dynamic algorithm with $O(t/\mu)$ amortized update time. 
With standard background rebuilding, we can
de-amortize and get $O(t/\mu)$ worst-case update time. Summing up, we will show
\begin{theorem}
\label{thm:dynamic-intro}
Suppose we have a static correlation clustering algorithm that given any cluster representation $(\mathcal{C},D)$, produces a $c$-approximate cluster representation $(\mathcal{C'},D')$ in $O(t|D|)$ time.
Then we have a fully-dynamic algorithm that maintains a $(1+\eps)c$-approximate correlation clustering in worst-case $O(t/\mu)$ update time per edge, where $\mu \le \frac{\eps}{2(1+\eps)c}$.
\end{theorem}

\paragraph*{Tricky randomization}
The above approach works perfectly if our algorithm for the static case is deterministic, but all the known near-linear time static constant approximate algorithms are Monte Carlo algorithms. 
They only produce a $c$-approximations in expectation or with some error probability bounded by $p<1$. 

Naively, the above should be fine, even against an adaptive adversary, because our current clustering is good if the last rebuild was a $c$-approximation as in the deterministic case, and this happened with probability at least $1-p$. 
The next rebuild is done using its own independent randomness, so an adaptive adversary cannot do anything to break a guarantee that holds for any input cluster representation, e.g., that it is a $c$-approximation with probability at least $1-p$. 
Nevertheless, there is an issue since the new cluster representation $(\mathcal C',D')$ will be kept for $\mu|D'|$ updates. 
Hence, if it is bad in the sense that $D'$ is large, then it will survive for longer time. 
We will present a concrete example about this effect, showing that we get a logarithmic factor more chances of failing at a particular update. 
In fact, this does not even depend on the adversary being adaptive. 
The same construction can be done by an oblivious adversary.

\paragraph*{Error probability against adaptive adversaries}
It is a bit subtle what we mean by saying that the dynamic algorithm is correct with some probability against an adaptive adversary, since this adversary knows our current clustering and can see exactly how bad it is compared with the optimum. 
It is therefore never a probabilistic statement whether the current clustering is good to the adversary. 
However, we can study statements of this kind: for any given update step $i$, what is the probability $P_i$ that we have a bad clustering just after update $i$. 
An adaptive adversary (1) can choose $i$ in advance, (2) knows our current clustering at any time, and (3) can adaptively pick the first $i$ updates so as to maximize $P_i$. 
Yet we want to bound $P_i$. 

The above definition of $P_i$ may
seem a bit cryptic, but this is important when we want constant error probability with constant update time.
More specifically, with a long update sequence, we can create and remove a linear number of different constant sized graphs. Our dynamic algorithm will cluster all of them independently, so if it fails with constant probability, then we expect it to fail on a constant fraction of them, and the adaptive adversary will see when it fails. Nevertheless we promise that it is correct just after a predetermined update $i$ with constant probability $P_i$. If we increase the update time to $O(\log n)$, then we can make the error probability so small that we do not expect any errors over a polynomially long update sequence, and then we do not need the special definition of $P_i$.

\paragraph*{Fully-dynamic framework for Monte Carlo algorithms}
What we first show is that if the static
algorithm gives a $c$-approximation with probability at least $1-p$, then our fully dynamic algorithm gives an $(1+\eps)c$-approximation with error probability $P_i=O(p\log n)$. 

We can avoid losing the factor $\log n$ if the static algorithm for input clustering $(\mathcal C,D)$, $|D|\geq 1$, produces a $c$-approximation with probability at least $1-q(|D|)$ where $q$ falls at least inversely in $|D|$. 
Then $P_i$ is bounded by $O(q(1))$. 

Technically, the assumption of $|D|\geq 1$ means that we avoid dividing by $0$. 
However, if our input cluster representation had $|D|=0$, then it would be a zero-cost optimal solution, and then we would not involve the static algorithm to try computing a better solution.

Finally, we have a combined result.
Suppose the static algorithm for input clustering
$(\mathcal C,D)$, $|D|\geq 1$ produces a $c$-approximation with probability at least $1-p$ and an $O(c)$ approximation with probability at least $1-q(|D|)$ where again $q$ falls at least inversely in $|D|$. 
Then $P_i$ is bounded by $O(p+q(1))$.

We shall see how all these bounds play together with existing static algorithms, but first we summarize them in the theorem below.
In this theorem, we consider randomized Monte-Carlo algorithms that \emph{aim} at certain targets that they may \emph{fail} to achieve. 

\begin{theorem}
\label{thm:dynamic-random-intro}
Suppose we have a static correlation clustering algorithm that, given
any cluster representation $(\mathcal{C},D)$ with $|D|\geq 1$, aims to produce
a $c$-approximate 
cluster representation $(\mathcal{C'},D')$ in $O(t|D|)$ time. 
Then we have a fully-dynamic
algorithm that aims to maintain
a $(1+\eps)c$-approximate correlation clustering in worst-case $O(t/\mu)$ time per update where $\mu =\min\left\{\frac{\eps}{2(1+\eps)c},1/6\right\}$.

\begin{itemize}
    \item[(a)] Suppose that with probability at most $p$, the static algorithm \emph{fails} to produce a $c$-approximate clustering. 
    Then, against an adaptive adversary, for any fixed $i$ unknown to the algorithm, the clustering maintained by the fully-dynamic algorithm fails to be $(1+\eps)c$-approximate at update $i$ with probability $P_i = O(p\log n)$.
    \item[(b)] Suppose that with probability at most $q(|D|)$, which falls at least inversely with $|D|\geq 1$, the static algorithm \emph{fails} to produce a $c$-approximate clustering. 
    Then, against an adaptive adversary, for any fixed $i$ unknown to the algorithm the clustering maintained by the fully-dynamic algorithm fails to be $(1+\eps)c$-approximate at update $i$ with probability $P_i = O(q(1))$.
\end{itemize}
\end{theorem}

By combining the two cases above, we can get a best of both worlds statement, namely that we can get the approximation ratio of \cref{thm:dynamic-random-intro}(a) case while avoiding the logarithmic blow-up in the probability, by using the properties of \cref{thm:dynamic-random-intro}(b).
\begin{theorem}\label{thm:dynamic-mixed-case}
    Let $\mathcal A$ and $\mathcal {\hat A}$ be two static correlation clustering algorithms that take as input any cluster representation $(\mathcal{C},D)$ with $|D|\geq 1$ and output a cluster representation $(\mathcal{C}', D')$, both in $O(t|D|)$ time. 
    Furthermore, let the probability of $\mathcal A$ failing to produce a $c$-approximate solution be bounded by $p=1-\Omega(1)$ and the probability of $\mathcal {\hat A}$ failing to produce a $\hat c$ approximate solution be bounded by $q(|D|)$, where $q$ falls at least inversely with $|D|$.

    Then we have a fully-dynamic algorithm that aims to maintain a $(1+\eps)c$-approximate correlation clustering in worst-case $O(t/\mu)$ time per update where $\mu=\min\left\{\frac{\eps}{2(1+\eps)c},\frac{1}{6},\frac{1}{2\hat c}\right\}$. 
    Furthermore, against an adaptive adversary, for any fixed $i$ unknown to the algorithm, the clustering maintained by the fully-dynamic algorithm fails to be $(1+\eps)c$-approximate at update $i$ with probability $P_i = O(p+q(1))$.
\end{theorem}

It should be mentioned that while this theorem does give results against an adaptive adversary, this does not mean that it can be strengthened against an oblivious adversary. 
We are in fact later going to show a matching lower bound for the probability of failure, which contains a strategy that can be implemented both by an adaptive and an oblivious adversary.

It is also worth noting that the probability $P_i$ in the statement is not independent for different values of $i$. You only get independence between two updates $i$ and $j$ if between processing them the clustering is rebuild. The frequency of this is however dependent on the cost of the clustering, a value that is very easy to influence by the adversary.

One may wonder if one could get the error probabilities to work simultaneously for a sequence of $T$ updates without resorting to a union bound, multiplying the error probability by $T$. As a base state, consider an arbitrarily large graph consisting of disjoint cliques, including singleton vertices. The optimal correlation clustering cost is zero, so any approximation algorithm has to agree on the clustering into cliques. 
With $s$ updates, the adversary can create a nontrivial instance of correlation clustering over some of the vertices. 
And with $s$ more updates, the adversary can change the graph back to the base state. 
Therefore, starting from the base state, in $T$ updates, the adversary can create $T/(2s)$ nontrivial and independent instances to challenge the algorithm, so getting global error bound better than a union bound is not likely.

\paragraph*{Consequence if a deterministic constant approximate algorithm is found}
Our result will be improved if there exists a deterministic $O(|D|)$ time algorithm $\mathcal{\hat{A}}$ with constant approximation ratio.
According to \cref{thm:dynamic-intro}, this would give a deterministic dynamic algorithm that maintains a clustering of $(1+\eps)c$ approximation at every update.
In addition, the existence would also imply that the dynamic algorithm would inherit any expected approximation ratio from algorithm $\mathcal{A}$, instead of only achieving this approximation ratio with constant probability, since we can get a deterministic upper bound on the cost anytime we rebuild the clustering. 
This would also give a worst-case guarantee on the quality of the clustering even when $\mathcal{A}$ always fails. 

\begin{theorem}\label{thm:dynamic-deterministic-upper-bound}
    Let $\mathcal A$ and $\mathcal {\hat A}$ be two static correlation clustering algorithms that take as input any cluster representation $(\mathcal{C},D)$ with $|D|\geq 1$ and output a cluster representation $(\mathcal{C}', D')$, both in $O(t|D|)$ time. 
    Furthermore, let the probability of $\mathcal A$ failing to produce a $c$-approximate solution be bounded by $p=1-\Omega(1)$ and $\mathcal {\hat A}$ be a deterministic $\hat c$-approximation algorithm.

    Then we have a fully-dynamic algorithm that aims to maintain a $(1+\eps)c$-approximate correlation clustering in worst-case $O(t/\mu)$ time per update where $\mu=\min\left\{\frac{\eps}{2(1+\eps)c},\frac{1}{6},\frac{1}{2\hat c}\right\}$. 
    Furthermore, against an adaptive adversary, for any fixed $i$ unknown to the algorithm, the clustering maintained by the fully-dynamic algorithm is an $(1+\eps)c$-approximation in expectation at update $i$.
\end{theorem}

\paragraph*{Existing static algorithms modified for our dynamic framework}
We will now discuss how the known near-linear static  algorithms can be modified and used in our dynamic framework.  

First, as a general standard note. Suppose we have a Monte-Carlo algorithm that, given
a cluster representation $(\mathcal C,D)$, aims to produce
a $c$-approximate cluster representation $(\mathcal C',D')$ and fails with probability at most $p$. 
We switch to the new clustering $\mathcal C'$ only if it has lower cost, that is, if $|D'|< |D|$.
This means that if we make $k$ iterations of the algorithm, then the probability that we do not end up with a $c$-approximation drops to $p^k$. 

The first algorithm we consider is the classical 3-approximate Pivot algorithm from \cite{ACN08}. It takes
$O(m)$ time on a given graph, and we show how to implement it in $O(|D|)$ time for a given input clustering $(\mathcal C,D)$. This ends up requiring weighted sampling of vertices instead of the normal uniform one. In addition, this also has the consequence of showing how we can 
speed up the running time of the Pivot algorithm with a ``hot start''.
The approximation factor of the Pivot algorithm is in expectation, but making 
$O(\log\log n)$ iterations, we
get a $(3+o(1))$-approximation with failure probability $o(1/\log n)$. Thus, we have 
\begin{theorem}\label{thm:3pivot}
We have a static correlation clustering algorithm that given
any cluster representation $(\mathcal{C},D)$,  in 
$O(|D|)$ time produces
an expected 3-approximate 
cluster representation $(\mathcal{C'},D')$. Repeating
$O(\log\log n)$ times, we
get a $(3+o(1))$-approximate
solution with probability at least $1-o(1/\log n)$.
\end{theorem}
Plugging \cref{thm:3pivot} into \cref{thm:dynamic-random-intro}(a) using a sufficiently small $\eps>0$, we get
\begin{corollary}\label{cor:3pivot}
For any $\eps>0$, we have a fully-dynamic
algorithm that aims to maintain
a $(3+\eps)$-approximate correlation clustering in 
$O(\frac{1}{\eps}\log\log n)$ worst-case time per edge update.
Against an adaptive adversary, for any $i$ the clustering fails to be $(3+\eps)$-approximate at update $i$ with probability $P_i = o(1)$. 
The failure probability can be reduced to any  $\delta$ if we spend $O(\frac{1}{\eps}(\log\log n + \log 1/\delta))$ worst-case time per edge update.
\end{corollary}
The main advantage of \cref{cor:3pivot} over the previous pivot-based dynamic algorithms with approximation factors around 3 is that it works against an adaptive adversary. 
Compared to \cite{BravermanDPSW25}, by setting $\delta = 1/\poly(n)$, \cref{cor:3pivot} improves both the approximation ratio from a large constant to $3 + \epsilon$ and the update time from amortized $O(\log^2 n)$ to worst-case $O(\log n)$, with the same error probability.

Next, we consider the 1.847-approximate local search algorithm from \cite{CLMTYZ24}.
As discussed earlier, it uses $\Omega(n \log n)$ time even if we just want constant error probability, but this is not $O(m)$ time if $m=o(n\log n)$. 
In this paper, we show an $O(m)$ time implementation of the local search algorithm with exponentially small error probability of $\exp(-\sqrt{m}/\log^3m)$.
The previous static algorithm tries to find clusters with roughly the same probability regardless of their sizes. 
However, we notice that each small cluster only has a small random contribution to the total cost, so their total contribution is strongly concentrated and we do not need to find all of them. 
On the other hand, for large clusters, we can afford to spend more time and make sure to find all of them with high probability. 
Our full implementation is based on a smooth sampling distribution according to the cluster sizes. 
Moreover, our implementation can also be made efficient in cluster representation, as stated below.

\begin{theorem}\label{thm:local-search}
We have a static correlation clustering algorithm that given any cluster representation $(\mathcal{C},D)$, in $O(|D|)$ time produces
a cluster representation $(\mathcal{C'},D')$ that is below $1.847$-approximate with probability at least $1-\exp(-\sqrt{|D|}/\log^3|D|)$. 
\end{theorem}%

The next corollary follows from this special case of \cref{thm:local-search} when $\mathcal{C}$ is the set of singletons and $D = E$. 
\begin{corollary}
    \label{cor:linear-static-local-search}
    There exists a $1.847$-approximate $O(m)$ time static correlation clustering algorithm in the normal graph representation with error probability at most $\exp(-\sqrt{m}/\log^3 m)$. 
\end{corollary}

Since the error probability in \cref{thm:local-search} 
falls more than inversely in $|D|$, we can apply \cref{thm:dynamic-random-intro}(b). The approximation factor 1.847 in \cref{thm:local-search} is already rounded up by a small constant, so using a sufficiently small $\eps>0$ in \cref{thm:dynamic-random-intro}, we get
\begin{corollary}\label{cor:local-search}
We have a fully-dynamic algorithm that aims to maintain a $1.847$-approximate correlation clustering in constant worst-case time per edge update. 
Against an adaptive adversary, for any fixed $i$ unknown to the algorithm, the clustering is $1.847$-approximate with constant probability at update $i$. 
The failure probability can be reduced to any $\delta$ if we spend $O(\log 1/\delta)$ worst-case time per edge update. 
\end{corollary}
Finally, we want to apply our dynamic framework
to the current best \pureclusterlpratio-approximation
algorithm from \cite{cao2024understanding} which was implemented in sublinear time in \cite{cao2025fastLP}. 
We will need all the techniques mentioned above. The algorithm has two parts. First it solves the Cluster LP from \cite{cao2024understanding}, second it rounds the solution.

The first part aims for a fractional solution to the Cluster LP that is $(1+\eps)$-approximate relative to the optimal integral solution. 
This part uses multiplicative weight updates and finds a nearly optimal fractional solution to the aggregated constraint with the above local search. 
We face a situation parallel to the one we
faced with the local search from \cite{CLMTYZ24}.
The algorithm is sublinear but
not linear if $m=o(n\log n)$
and it only works with constant
probability. 
Using the techniques we used to make the local search from \cite{CLMTYZ24} work for Theorem \ref{thm:local-search},
we can take any cluster representation $(\mathcal C,D)$ and solve the Cluster LP in
$O(|D|)$ time, yielding a fractional solution that is $(1+\eps)$-approximate relative to the optimal integral solution with probability at least
$1-\exp(-\sqrt{|D|}/\log^3|D|)$.

In \cite{cao2024understanding,cao2025fastLP}, the fractional solution is rounded using pivot-techniques. 
The rounding gets a \pureclusterlpratio-approximate solution with constant probability. 
Given a cluster representation $(\mathcal C,D)$, we want to implement the rounding in $O(|D|)$. Here we can employ some of the techniques we used for the classical 3-approximate pivot in Theorem \ref{thm:3pivot}. 
However, all the rounding algorithms only work with constant failure probability.

\begin{theorem}\label{thm:clusterLP}
We have a static algorithm that given any cluster representation $(\mathcal{C},D)$, in $O(|D|)$ time solves the Cluster LP.
For any given constant $\eps>0$, the fractional solution is within a factor $(1+\eps)$ of the optimal integral solution with probability at least $1-\exp(-\sqrt{|D|}/\log^3|D|)$. 

The above fractional solution can be rounded to an
integral cluster representation in $O(|D|)$ time. The solution is below $\pureclusterlpratio$-approximate with constant
probability.
\end{theorem}

The next corollary follows from this special case of \cref{thm:clusterLP} when $\mathcal{C}$ is the set of singletons and $D = E$. 

\begin{corollary}
    \label{cor:linear-static-lp-solving}
    There exists a $1.485$-approximate $O(m)$ time static correlation clustering algorithm in the normal graph representation with constant error probability. 
\end{corollary}

Because of the constant error probability from the rounding, we do not currently benefit from
the exponentially low error probability for finding the fractional solution. However, it will be important if we one day find rounding algorithms working with higher probability.

We will now plug both Theorem \ref{thm:local-search} and Theorem \ref{thm:clusterLP} into Theorem \ref{thm:dynamic-mixed-case}. More precisely, we use Theorem \ref{thm:local-search} to get an $O(1)$-approximate solution with exponentially high probability, and Theorem \ref{thm:clusterLP} to get a below \pureclusterlpratio-approximate cluster representation with constant probability. Applying  Theorem \ref{thm:dynamic-mixed-case} with a sufficiently small $\eps>0$, we get our main result in Theorem \ref{thm:dyn-clusterLP}.

%% file: figures/cctikzv2_1.tex
\begin{tikzpicture}[
  scale=0.6,
  every node/.style={font=\tiny},
  nodeStyle/.style={
    circle, draw=black, thick, inner sep=1.8pt,
    minimum size=4mm, fill=white
  },
  posEdge/.style={thick},     %
  posEdgeBad/.style={draw=red!70!black, very thick}, %
  negEdge/.style={draw=red!70!black, dashed, very thick}, %
  clusterBox/.style={
    rounded corners=6mm, draw=#1!60, fill=#1!20,
    fill opacity=0.15, very thick
  },
  peel/.style={rounded corners=6mm, draw=green!60!black,
               fill=green!15, fill opacity=0.18, very thick},
  super/.style={circle, double, draw=black,
                minimum size=6mm, very thick},
]

\node[nodeStyle,fill=blue!15] (a1) at (1.3,1.4)  {1};
\node[nodeStyle,fill=blue!15] (a2) at (0.9,2.4){2};
\node[nodeStyle,fill=blue!15] (a3) at (0,2.7){3};
\node[nodeStyle,fill=blue!15] (a4) at (-0.9,1.9){4};
\node[nodeStyle,fill=blue!15] (a5) at (-0.6,0.9){5};
\node[nodeStyle,fill=blue!15] (a6) at (0.8,0.8) {6};

\node[clusterBox=blue, fit=(a1) (a2) (a3) (a4) (a5) (a6), inner sep=10pt] (boxA) {};

\draw[posEdge] (a1) -- (a2) -- (a3) -- (a4) -- (a5) -- (a6) -- (a1);
\draw[posEdge] (a2) -- (a4) -- (a6) -- (a2);
\draw[posEdge] (a3) -- (a6);
\draw[posEdge] (a3) -- (a5) -- (a1);

\node[nodeStyle,fill=green!15] (b1) at (4.0,2.5) {7};
\node[nodeStyle,fill=green!15] (b2) at (5.0,2.1) {8};
\node[nodeStyle,fill=green!15] (b3) at (4.6,1.3) {9};
\node[nodeStyle,fill=green!15] (b4) at (3.4,1.2) {10};

\node[clusterBox=green, fit=(b1) (b2) (b3) (b4), inner sep=10pt] (boxB) {};

\draw[posEdge] (b1) -- (b2) -- (b3) -- (b4) -- (b1);
\draw[posEdge] (b2) -- (b4);

\node[nodeStyle,fill=orange!15] (c1) at (-0.5,-2.0) {11};
\node[nodeStyle,fill=orange!15] (c2) at (0.6,-1.6)  {12};

\node[clusterBox=orange, fit=(c1) (c2), inner sep=8pt] (boxC) {};

\draw[posEdge] (c1) -- (c2);

\node[nodeStyle,fill=violet!15] (d1) at (3.6,-1.9) {13};
\node[nodeStyle,fill=violet!15] (d2) at (4.6,-1.4) {14};
\node[nodeStyle,fill=violet!15] (d3) at (5.1,-2.3) {15};

\node[clusterBox=violet, fit=(d1) (d2) (d3), inner sep=8pt] (boxD) {};

\draw[posEdge] (d1) -- (d2) -- (d3) -- (d1);

\draw[posEdge] (a2) -- (b1);
\draw[posEdge] (a6) -- (b4);
\draw[posEdge] (a5) -- (c1);
\draw[posEdge] (b4) -- (d1);
\draw[posEdge] (c2) -- (d2);
\draw[posEdge] (a2) -- (d2);
\draw[posEdge] (b3) -- (d2);

\end{tikzpicture}

%% file: figures/cctikzv2_2.tex
\begin{tikzpicture}[
  scale=0.6,
  every node/.style={font=\tiny},
  nodeStyle/.style={
    circle, draw=black, thick, inner sep=1.8pt,
    minimum size=4mm, fill=white
  },
  posEdge/.style={thick},     %
  posEdgeBad/.style={draw=red!70!black, very thick}, %
  negEdge/.style={draw=red!70!black, dashed, very thick}, %
  clusterBox/.style={
    rounded corners=6mm, draw=#1!60, fill=#1!20,
    fill opacity=0.15, very thick
  },
  peel/.style={rounded corners=6mm, draw=green!60!black,
               fill=green!15, fill opacity=0.18, very thick},
  super/.style={circle, double, draw=black,
                minimum size=6mm, very thick},
]

\node[nodeStyle,fill=blue!15] (a1) at (1.3,1.4)  {1};
\node[nodeStyle,fill=blue!15] (a2) at (0.9,2.4){2};
\node[nodeStyle,fill=blue!15] (a3) at (0,2.7){3};
\node[nodeStyle,fill=blue!15] (a4) at (-0.9,1.9){4};
\node[nodeStyle,fill=blue!15] (a5) at (-0.6,0.9){5};
\node[nodeStyle,fill=blue!15] (a6) at (0.8,0.8) {6};

\node[clusterBox=blue, fit=(a1) (a2) (a3) (a4) (a5) (a6), inner sep=10pt] (boxA) {};

\draw[posEdge] (a1) -- (a2) -- (a3) -- (a4) -- (a5) -- (a6) -- (a1);
\draw[posEdge] (a2) -- (a4) -- (a6) -- (a2);
\draw[posEdge] (a3) -- (a6);
\draw[posEdge] (a3) -- (a5) -- (a1);

\node[nodeStyle,fill=green!15] (b1) at (4.0,2.5) {7};
\node[nodeStyle,fill=green!15] (b2) at (5.0,2.1) {8};
\node[nodeStyle,fill=green!15] (b3) at (4.6,1.3) {9};
\node[nodeStyle,fill=green!15] (b4) at (3.4,1.2) {10};

\node[clusterBox=green, fit=(b1) (b2) (b3) (b4), inner sep=10pt] (boxB) {};

\draw[posEdge] (b1) -- (b2) -- (b3) -- (b4) -- (b1);
\draw[posEdge] (b2) -- (b4);

\node[nodeStyle,fill=orange!15] (c1) at (-0.5,-2.0) {11};
\node[nodeStyle,fill=orange!15] (c2) at (0.6,-1.6)  {12};

\node[clusterBox=orange, fit=(c1) (c2), inner sep=8pt] (boxC) {};

\draw[posEdge] (c1) -- (c2);

\node[nodeStyle,fill=violet!15] (d1) at (3.6,-1.9) {13};
\node[nodeStyle,fill=violet!15] (d2) at (4.6,-1.4) {14};
\node[nodeStyle,fill=violet!15] (d3) at (5.1,-2.3) {15};

\node[clusterBox=violet, fit=(d1) (d2) (d3), inner sep=8pt] (boxD) {};

\draw[posEdge] (d1) -- (d2) -- (d3) -- (d1);

\draw[posEdgeBad] (a2) -- (b1);
\draw[posEdgeBad] (a6) -- (b4);
\draw[posEdgeBad] (a5) -- (c1);
\draw[posEdgeBad] (b4) -- (d1);
\draw[posEdgeBad] (c2) -- (d2);
\draw[posEdgeBad] (a2) -- (d2);
\draw[posEdgeBad] (b3) -- (d2);
\draw[negEdge] (b1) -- (b3);
\draw[negEdge] (a3) -- (a1) -- (a4);
\draw[negEdge] (a2) -- (a5);

\end{tikzpicture}

%% file: dynamic-framework.tex
\section{Dynamic Framework}
\label{sec:dynamic-framework}
In most part of the paper, we will work on a fixed graph $G=(V,E)$ represented by $(\calC, D)$, a pair of clustering and the symmetric difference between the edges $E$ and the edges of $\calC$. This graph will be receiving edge updates in the form of inserting new edges into $E$ or removing existing edges from $E$. 
For each vertex $v \in V$, let $d(v)$ be the degree of $v$, $N(v)$ be the set of neighbors of $v$ together with $v$ itself. 

Given a clustering $\calC$, let $C(v)$ be the cluster of $v$ in $\calC$. Let $\mathcal{E(C)}$ be the set of pairs $(u, v)$ such that $u$ and $v$ are in the same cluster in $\calC$. 
Let $\cost(\calC)$ be the cost for using this clustering, by definition $\cost(\calC) = |E \triangle \mathcal{E(C)}|$, 
where $A\triangle B$ denotes the symmetric difference between the sets $A$ and $B$. 
We will use these notations throughout the paper.

For the dynamic algorithm, we will follow the approach that for every certain number of updates we will recompute the full clustering. For the purpose of this, between the recomputations we are going to maintain the cluster representation of the current graph $G=(V,E)$:
\begin{itemize}
    \item $\calC$: A clustering of the graph $G$.
    \item 
    $D$: The edges from the symmetric difference $\mathcal{E(C)}\triangle E$, also referred to as the \textit{current violation}.
\end{itemize}
Then $|D|$ is the cost of the current clustering $\mathcal C$.
In addition, we shall assume the following trivial representation of the clustering $\calC$:
\begin{itemize}
    \item an identifier for each cluster.
    \item a doubly-linked list of the vertices in each cluster that can be accessed from the identifier.
    \item labeling each vertex with its cluster.
\end{itemize}
With this representation, we can easily, in constant time, ask if
two vertices are in the same cluster or move vertices between clusters. Also, for a given vertex $u$, we can also list the vertices in the same cluster in linear time.

One of the insights for our dynamic framework is that existing static $c$-approximate correlation clustering algorithms can be transformed to work
efficiently if we give them an arbitrary cluster
representation $(\mathcal C,D)$ as input instead of the classic graph representation $G=(V,E)$. 
With this input we will transform them so that they produce a new $c$-approximate correlation clustering
$\mathcal C'$ in $O(|D|)$ time. We will then use this algorithm to recompute our clustering before the updates has changed the cost of our clustering too much.

Since $\mathcal C'$ is then going to be the clustering we use as the initial clustering the next time we recompute, as long as we compute the clustering using an $O(1)$-approximate clustering for $\mathcal C'$ then $|D'|=|E\triangle \calC'|=O(\cost(\opt))$. This means that after $O(\epsilon|D|)$ updates the cost has at most changed by a multiplicative factor of $1+\eps$, while the recomputation uses $O(|D|)$ time. 
\begin{algorithm}[H]
	\caption{Dynamify($\mathcal{A},\eps$) : $\mathcal{A}$ takes a cluster presentation $(\mathcal C,D)$ and modifies it to a $c$-approximate cluster representation $(\mathcal C',D')$. %
    }
	\label{alg:dynamic}
	\begin{algorithmic}[1]
        \State $\calC \gets V$, $D \gets \emptyset$, $r \gets 0$, $\mu \gets \frac{\eps}{2(1+\eps)c}$
        \Procedure{flip}{$u,v$}\Comment{$(u,v)$ is the edge being updated.}
            \If{$(u,v)\in D$}
                \State $D \gets D\setminus \{(u,v)\}$
            \Else
                \State $D \gets D\cup \{(u,v)\}$
            \EndIf
            \State $r\gets r - 1$
            \If{$r \le 0$}
                \State $(\calC',D') \gets \mathcal{A}(\calC,D)$    
                \If{$|D'| \le |D|$} 
                    \State $(\calC,D) \gets (\calC',D')$
                \EndIf
                \State $r\gets \mu|D|$
            \EndIf
        \EndProcedure
	\end{algorithmic}
\end{algorithm}

It should be noted that we initially are going to present the algorithm as if the updates are computed in an amortized fashion. It is however not too problematic to deamortize the algorithm by computing the algorithm $\mathcal{A}$ in the background of the following $r$ updates. The details of this computation can be found in \cref{sec:deamortization}.

We are going to present this algorithm as being able to flip the occurrence of an edge. This is the method $\Call{flip}{}$ in \cref{alg:dynamic}. This is mostly to simplify the implementation since inserting and deleting edges are essentially the same operation, especially since we are working in a model where what matters is whether the edge currently is contained in the symmetric difference.

Regarding maintaining the set $D$ when performing the $\Call{flip}{}$ operations in \cref{alg:dynamic}, we have to be a bit careful. While this would be trivial using hashing, this would add an additional layer of complexity to handle the probability of the hashing not running in the expected time. Furthermore, in contrast to the clustering algorithm, we are not allowed to fail at this task for any update, since it could have adverse consequences for updates in the far future. 

For each edge, we are interested in whether or not it is present in the symmetric difference. The reason we have to be careful is that we only have $O\left(\eps^{-1}\right)$ time for each update, and we do not want this to increase the amount of space we have to use. We could do this trivially by using a 2d array of size $|V|^2$ for each possible edge. This would however increase the memory footprint of the whole algorithm to $O(|V|^2)$ from $O(|E|)$. We are instead going to perform this type of operation in an amortized way, by only calculating the symmetric difference $D$ right before we recompute with the algorithm $\mathcal{A}$. We keep an empty array of size $|V|$. Then for each vertex, we keep track of all updates that changes the neighborhood of this vertex, storing their occurrence each time they are made. Right before recomputing the clustering with $\mathcal{A}$, we then for each vertex $v$ incident to either an update or an edge in $D$ go through all the updates incident to $v$ and determine the final state of each of them. From this we can update all neighbors that could potentially be part of the symmetric difference, by flipping the corresponding entries in the empty array of size $|V|$ for each update and each entry in $D$, each in $O(1)$ time. This means that we can compute the symmetric difference after $r$ updates in time $O(r+|D|)$.

In the rest of this section, we are going to assume that the clustering algorithm $\mathcal{A}$ is deterministic. This is to simplify the description of our algorithm, since adding the randomness adds an additional layer of complications. We will later expand on this.

\begin{theorem}[Formal version of \cref{thm:dynamic-intro}]
\label{thm:dynamic}

    Let $\mu \le \frac{\eps}{2(1+\eps)c}$. Let $\mathcal{A}$ be a static correlation clustering algorithm that as input takes as a cluster representation $(\calC,D)$ and in $O(t|D|)$ time
    produces a $c$-approximate cluster representation $(\mathcal{C'},D')$.
    Then $\mathrm{Dynamify}(\mathcal{A},\eps)$ from \cref{alg:dynamic}, is a fully-dynamic algorithm on $G$ that in
    $O\left(\mu^{-1}t\right)$ worst-case time per edge update 
    maintains a $((1+\eps)c)$-approximate clustering.
\end{theorem}

To show this, we primarily need to make the observation that by not updating anything for $\mu|D|$ updates, we only increase the approximation ratio by a multiplicative ratio of $(1+\eps)$. This enables us to batch the updates and perform them only when the cost will have differed significantly.
\begin{lemma}\label{lem:dynamic-approx-small-change}
    Let $\calC$ be a $c$-approximation of the optimal clustering of the graph $G = (V,E)$ with the symmetric difference $\mathcal{E(C)} \triangle E$, and let $\mu \le \frac{\eps}{2(1+\eps)c}$. Let $G'=(V',E')$ be a graph such that $|E\triangle E'|\le \mu|\mathcal{E(C)} \triangle E|$. Then $\calC$ is an $((1+\eps)c)$-approximation for the graph $G'$.
\end{lemma}
\begin{proof}
    Let $\calC^*$ be an optimal clustering for $G$ and let $\calC'^*$ be an optimal clustering for $G'$. Then we have that $|\mathcal{E(C)} \triangle E|\le c|\mathcal{E}(\calC^* )\triangle E|$. Furthermore, $\calC^*$ gives a lower bound of the optimal solution in $G'$ as $|\mathcal{E}(\calC^*)\triangle E| - |E \triangle E'| \le |\mathcal{E}(\calC'^*)\triangle E'|$. Combining we have that
    \begin{align*}
        |E \triangle E'| &\le \mu|\mathcal{E(C)}\triangle E|\\
        &\le \frac{\eps}{2(1+\eps)}|\mathcal{E}(\calC^*) \triangle E|\\
        &\le \frac{\eps}{2(1+\eps)}(|\mathcal{E}(\calC^*) \triangle E'| + |E \triangle E'|)\\
        &\le \frac{\eps}{2(1+\eps)}(|\mathcal{E}(\calC'^*) \triangle E'| + 2|E \triangle E'|).
    \end{align*}
    Isolating for $|E \triangle E'|$, we get $|E \triangle E'| \le \frac{\eps}{2}|\mathcal{E}(\calC'^*) \triangle E'|$. 
    Using this inequality we get
    \begin{align*}
        |\mathcal{E}(\calC) \triangle E'| &\le |\mathcal{E}(\calC) \triangle E| + |E \triangle E'|\\
        &\le c|\mathcal{E}(\calC^*) \triangle E| + |E \triangle E'|\\ 
        &\le c|\mathcal{E}(\calC'^*) \triangle E'| + 2c|E \triangle E'|\\
        &\le(1+\eps)c|\mathcal{E}(\calC'^*) \triangle E'|.
    \end{align*}
\end{proof}

\begin{proof}[Proof of \cref{thm:dynamic}]
    By the previous discussion, we can maintain the symmetric difference in an amortized fashion just before it is handed over to $\mathcal{A}$.

    For the running time, consider the time at which we recompute the clustering. Let $\hat D$ be the violation that was computed the previous time that the cluster was recomputed and $D$ the true violation just before a recomputation. Then it holds that $|D| \le (1+\mu)|\hat D|$ from the fact that each update can increase the violation by at most $1$. Thus, we observe that whenever we run $\mathcal{A}(\mathcal C, D)$, we have it is at least $\frac{\mu}{1+\mu}|D|$ updates since we last recomputed with $\mathcal A$. As $\mathcal A$ has running time $O(t|D|)$ and $\mu \le 1$ by definition, the amortized running time is $O(\mu^{-1} t)$ per update.

    Finally, to show the approximation ratio, for any update let $E$ be the edge set at the previous time when the cluster representation $(\calC,D)$ was recomputed, and $E'$ the current edge set. Then, since each update at most changes the set of edges by $1$ element and that it is at most $\mu|D|$ updates since the last recomputation, we have $|E\triangle E'|\le \mu|D|$. As $\calC$ is $c$-approximate, by \cref{lem:dynamic-approx-small-change} $\mathrm{Dynamify}(\mathcal{A},\eps)$ maintains an $((1+\eps)c)$-approximate clustering.
\end{proof}

Finally, we want to mention that it is possible to maintain the cost of the solution at all times doing the computation, as long as the user is well-behaved, that is the user never attempts to delete an edge that does not exist nor insert an edge that already exists. To achieve this, we are further going to assume that the algorithm is told for each update/flip whether the update is an edge insertion or an edge deletion. In this case, you can simply check whether the two vertices being updated currently are in the same cluster, and then depending on this and whether the update is an insertion or a deletion directly compute the new cost.

\subsection{Dynamic blow-up of error probabilities}

As we mentioned early, the error probability of a Monte Carlo algorithm can blow up logarithmaticly through our dynamic framework. 
We will show this effect with a concrete example in this subsection. 
To show how this blow-up can behave, we are going to work with a specific hypothetical approximation algorithm $\mathcal A$, which on cluster representation $(\calC,D)$ with probability $1-p$ outputs the optimal clustering and with probability $p$ does nothing and outputs its own input $(\calC,D)$ for a constant $p$. We note that this is not only a problem against an adaptive adversary, but that since the following construction is fully deterministic and specified ahead of time, an oblivious adversary could also implement the same updates.

We imagine that the initial state of the dynamic algorithm is that $G$ consists of $n/3$ disjoint $2$-paths. 
Over the next $2n/3$ updates, we remove both edges in each of the $2$-paths in the following updates one by one. 
We are going to look at the probability that the clustering is not updated after the $(2n/3-1)$'th update, since this is the first time the optimal clustering hits cost $0$. 
Failure to recompute at this update would imply no approximation, and therefore that our algorithm has failed badly.

The issue now arises in the case when we at any point overestimate the actual cost by more than a factor $2\mu^{-1}$. 
If this happens, $r$ is going to be assigned a value larger than the number of updates before we hit cost $0$, implying that we are going to fail badly at the $(2n/3-1)$'th update. 
This is guaranteed to happen if $k=\mu^{-1}\ln (2\mu^{-1}) = O(1)$ recomputations in a row fail to compute an good approximation. 
If we were to look at the probability of encountering $k$ failures in a row, then this is given by $p^k$.

The total number of recomputations is at least $\log n/\log(\mu^{-1}) = \Omega(\log n)$. So we also have at least $r=\Omega(\log n)$ opportunities to achieve $k$ fails in a row. 
Thus the probability of \emph{never} having $k$ fails in a row is therefore bounded by
$(1-p^k)^r < \exp(-rp^k) = \frac{1}{n^{\Omega(p^k)}}$.
Since $p$ is a constant, the probability for the algorithm to succeed is polynomially small. 

\subsection{Bounding the dynamic error probabilities}
\paragraph*{The pure cases} We will now prove \cref{thm:dynamic-random-intro}, analyzing how the failure probability of the static algorithm affects the failure probability of the dynamic algorithm.

We are going to consider any update $i^*$, and we want to bound the probability that our current clustering $\mathcal K$ is not $(1+\eps)c$-approximate when we get to this update. In particular, this requires that the last rebuilt cluster representation $(\mathcal C',D')$ was not $c$-approximate.

Taking a step back, suppose we want to do a rebuild before some update $i$ and the input for this rebuild is the cluster representation $(\mathcal C,D)$. We then get a new
cluster representation $(\mathcal C',D')$ with $|D'|\leq |D|$ which
we will keep for $\mu |D'|$ updates. We say that a rebuild happening at update $i$ with input $(\mathcal{C},D)$ is \emph{risky with respect to} $i^*$ if $i^*-i\leq \mu|D|$, for if we do not get
a $|D'|<|D|$, then we will not get
a better cluster representation before update $i^*$. Conversely, if it is not risky, then we know we will get another rebuild before we get to update $j$. Thus, if our clustering is not $(1+\eps)c$-approximate at update $j$, then some risky rebuild must have failed. We are going to upper bound the probability that any risky rebuild fails. To do this we are going to bound the number of risky rebuilds that can affect a specific update.

\begin{lemma}\label{lem:risky-limited-layers}
    Let $\mu < 1/6$. Suppose we want to do a rebuild at update $i$ with input representation $(\mathcal C,D)$ and that 
    the rebuild is risky with respect to $i^*$. Then we will perform at most $3$ rebuilds with input violations above $|D|/2$ between update $i$ and $i^*$.
\end{lemma}
\begin{proof}
    If update $i$ is risky then $i^*-i\le \mu|D|$. Suppose the rebuild performs a rebuild that produces the cluster representation $(\mathcal{C}',D')$ such that $|D'|\ge |D|/3$. Then the number of updates until the next rebuild would be at least $\mu|D|/3$. We can therefore have at most $3$ such rebuilds we pass before we  pass update $i^*$ and this includes a risky rebuild at update $i$. Suppose on the other hand that we from some update $j\geq i$ do a rebuild leading to an  output violation $D'$ is of size less than $|D|/3$. We have at most $\mu|D| < |D|/6$ updates between $j\geq i$ and $i^*$ due to the requirement of $\mu$. Since every update can at most increase the violation by $1$, we conclude that between update $j$ and $i^*$, the violation is always of size below $|D|/2$ . We conclude there can at most exists $3$ rebuilds between update $i$ and update $i^*$ with an input violation greater than $|D|/2$. 
\end{proof}

It should be noted that the statement of \cref{lem:risky-limited-layers} is completely independent of how each rebuild is done or whether it
succeeds or not. It only uses that the output violation is never bigger than the input violation. Since we have an upper bound of $m$ for the maximum number of violations, we have at most $O(\log m) = O(\log n)$ risky rebuilds, and so if for any risky rebuild the probability that that rebuild fails is 
bounded by $p$, then $O(p\log n)$
bounds the probability that any risky rebuild fails, hence the probability that we do not have a $(1+\eps)c$-approximate clustering when we get to update $j$. This completes the proof of Theorem \ref{thm:dynamic-random-intro}(a).

We now want to understand what happens if the failure probability with input cluster representation $(\mathcal C,D)$ is $q(|D|)$ for some non-increasing function $q$.
Let $Q(d)$ be the maximal probability that some risky 
rebuild fails, starting from some risky rebuild with input violation at most $d$. Clearly this is an increasing function with $Q(d)=0$ if $d<1$. 
Stating from a risky rebuild with input size $d$, it follows from
\cref{lem:risky-limited-layers}, that we can have at most $3$ risky rebuilds before we get one with input violation size at most $d/2$. Therefore we have 
the recurrence $Q(d)\leq 3 q( d/2)+Q( d/2)$.
If $q$ is a function falling inversely in $d$, then the sum is geometrically increasing, and then we conclude that $Q(d)=O(q(1))$.
This completes the proof of 
Theorem \ref{thm:dynamic-random-intro}(b).

\paragraph*{The mixed case} We will now prove \cref{thm:dynamic-mixed-case}.  In this case we are working with two static algorithms.
\begin{itemize}
    \item $\mathcal A$ that is $c$-approximate with probability $p$ and
    \item $\hat{\mathcal{A}}$ that is $\hat c$-approximate with probability $q(d)$ where $q$ falls inversely with the input violation size $d$.
\end{itemize}
Every time we do a rebuild, we apply first $\hat{\mathcal A}$, followed by $\mathcal A$. We are furthermore always going to keep the best solution after each application of both algorithms, so the combined algorithm  $ \mathcal A \circ \hat{\mathcal A}$ gives the best of both worlds (in fact, $\mathcal A\circ \hat{\mathcal A}$ fails to be $\hat c$-approximate with probability at most $p\cdot q(d)$, but we will not exploit that).

We are again going to fix an update $i^*$ and we want to show that the dynamic algorithm is $(1+\eps)c$-approximate at $i^*$ with probability
$O(p+q(1))$.
We define $\mu$ small enough that if we have a $c$-approximate clustering with violation size $d$, then it remains
$(1+\eps)c$-approximate for
$\mu d$ updates. We will also 
define $\mu$ smaller than $1/6$ and $1/(2\hat{c})$. Picking the minimum of these three options is still a constant.
As in the proof for \cref{thm:dynamic-random-intro}, we say that a rebuild at update $i$ is at risk with respect to update $i^*$ if we have input cluster representation 
$(\mathcal C,D)$ and $i^*-i\leq\mu|D|$. 
We then apply $\hat{\mathcal A}$
producing a cluster representation
$(\hat{\mathcal C},\hat{D})$ that will be used as input for $\mathcal A$. 
Next, we only say $\mathcal A$ is at risk with respect to $i^*$ if 
$\hat{\mathcal C}$ is $\hat{c}$-approximate and $i^*-i\le \mu|D^*|$.
We then apply $\hat{\mathcal A}$ to $(\hat{\mathcal C},\hat{D})$ producing
the final output $(\mathcal C',D')$ of the rebuild with $\mathcal A\circ \hat{\mathcal A}$.

With the above definitions, suppose we get to update $i^*$ with a clustering that is not $(1+\eps)c$-approximate. Let the last rebuild be at update $i$. Then $\mathcal A\circ \hat{\mathcal A}$ must have failed producing a $c$-approximate clustering $\mathcal C'$. Moreover, the whole rebuild must have been at risk, for otherwise, there would be another rebuild before update $i^*$. 

If $\hat{\mathcal A}$ succeeded, then $\mathcal A$ must still have been at risk, for otherwise we again conclude that there would be another rebuild before update $j$.
Finally, $\mathcal A$ must
have failed producing a $c$-approximate clustering $\mathcal C'$.

Recall that $\hat{\mathcal A}$ fails producing a $\hat{c}$-approximate clustering with probability $q(d)$ where $q$ falls inversely with the input violation size $d$. As in our
previous analysis, we would like to conclude that the probability that we ever fail a risky application of $\hat{\mathcal A}$ which fails is $O(q(1))$. However, we need to 
revisit the argument to check that we are not cheating
because the full rebuild also
applies $\mathcal A$.

Consider a risky rebuild at update $i$. With input violations size $d$, risky means that
$i^*-i\leq \mu d$. We will now apply 
Lemma \ref{lem:risky-limited-layers} which doesn't care how the rebuild is done, so it also applies when now rebuild with $\mathcal A\circ \hat{\mathcal A}$. The lemma states 
that we can perform at most $3$ rebuilds with input violations above $d/2$ between update $i$ and $i^*$.

As before, we now define
$Q(d)$ as the maximal probability that $\mathcal A^*$
ever fails on on a risky rebuild starting from input violation size at most $d$. Then $Q(d)\leq 3 q(d/2)+Q(d/2)=O(q(1)).$
We will now assume that $\mathcal{\hat  A}$ never fails a risky rebuild. The last rebuild
must then have included a risky application of $\mathcal A$ which failed.  We will bound the probability of this event by $O(p)$. For this we use the following lemma.

\begin{lemma}\label{lem:risky-with-approx}
    Let $\mu \le 1/(2\hat{c})$. Suppose we do a rebuild at update $i$ with a $\hat c$-approximate input representation $(\hat{\mathcal C},\hat{D})$ and that the rebuild is risky with respect to $i^*$. Then we will perform at most $2\hat c$ rebuilds between update $i$ and $i^*$.
\end{lemma}
\begin{proof}
    Let $(\mathcal C_i^*,D_i^*)$ be an optimal clustering at update $i$. Each update can decrease the violation by at most one, so 
    $|D_{i^*}^*| \ge |D_i^*|-(i^*-i) \ge |\hat{D}|/\hat{c}-\mu|\hat{D}| \ge |\hat{D}|/(2\hat{c})$,
    where the last inequality used that $\mu \le 1/(2\hat{c})$. For every rebuild between update $i$ and update $i^*$ producing an output representation $(\mathcal{C}',D')$, we have $|D'|\ge |D_i^*|\geq |\hat{D}|/(2\hat{c})$, and so we have at least $\mu|\hat{D}|/(2\hat{c})$ updates between any two rebuilds between update $i$ and $i^*$. Thus the total number of rebuilds can be at most 
    $\frac{(i^*-i)}{\mu |\hat{D}|/(2\hat{c})}\leq \frac{\mu |\hat{D}|}{\mu |\hat{D}|/(2\hat{c})}=2\hat{c}$.
\end{proof}

Using \cref{lem:risky-with-approx}, we conclude that from the first risky rebuild with $\mathcal A$, there can be at $2\hat{c}$ times rebuilds, hence the probability that $\mathcal A$ makes a risky failure with respect to $i^*$ is bounded by $2\hat{c} p=O(p)$. 

Summing up, we conclude that the probability that $\hat{\mathcal A}$ or $\mathcal A$ ever fail a risky update with respect to $i^*$ is bounded by $O(p+q(1))$.
This also bounds the probability that our dynamic clustering is not $(1+\eps)c$-approximate when we get to update $i^*$, completing the proof of \cref{thm:dynamic-mixed-case}.

\input{deamortisation}

%% file: deamortisation.tex
\subsection{Deamortization}
\label{sec:deamortization}
In this section, we are going to explore how to de-amortize the algorithm so that each update runs in worst case $O(\mu^{-1}t)$ instead of just amortized $O(\mu^{-1}t)$.  For this, we are mostly going to be using standard techniques. The general idea is that whenever we determine that we are going to recompute the clustering, we start doing this by making $O(\mu^{-1}t)$ work for each update on the recomputation.

In the process of describing the amortization, we describe the algorithm as running over a number of epochs. Each epoch consists of a recomputation of the clustering, and the cleaning and comparison of our current available two clusterings and swapping in the best of the two.

The deamortization is essentially going to work by computing $O\left(\mu^{-1}t\right)$ steps of the static algorithm for each update. This means that we are going to have a pipeline-like structure, where whenever an update arrives in some epoch $i$, we are only going to start computing the clustering using it in the next epoch $i+1$. Then only from the epoch after that one, $i+2$, are we going to use the result of the computations to answer queries. This delay will at most result in a multiplicative error of $(1+O(\eps))$ to the approximation factor for each epoch we delay between receiving the update to an edge and the time in which it is used to answer the received queries.

We have to be a bit careful with this deamortization approach. The reason is that we need to swap out the old clustering for a new one in $O\left(\mu^{-1}t\right)$ time in the worst case. This is non-trivial, since the cluster representation $(\mathcal{C},D)$ has size $\Omega(n)$ which might be much larger than $|D|$, in which case we do not have the time to make a copy and swap it in. So we have to make sure that we do not copy more than what is necessary. 

\paragraph*{Swapping in a new solution} To swap in a new solution, for each vertex we maintain two entries, the possible ids of clusters that a particular vertex can belong to. These correspond to the clustering the vertex belongs to in the old clustering and the one in the new clustering. We furthermore have a global flag that indicates 
whether to use the old clustering or the new one. The way we use this information is that if two options exist for the cluster a vertex can belong to, then the global flag decides. Otherwise, the vertex belongs to the only cluster that is written down. This means that we can swap in the new solution by flipping the global flag.

After we have swapped in the new solution, we spend time proportional to the number of updates for the old computation to clean up the vertices that have two cluster ids written down. After that, the new solution has become the old one, so we finally flip the global flag to say use the old solution if it exists, to be ready for the next round of computations.

%% file: dynamicPivot.tex
\section{Dynamic Pivot}
\label{sec:pivot}

In this section, we will show how to implement the Pivot algorithm \cite{ACN08} in $O(|D|)$ time in cluster representation $(\calC, D)$, as a warm-up for our full algorithm. Before the detail, we first show some helpful properties of the cluster representation (See \cref{fig:pivot-preprocess}). 

\begin{definition}[Active]
    A vertex $v$ is \emph{active} if it is incident to some edges in $D$, otherwise it is \emph{inactive}. 
    A cluster in $\calC$ is \emph{active} if it contains an active vertex, otherwise it is \emph{inactive}. 
\end{definition}

\begin{definition}[Core]
    For any active cluster $C$, the set of inactive vertices in $C$ is called the \emph{core} of $C$, denoted by $\core(C)$. 
\end{definition}

In the next lemma, we will see that all inactive vertices can be treated in groups, and the only challenging part is on the active vertices. 

\begin{lemma}
    For any graph represented by $(\calC, D)$, the following two statements hold. 
    \begin{enumerate}
        \item Every inactive cluster $C$ is a perfect clique in the graph, with no exterior edges incident to its vertices. 
        \item For any active cluster $C$, all the vertices in $\core(C)$ have the same neighborhood, which equals to $C$. 
    \end{enumerate}
\end{lemma}
The proof of the lemma is straightforward. 
Therefore, any locally good clustering has to include all the inactive clusters. 
Moreover, it has to put all vertices in a core together and never put two cores in the same cluster. From this fact, we can contract each core into a weighted vertex. 
When creating a new clustering, we keep the label of the clusters to each core the same, and only move active vertices around. 
After contraction, we only need to deal with a graph of $O(|D|)$ vertices, which are the cores and active vertices. 

\begin{figure}
\centering
\begin{subfigure}{0.47\textwidth}
    \begin{center}
        \input{figures/framework-tikz}
    \end{center}
    \caption{Input representation $(\mathcal{C},D)$ of graph $G$.}
    \label{fig:pivot-preprocess-a}
\end{subfigure}
\hfill
\begin{subfigure}{0.47\textwidth}
    \begin{center}
    \input{figures/framework-tikz-2}
    \end{center}
    \caption{$V'$ after removing inactive vertices and adding virtual vertices}
    \label{fig:pivot-preprocess-b}
\end{subfigure}
\caption{Outline of preprocessing done for \textsc{Pivot}($\calC,D$).}
\label{fig:pivot-preprocess}
\end{figure}

Recall that the Pivot algorithm proceeds by selecting a vertex $v$ uniformly at random, then creates a new cluster $N(v)$ and recurses on the remaining graph. 
We will simulate the Pivot algorithm directly on the contracted graph in $O(|D|)$ time. 
The pseudocode is presented in \cref{alg:pivot}. 

\begin{algorithm}
    \caption{Pivot($\calC, D$)}
	\label{alg:pivot}
	\begin{algorithmic}[1]
        \State Let $V'$ be the set of active vertices and $\mathcal{S}$ be the set of active clusters. 
        \State For each active cluster $C$, add a virtual vertex $c$ for $\core(C)$ to $V'$. 
        \State For each vertex $v \in V'$, set its weight $w_v$ to be the number of vertices it represents. 
        \State For each active cluster $C$, construct a list $L(C)$ of all the vertices in $V'$ from $C$. 
        \State For each vertex $v \in V'$, let $N^+(v)$ be the list of neighbors of $v$ out of $C(v)$, and $N^-(v)$ be the list of non-neighbors of $v$ in $C(v)$. %
        \State $\calC' \gets \calC \setminus \mathcal{S}$
        \Comment{We can do this in-place in $O(|\mathcal{S}|)$ time}
        \While{$V' \neq \emptyset$} 
            \label{lst:pivot-while}
            \State Pick a vertex $v \in V'$ with probability proportional to $w_v$
            \label{lst:pivot-weighted-sampling}
            \State $T \gets (N^+(v) \cap V') \cup (L(C(v)) \setminus N^-(v))$
            \label{lst:pivot-cluster-create}
            \State $\calC' \gets \calC' \cup \{T\}$, $V' \gets V' \setminus T$, $L(C(v)) \gets L(C(v)) \cap N^-(v)$
            \label{lst:pivot-while-end}
        \EndWhile
        \State Unpack $\calC'$ and return the corresponding clustering for the original graph
	\end{algorithmic}
\end{algorithm}

\begin{theorem}
    Pivot($\calC, D$) from \cref{alg:pivot} runs in $O(|D|)$ time and has the same output distribution as Pivot($V, E$) from \cite{ACN08} for the same graph. 
\end{theorem}

\begin{proof}
    We will first show the correctness of this algorithm. 
    Each group of vertices contracted to the same one, are neighbors to each other and have the same outgoing neighbors, so they will always be put in the same cluster. 
    Therefore, we can treat them as a single one with a higher probability to be sampled. 
    Note that all clusters in $\calC \setminus \mathcal{S}$ are cliques without any other edges, so the Pivot algorithm will always put them together. 
    When we start line 16, $N^+(v)$ is the set of neighbors of $v$ out of $C(v)$, $N^-(v)$ is the set of non-neighbors of $v$ in $C(v)$. 
    While running the loop from line~\ref{lst:pivot-while} to line~\ref{lst:pivot-while-end}, $L(C)$ keeps track of the remaining vertices in $V'$ associated with cluster $C$. 
    So $T$ is the set of remaining neighbors of $v$ in $V'$, therefore it has the same output distribution as the Pivot algorithm in \cite{ACN08}. 

    Now we are going to analyze the running time of this algorithm. 
    We have $|\mathcal{S}| \le 2|D|$ and $|V'| \le 2|D| + |\mathcal{S}| \le 4|D|$, so it only takes $O(|D|)$ time before the start of line~\ref{lst:pivot-while}. 
    By \cref{lem:weighted-sampling}, line~\ref{lst:pivot-weighted-sampling} can be implemented in total time $O(|V'|)$ by weighted sampling without replacement. 
    Finally, note that $T$ in line~\ref{lst:pivot-cluster-create} can be computed in $O(|T| + |N^+(v)| + |N^-(v)|)$ time, and the total size of $T$ is at most $|V'|$, and the total size of $N^+(v) \cup N^-(v)$ is at most $2|D|$, so the running time of the full algorithm is $O(|D|)$. 
\end{proof}
The next lemma shows that line 8 in total can be implemented in $O(|V'|)$ time. 
\begin{lemma}[\cite{fischer2025faster} Lemma 28]
    \label{lem:weighted-sampling}
    Let $A$ be a set of initially $n$ objects with associated integer weights $\{w_a\}_{a \in A}$. 
    There is a data structure supporting the following operations on $A$:
    \begin{itemize}
        \item SAMPLE(): Samples and removes an element $a \in A$, where $a \in A$ is selected with probability $w_a / \sum_{a' \in A}w_{a'}$. 
        \item REMOVE($a$): Removes $a$ from $A$. 
    \end{itemize}
    The total time to initialize the data structure and to run the previous operations until $A$ is empty is bounded by $O(n)$, with high probability $1-\frac{1}{n^c}$, for any constant $c > 0$. 
\end{lemma}

It remains to show how to compute the new symmetric difference $D'$ with the new clustering $C'$. 
Here we use the fact that the new clustering $\calC'$ is constructed in-place from $\calC$, therefore, we get a log of modifications from $\calC$ to $\calC'$ for free. 
Let $L$ denote the set of vertices that moved to different clusters from $\calC$ to $\calC'$. 
We will compute $D' = E \triangle \mathcal{E(C')} = \mathcal{E(C)} \triangle \mathcal{E(C')} \triangle D$. 
Since our objective is to find a clustering of small symmetric difference, we can stop and report the original clustering if $\calC'$ is worse than $\calC$. 
The pseudocode is given in \cref{alg:symmetric-difference}. 

\begin{algorithm}
    \caption{SymmetricDifference($\calC, D, \calC', L$)}
	\label{alg:symmetric-difference}
	\begin{algorithmic}[1]
        \State Compute $D' \gets D \cap (\mathcal{E(C)} \triangle \mathcal{E(C')} \triangle D)$ by a scan of $D$
        \For{each $v \in L$}
            \State Let $C(v)$ and $C'(v)$ be the cluster of $v$ in $\calC$ and $\calC'$ respectively. 
            \State Compute $X = C(v) \triangle C'(v)$ and $Y = C(v) \cap C'(v)$ in $O(|C(v)| + |C'(v)|)$ time \label{lst:enumeration-symmetric-difference}
            \For{each $x \in X$ and $y \in Y$}
            \State Insert $(x, y)$ to $D'$ if $(x, y) \notin D$
            \EndFor
            \State $L \gets L \setminus Y$
        \EndFor
        \State \Return $(\calC', D')$
	\end{algorithmic}
\end{algorithm}

\begin{lemma}
    \cref{alg:symmetric-difference} outputs $D' = E \triangle \mathcal{E(C')}$ in $O(|D| + |D'|)$ time. 
\end{lemma}

\begin{proof}
    The correctness of the algorithm is straightforward, and we only focus on its running time. 
    In line~\ref{lst:enumeration-symmetric-difference}, we know $v \in C(v) \cap C'(v)$, and $C(v) \neq C'(v)$ as $v$ is moved between clusters. 
    So we have $|X| > 0$ and $|Y| > 0$ and $|X| + |Y| = |C(v)| + |C'(v)|$. 
    Therefore, we will detect $|X| \cdot |Y| = \Omega(|C(v)| + |C'(v)|)$ pairs $(x, y) \in \mathcal{E(C)} \triangle \mathcal{E(C')}$ in the next for-loop. 
    Since every pair in $\mathcal{E(C)} \triangle \mathcal{E(C')}$ will be checked at most twice, the total running time is $O(|\mathcal{E(C)} \triangle \mathcal{E(C')}| + |D|) = O(|D' \triangle D| + |D|) = O(|D| + |D'|)$. 
\end{proof}

We can directly translate the algorithm into another with running time $O(|D|)$ that output the best of the $(\calC, D)$ and $(\calC', D)$. 
\begin{corollary}
    Given a graph $(\calC, D)$, a clustering $\calC'$ and the list $L$ of moved vertices, we can output the best of the two clustering $(\calC, D)$ and $(\calC', D')$ in $O(|D|)$ time. 
\end{corollary}
\begin{proof}
    We can simulate \cref{alg:symmetric-difference} for at most $c|D|$ steps, where $c$ is a sufficiently large constant determined by \cref{alg:symmetric-difference}. 
    If it terminates, we return the best of $(\calC, D)$ and $(\calC', D')$; 
    if it does not stop, we are sure that $|D'| > |D|$, and we can return $(\calC, D)$. 
\end{proof}

As a remark, by \cref{thm:dynamic-random-intro} (a), if at each rebuild we repeat \cref{alg:pivot} for $O(\log\log n)$ time and take the best clustering, we can turn Pivot$(\calC, D)$ into a $(3+\eps)$-approximate dynamic correlation clustering algorithm with $O(\log\log n)$ update time and error probability of $O(1/\poly(\log n))$. 

%% file: figures/framework-tikz.tex
\begin{tikzpicture}[
  scale=0.6,
  every node/.style={font=\tiny},
  nodeStyle/.style={
    circle, draw=black, thick, inner sep=1.8pt,
    minimum size=4mm, fill=white
  },
  posEdge/.style={line width=0.5pt, draw=black!45},     %
  posEdgeBad/.style={line width=0.5pt, draw=black!70, very thick},     %
  negEdge/.style={draw=red!70!black, dashed, very thick},%
  clusterBox/.style={
    rounded corners=6mm, draw=#1!60, fill=#1!20,
    fill opacity=0.15, very thick
  },
  peel/.style={rounded corners=6mm, draw=green!60!black,
               fill=green!15, fill opacity=0.18, very thick},
  super/.style={circle, double, draw=black,
                minimum size=6mm, very thick},
]

\node[nodeStyle,fill=blue!15] (a1) at (1,1.6)  {1};
\node[nodeStyle,fill=blue!15] (a2) at (0.9,2.4){2};
\node[nodeStyle,fill=blue!15] (a3) at (0,2.7){3};
\node[nodeStyle,fill=blue!15] (a4) at (-0.9,1.9){4};

\node[clusterBox=blue, fit=(a1) (a2) (a3) (a4), inner sep=7pt] (boxA) {};

\draw[posEdge] (a1) -- (a2) -- (a3) -- (a4) -- (a1) -- (a3);
\draw[posEdge] (a2) -- (a4);

\node[nodeStyle,fill=green!15] (b1) at (4.0,2.5) {5};
\node[nodeStyle,fill=green!15] (b2) at (5.0,2.1) {6};
\node[nodeStyle,fill=green!15] (b3) at (4.6,1.3) {7};
\node[nodeStyle,fill=green!15] (b4) at (3.4,1.2) {8};

\node[clusterBox=green, fit=(b1) (b2) (b3) (b4), inner sep=10pt] (boxB) {};

\draw[posEdge] (b1) -- (b2) -- (b3) -- (b4) -- (b1) -- (b3);
\draw[posEdge] (b2) -- (b4);

\node[nodeStyle,fill=violet!15] (d1) at (3.6,-1.9) {9};
\node[nodeStyle,fill=violet!15] (d2) at (4.6,-1.4) {10};
\node[nodeStyle,fill=violet!15] (d6) at (5.2,-3.2) {14};

\node[nodeStyle,fill=violet!15] (d3) at (5.3,-2.1) {11};
\node[nodeStyle,fill=violet!15] (d4) at (3.6,-2.9) {12};
\node[nodeStyle,fill=violet!15] (d5) at (4.5,-3.8) {13};
\node[clusterBox=violet, fit=(d1) (d2) (d3) (d4) (d5) (d6), inner sep=8pt] (boxD) {};

\draw[posEdge] (d1) -- (d2);
\draw[negEdge] (d1) -- (d6);
\draw[posEdge] (d2) -- (d6);

\draw[posEdge] (d1) -- (d3);
\draw[posEdge] (d1) -- (d4);
\draw[posEdge] (d1) -- (d5);
\draw[posEdge] (d2) -- (d3);
\draw[posEdge] (d2) -- (d4);
\draw[posEdge] (d2) -- (d5);
\draw[posEdge] (d3) -- (d4);
\draw[posEdge] (d3) -- (d5);
\draw[posEdge] (d3) -- (d6);
\draw[posEdge] (d4) -- (d5);
\draw[posEdge] (d4) -- (d6);
\draw[posEdge] (d5) -- (d6);

\draw[posEdgeBad] (b3) -- (d2);

\end{tikzpicture}

%% file: figures/framework-tikz-2.tex
\begin{tikzpicture}[
  scale=0.6,
  every node/.style={font=\tiny},
  nodeStyle/.style={
    circle, draw=black, thick, inner sep=1.8pt,
    minimum size=4mm, fill=white
  },
  posEdge/.style={line width=0.5pt, draw=black!45},     %
  posEdgeBad/.style={line width=0.5pt, draw=black!70, very thick},     %
  negEdge/.style={draw=red!70!black, dashed, very thick},%
  clusterBox/.style={
    rounded corners=6mm, draw=#1!60, fill=#1!20,
    fill opacity=0.15, very thick
  },
  peel/.style={rounded corners=6mm, draw=green!60!black,
               fill=green!15, fill opacity=0.18, very thick},
  super/.style={circle, double, draw=black,
                minimum size=6mm, very thick},
]

\node[super,fill=green!15] (bs) at (4.0,2.5) {};

\node[nodeStyle,fill=green!15] (b3) at (4.6,1.3) {7};

\node[clusterBox=green, fit=(b3) (bs), inner sep=7pt] (boxB) {};

\draw[posEdge] (bs) -- (b3);

\node[nodeStyle,fill=violet!15] (d1) at (3.6,-1.9) {9};
\node[nodeStyle,fill=violet!15] (d2) at (4.6,-1.4) {10};
\node[nodeStyle,fill=violet!15] (d6) at (5.2,-3.2) {14};

\node[super,fill=violet!15] (ds) at (3.8,-3.1) {};
\node[clusterBox=violet, fit=(d1) (d2) (d6) (ds), inner sep=8pt] (boxD) {};

\draw[posEdge] (d1) -- (d2);
\draw[negEdge] (d1) -- (d6);
\draw[posEdge] (d2) -- (d6);

\draw[posEdge] (ds) -- (d1);
\draw[posEdge] (ds) -- (d2);
\draw[posEdge] (ds) -- (d6);

\draw[posEdgeBad] (b3) -- (d2);

\end{tikzpicture}

%% file: preclustering.tex
\section{Preclustering}
\label{sec:preclustering}

A lot of the recent progress in correlation clustering \cite{CLLN23+,cao2024understanding,CLMTYZ24,cao2025fastLP} has made critical use of a technique known as preclustering. 
The goal of this section is to make the existing preclustering results applicable in the format of cluster representation. 
Though the techniques used in this section are essentially the same as previous works, we include all details here to be self-contained. 

Preclustering was first introduced in \cite{CLLN23+}, where they showed that the algorithm introduced in \cite{CLMNP21} has some nice properties that can be used to preprocess the input (see also~\cite{DBLP:conf/innovations/Assadi022}). At a very high-level view, the purpose of the preclustering is to make all the ``easy'' decisions, while explicitly deferring ``hard'' decisions to later. With ``easy'' decisions, we consider decisions where the optimal solution has to pay very little to none, while the ``hard'' decisions are the decisions where we can guarantee that the optimal solution will have to pay for them somewhere in the graph. What this ends up doing is classifying pairs of vertices into three categories. The easy decisions corresponds to either non-admissible pairs, which are pairs that should never be in the same cluster, and atom pairs, which should always be in the same cluster. Finally, the hard decisions are categorized as admissible pairs, that is pairs of vertices that could be in the same cluster but do not have to be. The atom pairs construct connected components, that are typically referred to as atoms from the fact that they shouldn't be divided.

Like previous works, we want to use preclustering as the first step of our static algorithm. However, since we only have $O(|D|)$ time, we are not able to look at the whole graph. Instead, we only consider the clusters incident to $D$, since clusters not incident to $D$ must be perfect cliques. For each cluster, we will detect whether it is good enough to be an atom in the preclustering, by using a cleaning procedure similar to \cite{cao2024understanding}. Our new analysis allows us to clean arbitrary clusterings while only increasing the cost by a constant factor.

In the local search algorithm \cite{CLMTYZ24}, they consider a local change to a clustering $\mathcal{C}$ by adding a cluster $K$. 
To be precise, let $\mathcal{C} = \{C_1, \dots, C_k\}$, then define $\mathcal{C} + K = \{C_1 \setminus K, \dots, C_k \setminus K, K\}$. 

\begin{definition}[Degree-similarity]
    \label{def:degree-similar}
    Two vertices $u, v$ are called $\eps$-degree-similar if $\eps d(u) < d(v) < d(u)/\eps$. 
\end{definition}

\begin{definition}[Agreeing]
    A cluster $C \subseteq V$ is called $\beta$-agreeing if for every $v \in C$, $|N(v) \triangle C| < \beta|C|$. 
\end{definition}
\begin{lemma}
    \label{lem:agreeing-similar-size}
    For any $\beta$-agreeing cluster $C$ and $v \in C$, $(1 - \beta)|C| < |N(v)| < (1 + \beta)|C|$. 
\end{lemma}

\begin{proof}
    $|N(v)| \le |C| + |N(v) \triangle C| < (1 + \beta)|C|$, $|N(v)| \ge |C| - |N(v) \triangle C| > (1 - \beta)|C|$. 
\end{proof}

\begin{definition}[Subsume]
    A cluster $K$ is subsumed in a clustering $\mathcal{C}$ if there exists $C \in \mathcal{C}$ such that $K \subseteq C$. 
\end{definition}

\begin{definition}[Strong]
    A cluster $C \subseteq V$ is called strong if it is $\frac{1}{6}$-agreeing. 
    A clustering $\mathcal{C}$ is called strong if all the clusters in $\mathcal{C}$ are either singletons or strong. 
\end{definition}

Next we show one of the central properties of strong clusterings.
\begin{lemma}
    \label{lem:strong-subsume}
    Every strong cluster is subsumed in all locally optimal clusterings, where 
    a clustering $\mathcal{C}$ is called locally optimal if $\cost(\mathcal{C}) \le \cost(\mathcal{C} + K)$ for any cluster $K$. 
\end{lemma}
\begin{proof}
    Let $K$ be a $\beta$-agreeing cluster. 
    We will prove that every clustering that do not subsume $K$ is not locally optimal for some $\beta$. 
    Let $\mathcal{C} = \{C_i\}$ be any clustering, $C(v)$ be the cluster in $\mathcal{C}$ containing $v$ for any vertex $v$. 
    We consider two cases. 

    Let $\alpha$ be some constant to be determined later. If $|C_i \cap K| \le \alpha|K|$ for all $i$, then we will show $\cost(\mathcal{C} + K) < \cost(\mathcal{C})$. For every $v \in K$, it has to pay $|(K \setminus C(v)) \cap N(v)|$ for its adjacent neighbors in $K \setminus C(v)$, which leads to a total cost of
    \begin{align*}
        \text{cost of $\mathcal{C}$ for vertices in $K$} & \ge \frac{1}{2}\sum_{v \in K}|(K \setminus C(v)) \cap N(v)| \\
        & \ge \frac{1}{2}\sum_{v \in K}(|K| - |K \cap C(v)| - |K \setminus N(v)|) \\
        & \ge \frac{1}{2}(1 - \alpha)|K|^2 - \frac{1}{2}\sum_{v \in K}|K \setminus N(v)|.
    \end{align*}
    In the clustering $\mathcal{C} + K$, the total cost for vertices in $K$ is 
    \[
        \text{cost of $\mathcal{C} + K$ for vertices in $K$} = \sum_{v \in K}\left(|N(v) \setminus K| + \frac{1}{2}|K \setminus N(v)|\right).
    \]
    Therefore we have
    \begin{align*}
        & \cost(\mathcal{C}) - \cost(\mathcal{C} + K) \\
        = \quad & \text{(cost of $\mathcal{C}$ for vertices in $K$) - \text{(cost of $\mathcal{C} + K$ for vertices in $K$})} \\
        \ge \quad & \frac{1}{2}(1 - \alpha)|K|^2 - \frac{1}{2}\sum_{v \in K}|K \setminus N(v)| - \sum_{v \in K}\left(|N(v) \setminus K| + \frac{1}{2}|K \setminus N(v)|\right) \\
        \ge \quad & \frac{1}{2}(1 - \alpha)|K|^2 - \sum_{v \in K}|N(v) \triangle K| \\
        > \quad & \frac{1}{2}(1 - \alpha)|K|^2 - \beta|K|^2.
    \end{align*}
    So $\cost(\mathcal{C}) > \cost(\mathcal{C} + K)$ if $\frac{1}{2}(1 - \alpha) \ge \beta$. 

    If there exists $C_i \in \mathcal{C}$ such that $|C_i \cap K| > \alpha|K|$. 
    If $|C_i \setminus K| > 2\beta|K|$, for every vertex $v \in C_i \cap C$, it has 
    \[
        |C_i \setminus K \setminus N(v)| \ge |C_i \setminus K| - |K \setminus N(v)| \ge |C_i \setminus K| - |K \triangle N(v)| > |C_i \setminus K| - \beta|K| > \frac{1}{2}|C_i \setminus K|
    \]
    non-neighbors in $C_i \setminus K$, therefore $\cost(\mathcal{C} + (C_i \cap K)) < \cost(\mathcal{C})$ and $\mathcal{C}$ is not locally optimal. 
    So we only need to consider the case when $|C_i \setminus K| \le 2\beta|K|$. 
    In this case, we will consider $\mathcal{C} + (C_i \cup K)$. For every vertex $v$ in $K \setminus C_i$, $\mathcal{C}$ has to pay $|N(v) \cap C_i \cap K|$ for its positive neighbors in $C_i \cap K$, so in the total cost for vertices in $K \setminus C_i$ is
    \begin{align*}
        \text{cost of $\mathcal{C}$ for vertices in $K \setminus C_i$} & \ge \sum_{v \in K \setminus C_i}|N(v) \cap C_i \cap K| \\
        & \ge \sum_{v \in K \setminus C_i}(|C_i \cap K| - |K \setminus N(v)|) \\
        & \ge \sum_{v \in K \setminus C_i}(|C_i \cap K| - |K \triangle N(v)|) \\
        & > (\alpha - \beta)|K| \cdot |K \setminus C_i|.
    \end{align*}
    And in $\mathcal{C} + (C_i \cup K)$, 
    \begin{align*}
        \text{cost of $\mathcal{C} + (C_i \cup K)$ for vertices in $K \setminus C_i$} & \le \sum_{v \in K \setminus C_i}|N(v) \triangle (C_i \cup K)| \\
        & \le \sum_{v \in K \setminus C_i}(|N(v) \triangle K| + |C_i \setminus K|) \\
        & < 3\beta|K| \cdot |K \setminus C_i|.
    \end{align*}
    Then $\cost(\mathcal{C}) > \cost(\mathcal{C} + (C \cup C_i))$ if $\alpha - \beta \ge 3\beta$. 

    By choosing $\alpha = 2/3$ and $\beta = 1/6$, this lemma holds. 
\end{proof}

The next corollary follows directly from this lemma. 
\begin{corollary}
    \label{coro:strong-subsume}
    Every strong cluster is subsumed in $\opt$. 
\end{corollary}

Next, we will introduce the preclustering formally. 
\begin{definition}[Preclustering]
    A preclustering consists of a pair $(\mathcal{K}, E^{\adm})$, where $\mathcal{K}$ is a strong clustering and $E^{\adm} \subseteq \binom{V}{2}$ is the set of admissible pairs of vertices. 
\end{definition}
By admissible pairs of vertices, this refers to vertices that are allowed to be in the same cluster but are not required to. 
From this, we get that if a pair of vertices are not in the same cluster in $\mathcal{K}$ and not a pair in $E^{\adm}$ then they should never be clustered together.

Every non-singleton cluster $K \in \mathcal{K}$ is called an \emph{atom}, which must be subsumed in $\opt$ according to \cref{lem:strong-subsume}. 
All pairs that are contained in the same atom are called \emph{atomic}, and the others are non-admissible if they are also not admissible. 
A cluster $\mathcal{C}$ is \emph{accepted} by the preclustering $(\mathcal{K}, E^{\adm})$ if it does not separate any atom pairs nor join any non-admissible pairs. 
To achieve a good approximation, we require that there exists a clustering $\mathcal{C}$ accepted by $(\mathcal{K}, E^{\adm})$, such that $\cost(\mathcal{C}) \le \cost(\opt) + O(\eps\cdot\cost(\mathcal{K}))$. We are going to show that one such exists. 

For every $D \subseteq \binom{V}{2}$, let $d_D(v)$ be the number of neighbors of $v$ if the edge set is $D$, that is, $d_D(v) = |\{u \in V \mid (u, v) \in D\}|$. Similarly, for any $K \subseteq V$, let $d_D(K) = \sum_{v \in K}d_D(v)$. 
As a special case, in a preclustering $(\mathcal{K}, E^{\adm})$, let $\dadm(v) = d_{E^{\adm}}(v)$ for each vertex $v$. 
A vertex $u$ is called an admissible neighbor of $v$ if $(u, v) \in E^{\adm}$. 
To make a fast algorithm, we need an upper bound on the number of admissible pairs. 

We also need some more properties for a fast algorithm, which will be discussed in \cref{sec:preclustering} in detail. 

\begin{theorem}[Cleaning]
    For any $(\mathcal{C},D)$, we can update it to $(\mathcal{C'},D')$ in $O(|D|)$ time, such that $|D'| = O(|D|)$ and each (non-singleton) cluster in $\mathcal{C'}$ is an atom.
\end{theorem}

In this section, we will discuss how to construct a preclustering from the clustering $\mathcal{C}$ in $O(|D|)$ time, captured in the next theorem. We remark that, in the remaining sections, we will compare our clustering to the optimal accepted clustering instead of the real optimal clustering, as we have shown that their costs are close enough.

\begin{theorem}
    \label{thm:preclustering}
    For any $\eps \in (0, 0.1)$, given a graph represented by $(\mathcal{C}, D)$, in $O(|D|)$ time, we can construct a preclustering $(\mathcal{K}, E^{\adm})$, such that
    \begin{enumerate}
        \item For any $(u, v) \in E^{\adm}$, either $u$ or $v$ is a singleton in $\mathcal{K}$; 
        \item For any vertex $v$, $\dadm(v) = O(1/\eps^2) d(v)$; 
        \item For any atom $K$, $\dadm(K) \le \frac{3d_{D}(K)}{|K|}$; 
        \item $|E^{\adm}| = O(1/\eps^2)\cost(\mathcal{K})$; 
        \item For any $(u, v) \in E^{\adm}$, $u, v$ are $\eps$-degree-similar;
        \item There exists an accepted clustering $\mathcal{C}$, such that
        \setcounter{enumcount}{\value{enumi}} %
        \begin{itemize}
            \item $\cost(\mathcal{C}) \le \cost(\opt) + 4\eps \cdot \cost(\mathcal{K})$; 
            \item For each $C \in \mathcal{C}$ and $v \in C$, either $|C| = 1$ or $|C| > \eps d(v)$. 
        \end{itemize}
    \end{enumerate}
    We provide the following ways to access $E^{\adm}$:
    \begin{enumerate}
        \setcounter{enumi}{\value{enumcount}}%
        \item List out all admissible neighbors of the atom $K$ in time proportional to its number of admissible neighbors. 
        \item Check if two vertices are admissible in $O(t)$ time with probability $1 - \exp(-\Omega(t))$; 
        \item List out all admissible neighbors of a singleton vertex $v$ in $O(t^2\cdot d(v))$ time with probability $1 - \exp(-\Omega(t))$ for $t = \Omega(\log d(v))$; 
    \end{enumerate}
\end{theorem}

We are going to show this theorem over course of this section. Our algorithm is based on the cleaning algorithm (Algorithm 4 in \cite{cao2024understanding}). 

\subsection{Finding atoms}

The first part of the algorithm is to construct a strong clustering from the given clustering $\mathcal{C}$ with the symmetric difference $D$ in $O(|D|)$ time. 
The algorithm does exactly the same thing as the previous work \cite{cao2024understanding}. For each cluster, $C \in \mathcal{C}$ we remove vertices from $C$ whose neighborhood differs from $C$ by at least $\eps|C|$. Then, if a sufficiently large fraction was not removed, the remaining vertices will form a strong cluster. 
After determining those strong clusters, we can go through the list $D$ and compute the new symmetric difference, just like how we implement the Pivot algorithm. 
The pseudocode is given in \cref{alg:atoms}. 
And it's easy to check that the algorithm returns a correct clustering $(\mathcal{C}', D')$ and takes $O(|D| + |D'|)$ time. 
\begin{algorithm}[ht]
    \caption{Clean($\mathcal{C}, D$)}
	\label{alg:atoms}
	\begin{algorithmic}[1]
        \State set $a = b = 0.05$
        \State $\mathcal{S} \gets \emptyset$, $V' \gets \emptyset$
        \For{each $(u, v) \in D$}
            \State Insert $C(u)$ and $C(v)$ to $\mathcal{S}$
            \State Insert $u$ and $v$ to $V'$
            \State $p_u \gets 0, p_v \gets 0$
        \EndFor
        \For{each $(u, v) \in D$}
            \State $p_u \gets p_u + 1, p_v \gets p_v + 1$
        \EndFor
        \For{each $v \in V'$}
            \If{$p_v \ge a|C(v)|$}
                \State Mark $v$
            \EndIf
        \EndFor
        \For{each $C \in \mathcal{S}$}
            \State remove $C$ from $\mathcal{C}$
            \If{less than $b|C|$ vertices are marked in $C$}
                \State Let $K$ be the set of unmarked vertices in $C$
                \State Insert $K$ to $\mathcal{C}$
                \State Insert a singleton to $\mathcal{C}$ for each marked vertex in $C$
            \Else
                \State Insert a singleton to $\mathcal{C}$ for each vertex in $C$
            \EndIf
        \EndFor
        \State $D' \gets \emptyset$
        \For{each $(u, v) \in D$}
            \State add $(u, v)$ to $D'$ if $(u, v)$ is
            \begin{itemize}
                \item either a minus edge inside a cluster of $\mathcal{C}$
                \item or a plus edge cross two non-singleton clusters of $\mathcal{C}$
            \end{itemize}
        \EndFor
        \For{each $u \in S$ where $u$ is a singleton in $\calC$}
            \State $D' \gets D' \cup \{(u, v): v \text{ is a neighbor of } u\}$
        \EndFor
        \State Return $(\mathcal{C}, D')$
	\end{algorithmic}
\end{algorithm}

\begin{lemma}
    Let $(\mathcal{C}', D')$ be the output of \hyperref[alg:atoms]{\color{black}Clean}($\mathcal{C}, D$). Then
    \[
        |D'| \le \left(1 + \frac{1}{ab}\right)|D| ~.
    \]
    where $a$ and $b$ are the parameters of the \hyperref[alg:atoms]{\color{black}Clean}$(\calC,D)$ algorithm.
\end{lemma}

\begin{proof}
    Now we will bound $|D'| - |D|$ induced by subdividing each cluster. 
    For each $C \in \mathcal{C}$, let $K$ be the set of unmarked vertices in $C$ in Clean($\mathcal{C}, D$). 
    According to \cref{alg:atoms}, for every vertex $v \in C \setminus K$, $|N(v) \triangle C| \ge a|C|$. 

    If $K > (1 - b)|C|$, then we will partition $C$ into $K$ and singletons, and the number of new edges induced for cluster $K$ is at most
    \[
        |C \setminus K| \cdot |C| ~,
    \]
    and we have
    \[
        \sum_{v \in C \setminus K}|N(v) \triangle C| \ge |C \setminus K| \cdot a|C| ~.
    \]
    If $K \le (1 - b)|C|$, then we will partition $C$ into singletons, and the number of new edges induced is at most $\binom{|C|}{2} < \frac{1}{2}|C|^2$, and we have
    \[
        \sum_{v \in C \setminus K}|N(v) \triangle C| \ge |C \setminus K| \cdot a|C| \ge ab|C|^2 ~.
    \]
    Since $b = 0.05 < 0.5$, we have the extra cost is at most
    \[
        \frac{1}{2ab}\sum_{v \in C' \setminus K'}|N(v) \triangle C'| ~.
    \]
    Summing up over all clusters in $\mathcal{C}$, we get
    \[
        |D'| - |D| \le \frac{1}{2ab}\sum_{v \in V'}|N(v) \triangle C(v)| = \frac{1}{ab}|D| ~.
    \]
\end{proof}

Since $|D'| = O(|D|)$, and the algorithm runs in $O(|D| + |D'|)$ time, we have the total running time is also $O(n)$, summarized by the next corollary. 

\begin{corollary}
    Algorithm~\ref{alg:atoms} takes $O(|D|)$ time. 
\end{corollary}

\begin{lemma}
    Every cluster in $\mathcal{C'}$ is either $\frac{a + b}{1 - b}$-agreeing or a singleton. For $a = b = 0.05$, $\mathcal{C'}$ is a strong clustering. 
\end{lemma}

\begin{proof}
    Let $(\mathcal{C}, D)$ be the input of \cref{alg:atoms}. 
    At the start of line 8, $|N(v) \triangle C(v)| = p_v$ if $v \in V'$ otherwise 0. 
    Then, for any cluster $C \notin S$ and $v \in C$, $N(v) = C(v)$, so $C$ is a perfect clique. 
    Therefore, all remaining cliques satisfy the condition. 
    Now we focus on the new inserted clusters. 
    For each $C \in S$, let $K$ be the set of unmarked vertices in $C$. 
    If $|K| \le (1 - b)|C|$, then we will replace $C$ by a collection of singletons. 
    Without loss of generality, assume $|K| > (1 - b)|C|$., and we will replace $C$ by $K$ and a collection of singletons. 
    For each $v \in K$, we have
    \[
        |N(v) \triangle K| \le |N(v) \triangle C| + |C \setminus K| < (a + b)|C| < \frac{a + b}{1 - b}|C| ~.
    \]
    Therefore $K$ is an $\frac{a + b}{1 - b}$-agreeing cluster. 
    When $a = b = 0.05$, $K$ is $\frac{2}{19}$-agreeing and strong. 
    By definition, $\mathcal{C}'$ is a strong clustering. 
\end{proof}

\subsection{Admissible pairs}

In the previous subsection, we construct a strong clustering $\mathcal{C}$. 
In the rest of this section, we will use it as the set of atoms for a preclustering, and referred to it as $\mathcal{K}$. 
All non-singleton clusters in $\mathcal{K}$ are atoms. 
In this subsection, we will define admissible pairs based on the atoms, and provide an oracle to access the admissible neighbors. 

With an $O(|D|)$ time preprocess, we are able to compute the degree for each vertex in constant time. 

\begin{definition}[Admissible pair]
    Let $u, v \in V'$ be two distinct vertices, we say $(u, v)$ is admissible if and only if they are degree similar \footnote{Recall that $u, v$ are $\eps$-degree-similar if $\eps d(u) < d(v) < d(u)/\eps$. } and one of the following conditions hold
    \begin{itemize}
        \item Both $u$ and $v$ are singletons, and at least $\eps(d(u) + d(v))$ common
        $\eps$-degree-similar vertices exists in $N(u) \cap N(v)$; 
        \item $u$ is a singleton, $v$ is contained in an atom $K(v)$, $u$ is $\eps$-degree-similar to all vertices in $K(v)$ and $|N(u) \cap K(v)| \ge \frac{1}{3}|K(v)|$, and similarly if $u$ is contained in an atom and $v$ is a singleton. 
    \end{itemize}
\end{definition}

Note that in this definition, for any atom $K$, $v \in K$ and $u \in V \setminus K$, $(u, v)$ is admissible if and only if $(u, v')$ is admissible for all $v' \in K$. 
Therefore, we say a vertex $u$ is admissible to $K$ if $u$ is admissible to all vertices in $\mathcal{K}$. 
And we use $\dadm(K)$ as the number of admissible neighbors in $\mathcal{K}$. 

We will show that the atoms together with the admissible pairs give a $\eps$-preclustering for any $\eps < 0.1$. 
First, we will show that $\dadm$ is bounded by $d$ for all vertices. 

\begin{lemma}
    Suppose $\eps < 0.1$. 
    For any vertex $v \in V$, $\dadm(v) \le (1/\eps^2 + 3)d(v)$. 
    For any atom $K$, $\dadm(K) \le \frac{3d_{D}(K)}{|K|}$. 
\end{lemma}

\begin{proof}
    It's easy to see that $\dadm(v) = 0$ if $d(v) = 0$, so without generality, assume $d(v) \ge 1$. 

    If $v$ is contained in an atom $\mathcal{K}$, by definition, $K$ is a strong cluster, so for any $v' \in K$, $|N(v') \triangle K| < \frac{1}{6}|\mathcal{K}|$. 
    Therefore, the total number of edges that cross $K$, is
    \[
        \sum_{v' \in K}|N(v') \setminus K| \le \sum_{v' \in K}|N(v') \triangle K| \le \frac{1}{6}|K|^2 ~.
    \]
    By definition, $(u, v)$ is admissible only if $u$ is not contained in any atom and $|N(u) \cap K| \ge \frac{1}{3}|K|$. 
    According to the upper bound on the total number of edges cross $K$, $\dadm(v)$ is at most $\left(\frac{1}{6}|K|^2\right) / \left(\frac{1}{3}|K|\right) = 2|K|$. 
    According to \cref{lem:agreeing-similar-size}, $|K| \le \frac{6}{5}|N(v)| = \frac{6}{5}(d(v) + 1) < 5/2\cdot d(v)$. 
    Therefore, $\dadm(v) < (1/\eps^2 + 3)d(v)$ in this case. 
    Since $|N(v') \triangle K| = d_D(v)$ for any $v' \in K$, we have $\dadm(K) \le \frac{3\sum_{v \in K}d_{D}(v)}{|K|} = \frac{3d_D(K)}{|K|}$. 
    
    If $v$ is not contained in any atom, we will bound the number of admissible neighbors contained in an atom or not contained in any atom separately. 
    If $(u, v)$ is admissible and $u$ is not contained in any atom, by definition, they have $\eps(d(u) + d(v))$ admissible neighbors. 
    For each $\eps$-degree-similar vertex $p$ of $v$, by definition of degree-similarity, $d(p) < d(v)/\eps$, so the total number of triples $(v, p, u)$ that $(v, p)$, $(u, p)$ are both admissible is at most $d^2(v)/\eps$. 
    For each $u$ that $(u, v)$ is admissible in this case, there has to be at least $\eps(d(u) + d(v))$ triples $(v, p, u)$, therefore the number of $u$'s is at most $d(v)/\eps^2$. 

    For each atom $K$, $v$ has all vertices in $K$ as its admissible neighbors if $|N(v) \cap K| \ge \frac{1}{3}|K|$, or none of the vertices in $\mathcal{K}$ are admissible neighbors of $v$. 
    Since $v$ is not contained in any atom, $N(v) \cap K = (N(v) \setminus \{v\}) \cap K$ for all atom $K$. 
    Therefore, the number of admissible neighbors of $v$ contained in any atom, is at most $3|N(v) \setminus \{v\}| = 3d(v)$. 
    So for any $v$ that is not contained in any atom, $\dadm(v) \le (1/\eps^2 + 3)d(v)$. 
\end{proof}

Next, we will give an upper bound on the number of admissible pairs in total. 

\begin{lemma}
    $|E^{\adm}| < (2/\eps^2 + 6)\cost(\mathcal{K})$. 
\end{lemma}

\begin{proof}
    We don't have admissible pairs $(u, v)$ that both $u$ and $v$ are contained in atoms, so we have
    \[
        |E^{\adm}| \le \sum_{v \text{ is a singleton in }\mathcal{K}}\dadm(v) \le \sum_{v \text{ is a singleton in }\mathcal{K}}(1/\eps^2 + 3)d(v) \le (2/\eps^2 + 6)\cost(\mathcal{K}) ~.
    \]
\end{proof}

We will first prove that no two atoms are subsets of the same cluster in $\opt$ simultaneously. And it suffices with the lemma below. 

\begin{lemma}
    \label{lem:opt-strong-size}
    Let $C$ be a strong cluster and $C^*$ be the cluster in $\opt$ where $C \subseteq C^*$. Then $|C| \ge \frac{3}{4}|C^*|$. 
\end{lemma}

\begin{proof}
    For contradiction, suppose $|C| < \frac{3}{4}|C^*|$, so therefore $|C^*| > \frac{4}{3}|C|$. 
    By definition of strong clusters, $C$ is $\frac{1}{6}$-agreeing so for any $v \in C$, $|N(v) \triangle C| < \frac{1}{6}|C|$. 
    Therefore, for any $v \in C$, it has at most $\frac{1}{6}|C|$ neighbors in $C^* \setminus C$, and at least $|C^* \setminus C| - \frac{1}{6}|C| > \frac{1}{6}|C|$ non-neighbors in $C^* \setminus C$. 
    Since it's preferred for all vertices in $C$ to be separated from $C^* \setminus C$, we have $\cost(\opt + C) < \cost(\opt)$, which contradicts with the optimality. 
\end{proof}

\begin{lemma}
    \label{lem:eps-apx-admissible}
    For sufficiently small $\eps$, let $\mathcal{K}$ be the strong clustering used as atoms, $E^{\adm}$ be the set of admissible pairs defined above based on $\mathcal{K}$. 
    There exists an accepted clustering $\mathcal{C}'$ such that
    \begin{enumerate}
        \item Every cluster $K \in \mathcal{K}$ is subsumed in $\mathcal{C}$; 
        \item There is no two distinct clusters $K_1, K_2 \in \mathcal{C}$ and $C \in \mathcal{C}'$ such that $|K_1| > 1$, $|K_2| > 1$ and $K_1 \cup K_2 \subseteq C$; 
        \item For any $u \in C$, let $C(u)$ be the cluster of $u$ in $\mathcal{C}$, then $|N(u) \cap C(u)| > \frac{1}{2}|C(u)| + \eps d(u)$; 
        \item $\cost(\mathcal{C}) \le \cost(\opt) +  4\eps\cdot\cost(\mathcal{K})$. 
    \end{enumerate}
\end{lemma}

\begin{proof}
    We will construct a clustering from $\opt$ by breaking a cluster into smaller ones. 
    Note that every cluster $K \in \mathcal{K}$ is strong, so it must be subsumed in $\opt$ by \cref{coro:strong-subsume}. 
    According to \cref{lem:opt-strong-size}, no two disjoint strong clusters will be subsets of the same cluster in $\opt$. 
    Since we will only break a cluster into smaller ones, condition 2 must be satisfied throughout the process. 

    Assume condition 1 holds throughout the process. 
    First, we will show that for any $u$ contained in an atom $K$, throughout the process, $|N(u) \cap C| > \frac{1}{2}|C| + \eps d(u)$. 
    According to condition 1, $K \subseteq C(u)$. 
    By \cref{lem:opt-strong-size}, $\frac{3}{4}|C^*(u)| \le |K| \le |C(u)|$, where $C^*(u)$ is the cluster of $u$ in $\opt$. 
    Since we will only break cluster into smaller ones, $|C^*(u)| \ge |C(u)|$, so we have $\frac{3}{4}|C(u)| \le |K| \le |C(u)|$. 
    Therefore, $|N(u) \cup C(u)| \ge |N(u) \cup K| \ge |K| - |N(u) \triangle K| > \frac{5}{6}|K| \ge \frac{5}{8}|C(u)|$. 
    By \cref{lem:agreeing-similar-size}, $d(u) < |N(u)| < \frac{7}{6}|K| \le \frac{7}{6}|C(u)|$. 
    Therefore, $|N(u) \cup C(u)| > \frac{1}{2}|C(u)| + \eps d(u)$ when $\eps < 0.1$. 
    
    We start with $\mathcal{C} = \opt$. 

    Suppose there exists a cluster $u \in V$ that violates condition 3. 
    Let $C(u)$ be the cluster of $u$ in $\mathcal{C}$. 
    According to the assumption, $|N(u) \triangle C(u)| \le \frac{1}{2}|C(u)| + \eps d(u)$. 
    Therefore, $u$ must be a singleton in $\mathcal{K}$. 
    Then we will replace $\mathcal{C}$ by $\mathcal{C} + \{u\}$ and continue. 
    The extra cost induced in this step is
    \begin{align*}
        & |N(u) \cap C(u)| - 1 - |C(u) \setminus N(u)| \\
        \le & |N(u) \cap C(u)| - 1 - (|C(u)| - |N(u) \cap C(u)|) \\
        = & 2|N(u) \cap C(u)| - 1 - |C(u)| \\
        < & 2\eps\cdot d(u)
    \end{align*}
    Let $C$ be the final cluster we get. 
    We have
    \[
        \cost(\mathcal{C}) - \cost(\opt) \le \sum_{v \text{ is a singleton in }\mathcal{K}}2\eps\cdot d(v) \le 4\eps\cdot\cost(\mathcal{K}) ~.
    \]
    Therefore, the final cluster $C$ satisfies all the conditions. 
\end{proof}

Next, we will verify that the cluster $\mathcal{C}$ constructed in the previous proof satisfies condition 6 in \cref{thm:preclustering}. 
Since we have showed that $\cost(\mathcal{C}) \le \cost(\opt) + 4\eps\cdot\cost(\mathcal{K})$, we only need to prove that $\mathcal{C}$ is accepted and each $C \in \mathcal{C}$ has a proper size. 
We prove this in the next lemma. 

\begin{lemma}
    Let $\mathcal{C}$ be the cluster generated in the proof of \cref{lem:eps-apx-admissible}. 
    For any cluster $C \in \mathcal{C}$ and two distinct vertices $u, v \in C$, $u, v$ are $\eps$-degree-similar and $(u, v)$ is either atomic or admissible. 
    Moreover, for any $C \in \mathcal{C}$ and $u \in C$, either $|C| = 1$ or $|C| > \eps d(u)$. 
\end{lemma}

\begin{proof}
    By construction of $\mathcal{C}$, if we have $|N(u) \cap C| > \frac{1}{2}|C| + \eps d(u)$ and $|N(v) \cap C| > \frac{1}{2}|C| + \eps d(v)$, 
    then
    \[
        |C| \ge |N(u) \cap C| > \eps d(u) ~, 
    \]
    Therefore $|C| = 1$ or $\forall u \in C, |C| > \eps d(u)$ always hold. 

    We will next show that $u, v$ are $\eps$-degree-similar. 
    We have
    \[
        d(u) + 1 = |N(u)| \ge |N(u) \cap C| > \frac{1}{2}|C| + \eps d(v) \ge 1 + d(v) ~,
    \]
    Therefore we have $d(u) > \eps d(v)$. By symmetry, $d(v) > \eps d(v)$, therefore $u, v$ are degree similar. 
    
    If $u$ is contained in an atom $K$ and $v$ is a singleton, according to the proof, $|K| \ge \frac{3}{4}|C|$, we have
    \begin{align*}
        \frac{|C \setminus K|}{|K|} & = \frac{|C| - |K|}{|K|} \le \frac{1}{3} ~, \\
        |N(v) \cap K| & \ge |N(v) \cap C| - |C \setminus K| > \frac{1}{2}|C| - |C \setminus K| \ge \frac{1}{2}(|K| - |C \setminus K|) \ge \frac{1}{3}|K| ~.
    \end{align*}
    If $u$ is a singleton and $v$ is contained in an atom, the same proof above holds if we swap $u$ and $v$. 
    If $u, v$ are both singletons in $\mathcal{K}$, we have
    \[
        |N(u) \cap N(v)| \ge |N(u) \cap C| + |N(v) \cap C| - |C| > \eps(d(u) + d(v)) ~.
    \]
    By construction of $\mathcal{C}$, if $u, v$ are both contained in some atom, then they must be in the same atom. 
    Therefore, $(u, v)$ is always either atomic of admissible, and $\mathcal{C}$ is accepted. 
\end{proof}

We summarize our conclusion by now in the following corollary. 
\begin{corollary}
    \label{cor:admissible-edges}
    Given any strong clustering $\mathcal{K}$, for any $\eps < 0.1$, we can use all non-singleton clusters in $\mathcal{K}$ as atoms and define admissible edges in the above way, to get an $\eps$-preclustering instance $(\mathcal{K}, E_{adm})$, such that
    \begin{enumerate}
        \item For each vertex $v \in V$, $\dadm(v) < (1/\eps^2 + 3) d(v)$; 
        \item $|E^{\adm}| < (2/\eps^2 + 6)\cost(\mathcal{K})$; 
        \item There exists an accepted clustering $\mathcal{C}$, such that
        \begin{itemize}
            \item $\cost(\mathcal{C}) \le \cost(\opt) + 4\eps \cdot \cost(\mathcal{K})$; 
            \item For each $C \in \mathcal{C}$ and $v \in C$, either $|C| = 1$ or $|C| > \eps d(v)$. 
        \end{itemize}
    \end{enumerate}
\end{corollary}

\subsection{Computing admissibility}

The only thing left now is how to actually compute the admissibility of individual pairs of vertices. In the process of running the local search with high probability, we need to determine the admissible edges with similar running time and a probability of success as the rest of the local search. 

We will first give a $O(|D|)$ time algorithm to list out all admissible neighbors to all atoms. 

\begin{lemma}
    \label{lem:admissible-to-atom}
    Let $\mathcal{K}$ be the atoms and $D = \mathcal{E(K)} \triangle E$ be the symmetric difference. 
    In $O(|D|)$ time, we can compute the list of all admissible neighbors for all atoms, and the list of all admissible atoms for all singletons. 
\end{lemma}

\begin{proof}
    For each atom $K$, we can compute $|N(v) \cap K|$ for each $v \notin K$ by enumerating the outgoing edge from $K$, then we can check for each vertex $v \notin K$ with at least one neighbor in $K$, if it is admissible to $K$ or not by definition. 
    The total running time is proportional to the number of outgoing edges from atoms, which is $O(|D|)$. 
\end{proof}

Next, we will give an algorithm in $O(t)$ time that can check if a pair $(u, v)$ is admissible when both $u$ and $v$ are singletons. 
\begin{algorithm}[ht]
	\caption{CheckAdmissible($u, v \in V, t$)}
	\label{alg:check-admissible}
	\begin{algorithmic}[1]
        \State If $d(u) < d(v)$ swap $u$ and $v$
        \State Sample $T = t/\eps^2$ vertices $x_1, \dots, x_t$ from $N(u)$ with replacement
        \State $\tilde{s} \gets \frac{|N(u)|}{t}\sum_{i = 1}^T [x_i \in N(v) \text{ and } \eps\text{-degree-similar to both }u \text{ and } v]$
        \State return $\tilde{s} > \eps/2(d(u) + d(v))$
	\end{algorithmic}
\end{algorithm}

\begin{lemma}
    \label{lem:prob-check}
    For any $u, v \in V$ and $t$, 
    \begin{itemize}
        \item If $(u, v)$ is admissible in the $\eps$-preclustering, \cref{alg:check-admissible} will return true with probability at least $1 - \exp(-t/2)$; 
        \item If $(u, v)$ is not admissible in the $(\eps/6)$-preclustering, \cref{alg:check-admissible} will return false with probability at least $1 - \exp(-t/18)$. 
    \end{itemize}
\end{lemma}

\begin{proof}
    Note that $\mu = E[\tilde{s}]$ is just the number of vertices in $N(u) \cap N(v)$ that are $\eps$-degree-similar to both $u$ and $v$. 
    
    If $(u, v)$ is admissible in the $\eps$-preclustering, by definition, there are at least $\eps(d(u) + d(v))$ vertices in $N(u) \cap N(v)$ that are $\eps$-degree-similar to both $u$ and $v$. 
    Therefore $\mu \ge \eps(d(u) + d(v)) \ge \eps |N(u)|$. 
    According to Hoeffding's inequality, 
    \[
        \Pr[\tilde{s} \le \mu/2] \le \exp\left(-\frac{T\mu^2}{2|N(u)|^2}\right) \le \exp\left(-t/2\right) ~.
    \]
    If $(u, v)$ is not admissible in the corresponding $(\eps/6)$-preclustering, by definition, there are at most $\eps(d(u) + d(v))/6$ vertices in $N(u) \cap N(v)$ that are $(\eps/6)$-degree-similar to both $u$ and $v$. 
    Therefore, there are at most $\eps(d(u) + d(v))/6$ vertices in $N(u) \cap N(v)$ that are $\eps$-degree-similar to both $u$ and $v$. 
    So $\mu \le \eps(d(u) + d(v))/6 \le \eps |N(u)|/3$. 
    According to Hoeffding's inequality, 
    \[
        \Pr[\tilde{s} > \eps/2(d(u) + d(v))] \le \Pr[\tilde{s} - \mu > \eps/6|N(u)|] \le \exp\left(-t/18\right) ~.
    \]
\end{proof}

Therefore, with exponentially high probability on $t$, we can use \cref{alg:check-admissible} to find a fixed admissible graph that contains all admissible pairs in the $\eps$-preclustering, and also a subset of all admissible pairs in the $(\eps/6)$-preclustering. 

Next we will give an algorithm in $O(t^2 d(v))$ time to list out all the admissible neighbors of a singleton vertex $v$. 
\begin{algorithm}[ht]
	\caption{ListAdmissible($v \in V, t$)}
	\label{alg:list-admissible}
	\begin{algorithmic}[1]
        \State For each admissible atom of $v$, add all vertices in the atom to the list
        \State Sample $T' = t/\eps$ $\eps$-degree-similar neighbors of $v$ with replacement, check all their neighbors whether they are admissible to $v$ or not using \cref{alg:check-admissible} with the same parameter $t$. 
        \State return the combined list without repetition
	\end{algorithmic}
\end{algorithm}

\begin{lemma}
    \label{lem:prob-list}
    For every vertex $u$ that is admissible to $v$, with probability at least $1 - \exp(-t)$, we will check if it is admissible to $v$ in \cref{alg:list-admissible}. 
\end{lemma}

\begin{proof}
    If $(u, v)$ is admissible, then we will check this pair $(u, v)$ in \cref{alg:list-admissible} when common $\eps$-degree-similar vertex in $|N(u) \cap N(v)|$ is sampled. 
    Since there are at least $\eps(d(u) + d(v))$ of them, we will sample one of them each time with probability $\eps$ each time, therefore the probability that none of them are hit anytime, is at most $(1 - \eps)^{T'} \le \exp(-t)$. 
\end{proof}

According to union bound, we will list all admissible neighbors of $v$ between the $\eps$-preclustering and the $(\eps/6)$-preclustering with exponentially high probability on $t$ for any $t = \Omega(\log d(v))$. 

Suppose we ``record'' all the result when we query \cref{alg:check-admissible}, and define admissible pairs to be consistent with the output so far, unless it fails as \cref{lem:prob-check} says. 
Then we get a preclustering that satisfies all the conditions, and so we have completed the proof of \cref{thm:preclustering}. 

%% file: dynamicLocalSearch-new.tex
\section{Dynamic Local Search}
\label{sec:local-search}

\subsection{Iterated-flipping local search}

In this section, we introduce a slightly different version of the iterated-flipping local search algorithm in $\cite{CLMTYZ24}$. 
The algorithm takes a preclustering $(\mathcal{C},D)$ as input, and when we talk about a cluster or a clustering, we always assume they respect the preclustering. 
Combining with our new implementation of the local search procedure in \cref{sec:local-search-implementation}, we will get a running time of $O(|D|)$ and then it can be turned into a dynamic algorithm by our dynamic protocol in \cref{sec:dynamic-framework}. 

\begin{definition}
    Let $D$ be the symmetric difference (in the input of the iterated-flipping local search) and $\mathcal{C}^*$ be the current optimal clustering. For any constant $\eps$, we say a clustering $\mathcal{C}$ is a \emph{$\eps$-good local optimum (with respect to $D$)} if
    \[\sum_{C\in\calC^*} ( \cost(\calC) - \cost(\calC + C) )  \le \eps |D| .\]
\end{definition}

Note that our definition of $\eps$-good local optimum is different from that in $\cite{CLMTYZ24}$. We reformulate the iterated-flipping local search algorithm and the approximation guarantee as follows.

\begin{algorithm}[H]
    \caption{Iterated-flipping Local Search from \cite{CLMTYZ24}}
    \label{alg:flipping-local-search}
	\begin{algorithmic}[1]
        \State Let $w_0$ be the initial weight function where each edge has a cost $1$.
        \State Let $\calC_0'$ be an $\eps$-good local optimum for $w_0$.
        \For{$i=1,\dots,O(1/\eps)$}
            \State Let $w_i$ be the weight function created by increasing the cost of violated edges in $\calC_{i-1}'$ by $1/2$ from $w_0$.
            \State Let $\calC_i$ be an $\eps$-good local optimum for $w_i$.
            \State Let $w_i'$ be the weight function created by increasing the cost of violated edges in $\calC_i$ by $1/2$ from $w_i$.
            \State Let $\calC_i'$ be an $\eps$-good local optimum for $w_i'$.
            \State $\calC_i'' \gets \mathrm{Pivot}(\calC_{i-1}',\calC_i,\calC_{i}')$ using the Pivot algorithm in Lemma 18 of \cite{CLMTYZ24}
            \footnotemark
        \EndFor
        \State Return the best of all the above clusterings.
	\end{algorithmic}
\end{algorithm}
\footnotetext{It can be easily computed in $O(n)$ time originally and in $O(|D|)$ time in our case. See the pivot algorithm in section 3.1 of \cite{CLMTYZ24} for more details. }
\begin{theorem}[\cite{CLMTYZ24} Theorem 23]
\label{thm:iter-ls}
    The iterated-flipping local search algorithm returns a $(2-2/13+O(\eps k))$-approximation, where $k$ is the approximation factor of the input $(\mathcal{C},D)$.
\end{theorem}

To achieve our intended result, it remains to show that we can find an $\eps$-good local optimum via local search in $O(|D|)$ time. Before that, let's introduce some useful properties about local search.
Let $\calC^*=\{C_1^*,\dots,C_k^*\}$ be the optimal clustering. For the current clustering $\calC$, define $\Delta_i=\max\{\cost(\calC) - \cost(\calC + C_i^*) - \eps^3 d^{adm}(C_i^*) , 0\}$ to be the potential improvement of inserting $C_i^*$, and define $\Delta=\sum_{i=1}^k \Delta_i$. Note that $\Delta$ and $\Delta_i$ change over time as we modify our clustering along the algorithm step by step.

\begin{lemma}[\cite{CLMTYZ24} Lemmas 40 and 54]
\label{lem:gen-cluster}
    Given a clustering $\calC$, a vertex $v\in C_i^*$\footnote{Strictly speaking, if $C_i^*$ contains an atom, $v$ should be in the atom. We will deal with it in \cref{lem:hit}.} (without knowing $C_i^*$) and the admissible neighborhood of $v$, we can generate a cluster $C$ in $O(d^{adm}(v)d(v))$ time which can improve $\calC$ by at least $\Delta_i$ with constant probability.
\end{lemma}

\begin{lemma}[Improving \cite{CLMTYZ24} Lemma 55]
\label{lem:est-imp}
    For any vertex $v$ and any $\gamma\ge d^{adm}(v)\log^2 d^{adm}(v)$, the cluster $C$ stated in Lemma \ref{lem:gen-cluster} can be generated in $O(\gamma)$ time, which has an estimated improvement of at least $\Delta_i$ with constant probability, and the real improvement is greater than or equal to the estimated improvement with probability $1-exp(-\gamma)$.
\end{lemma}%

\begin{proof}
    The basic idea is to use sampling to estimate the improvement of the cluster $C$ in \cref{lem:gen-cluster}. Lemma 55 in \cite{CLMTYZ24} suffices to show the same result but with $O(\gamma d^{adm}(v))$ time. The reason they have the extra $d^{adm}(v)$ term is that they sample $O(\gamma)$ vertices and then spend $O(d^{adm}(v))$ to compute the exact contribution of these vertices. However, if we sample $O(\gamma)$ edges instead, the same concentration bound holds. We can efficiently sample edges since the vertices in the same cluster must have similar degree. Other proof details are the same as \cite{CLMTYZ24} Lemma 55.
\end{proof}

The property of local search algorithm guarantees that from a vertex $p$, the generated cluster $C$ satisfies $K(p) \subseteq C \subseteq \Ncand(p)$ (in line~\ref{line:generate}), where $K(p)$ is the atoms containing $p$ or $\{p\}$ if $p$ is singleton in the preclustering, and $\Ncand(p) = K(p) \cup (\Nadm(p) \setminus (\cup_{K \text{ is an atom}}K))$. 
We will reuse these notations in \cref{sec:LP} as well. 
When we sample vertices, we say a cluster $C$ is hit if $K(p) \subseteq C \subseteq \Ncand(p)$. 

\begin{lemma}[Similar to \cite{CLMTYZ24}]
\label{lem:hit}
    If we sample a vertex such that each vertex $v$ is sampled with probability $\Theta\left(\frac{1}{|D|\log^2 d(v)}\right)$, then each non-singleton cluster $C_i^*$ will be hit with probability $\Omega\left(\frac{|C^*_i|}{|D|\log^2 |C^*_i|}\right)$. 
\end{lemma}

\begin{proof}
    By \cref{thm:preclustering}(6), $d(v)=O(|C_i^*|)$ for every $v\in C_i^*$. The lemma already follows when $C_i^*$ does not contain an atom. When $C_i^*$ contains an atom, it follows from \cref{lem:opt-strong-size}, which shows the atom is the major part of $C_i^*$.
\end{proof}

The following lemma shows that it suffices for us to reduce the potential improvement $\Delta$.

\begin{lemma}[Similar to \cite{CLMTYZ24} Lemma 49]
\label{lem:Delta}
    Let $\calC$ be a clustering. Then $\calC$ is an $O(\eps)$-good local optimum (with respect to some $D$) if $\Delta\le\eps|D|$.
\end{lemma}%

\begin{proof}
    When $\Delta\le\eps|D|$, we have:
    \begin{align*}
        \sum_{C\in\calC^*} ( \cost(\calC) - \cost(\calC + C) )
        &
        \le
        \sum_i(\Delta_i + \eps^3 d^{adm}(C_i^*)) \tag{By definition} \\
        &
        = \Delta + 2\eps^3|E^{adm}| \\
        &
        \le \Delta + O(\eps|D|) \tag{\cref{thm:preclustering}(4)} \\
        &
        = O(\eps|D|).
    \end{align*}
\end{proof}

Finally, a key property we will use to improve the local search is that, for a vertex $v$, the (both positive and negative) improvement we get by inserting a cluster containing $v$ to our clustering can be bounded by $O(d(v)^2)$. In particular, it means $\Delta_i=O(|C_i^*|^2)$ since the vertices in $C_i^*$ have degrees of $O(|C_i^*|)$ by \cref{thm:preclustering}(6). So, since a small cluster only contributes a small proportion, intuitively we can use Chernoff bound to show we can hit a constant fraction of them with high probability. On the other hand, we must increase the hitting probability for large clusters.

\begin{lemma}
\label{lem:low-cost-cluster}
    Let $\calC$ be a clustering.
    For any cluster $C$ and any vertex $v\in C$, we have $| \cost(\calC) - \cost(\calC + C) | = O(d(v)^2)$
\end{lemma}%

\begin{proof}
    By the degree-similarity (Theorem \ref{thm:preclustering}(3)), all vertices in $C$ have degrees of $O(d(v))$. Then, for any $u\in C$, $|C(u)|=O(d(v))$, so moving $u$ from $C(u)$ to $C$ will only change the cost by $O(d(v))$. On the other hand, $|C|=O(d(v))$. So, moving these vertices to $C$ one by one, the total change is $O(d(v)^2)$.
\end{proof}

\subsection{Finding local optimum via local search}
\label{sec:local-search-implementation}

In this section, we show how to implement our local search algorithm in $O(|D|)$ time. Recall that our goal is to find an $O(\eps)$-good local optimum.

First, note that in the optimal clustering $\calC^*$, some clusters $C_i^*$ might be a singleton cluster. These singleton clusters have some different behavior with non-singleton clusters, e.g. \cref{lem:hit,thm:preclustering}(6). Although intuitively detecting singletons is an easier thing as we just need to check for each vertex whether it should leave its cluster, we do need to deal with it in the different way with non-single clusters. For simplicity, we will present our algorithm and analysis assuming there are no singleton vertices in the optimal clustering, and afterwards show how to easily incorporate singletons into our consideration.

Consider the input $(\mathcal{C},D)$ where each cluster in $\mathcal{C}$ is a singleton or an atom. For any cluster $C\in\mathcal{C}$, if the number of edges between $C$ and $V\setminus C$ is less than $\frac13|C|$, then no other vertices can be admissible to vertices in $C$ and hence we don't need to touch $C$ in the local search. So, in the remaining part of this section, we will ignore all clusters $C$ which has less than $\frac13|C|$ incident edges in $D$, and only consider the remaining part of the graph.
Let $n_D$ denote the number of remaining vertices.
Note that $n_D=O(|D|)$ and they can be easily found in $O(|D|)$ time. 

\begin{algorithm}[H]
    \caption{Local Search($\mathcal{C}, D$)}
    \label{alg:local-search-new}
	\begin{algorithmic}[1]
        \State $(\mathcal{C}', D') \gets (\mathcal{C}, D)$.
        \State Fix some sufficiently small constant $\eps'\ll\eps$.
        \While{$\mathcal{C}'$ improves by at least $\eps'|D|$}
            \For{$i=1,\dots,|D|$}
                \State Sample a pivot $p$ such that each vertex $v$ is sampled with probability $\Theta\left(\frac{1}{|D|\log^2 d(v)}\right)$.
                \State Compute the admissible neighborhood of $p$ with $t = \Theta(\log d(p))$ (\cref{thm:preclustering}).
                \State Generate a cluster $C$ with $\gamma=\Theta(d(p)\log^2 d(p))$ (Lemma \ref{lem:est-imp}). 
                \label{line:generate}
                \State Improve $\calC'$ to $\calC'+C$ if the estimated improvement is positive.
            \EndFor
            \State Compute $|D'|$ from $\mathcal{C}'$ and $(\mathcal{C}, D)$. If $|D'|>|D|$, stop and return $(\mathcal{C}, D)$.
        \EndWhile
        \State Return $(\mathcal{C}', D')$.
	\end{algorithmic}
\end{algorithm}

As said before, we shall increase the probability of hitting large clusters. We do it in a smooth way that a vertex $v$ with degree $d(v)$ is sampled with probability $\Theta\left(\frac{1}{|D|\log^2 d(v)}\right)$, increased from $\Theta\left(\frac{1}{|D| d(v)}\right)$. 
We then prove that this gives us an exponentially small error probability. 

\begin{lemma}
\label{lem:ls-correct}
    When the Local Search algorithm terminates, with probability $p$ where 
    \[
        p=1-\exp\left(-\Omega\left(\frac{\sqrt{|D|}}{\log^3|D|}\right)\right) ~,
    \]
    the final clustering $(\mathcal{C}',D')$ is an $O(\eps)$-good local optimum.
\end{lemma}

\begin{proof}
    Consider each iteration of the while-loop. We are going to bound the probability that $\calC'$ is not an $O(\eps)$-good local optimum in the beginning and we have not improved it by at least $\eps'|D|$.
    Suppose $\calC'$ is not an $O(\eps)$-good local optimum in the beginning, then $\Delta \ge \eps|D|$ in the beginning by Lemma \ref{lem:Delta}.

    When we compute the admissible neighborhood of $p$ by \cref{thm:preclustering}(9), we set $t=\Theta(\log d(p))$. This is more than enough to get a constant probability of success.
    When we compute the estimated improvements by Lemma \ref{lem:est-imp}, we set $\gamma=\Theta(d(p)\log^2 d(p))$. With constant probability, our estimate is large enough and shows the desired improvement. %

    We then show that the total positive improvement we get is large enough. 
    As we are analyzing the behavior of the algorithm over multiple steps, we have to be careful with our notation $\Delta$ and $\Delta_i$. 
    Let $\calC^{(t)}$ be the current clustering before step $t$. 
    We will use $\Delta^{(t)}$ and $\Delta_i^{(t)}$ for the total potential improvement and the potential improvement from $C^*_i$ at step $t$, respectively. Formally speaking, 
    \[
        \Delta_i^{(t)} = \max(\cost(\calC^{(t)}) - \cost(\calC^{(t)} + C^*_i) - \eps^3 d^{\adm}(C^*_i), 0),~ \Delta^{(t)} = \sum_i \Delta_i^{(t)}. 
    \]
    
    Imagine we have $\log {n_D}$ lists $L_i$ ($1\le i\le \log {n_D}$), where $L_i$ is the list of all improvements over the time when we hit a cluster of size in $(2^{i - 1}, 2^i]$. 
    We say a cluster is in level-$i$ if its size is in $(2^{i - 1}, 2^i]$. 
    For all improvement in the same list, they have the similar upper bound, therefore we can get an exponential concentration bound via Chernoff-like inequalities.
    When the number of improvement within this list is large enough, we get a desired bound. 

    Initially all of $L_i$'s are empty. 
    In each of the $|D|$ rounds, we do one of the following (based on the results of this round):    
    \begin{itemize}
        \item If the total positive improvement over previous rounds is at least $\eps'|D|$ , do nothing.
        \item Otherwise, we know $\Delta^{(t)} \ge \eps|D|-\eps'|D| = (\eps-\eps')|D|$, since it is not an $O(\eps-\eps')$-good clustering.
        We partition $\Delta_i$'s by logarithmic sizes of the clusters as follows.
        For each level $i$, we will use $\Delta_{(i)}^{(t)}$ to denote the total improvement for all clusters of $\opt$ in this level. %
        By definition,
        \[
            \Delta_{(i)}^{(t)}=\sum_{j:C_j^*\text{ is level-}i} \Delta_j^{(t)}. 
        \]
        Let $\Delta'^{(t)}=\frac{\Delta^{(t)}}{2\log {n_D}}$ and $k_i^{(t)}=\left\lfloor\frac{\Delta_{(i)}^{(t)}}{\Delta'^{(t)}}\right\rfloor$. Since $\sum_{i=1}^{\log {n_D}}\Delta_{(i)}^{(t)}=\Delta^{(t)}$, we have
        \begin{equation}
        \label{eqn:sum-k}
            \log {n_D} \le \sum_{i=1}^{\log {n_D}}k_i^{(t)} \le 2\log {n_D} ~.
        \end{equation}
        Intuitively, $k_i^{(t)}$ measures how much we would expect to improve at time $t$ within level $i$ clusters. 
        Therefore, for every $i \in \{1, \dots, \log {n_D}\}$, if $k_i^{(t)}>0$, we insert a random variable to $L_i$, which is the (actual) positive improvement we get in this round by hitting level-$i$ clusters (which means it is $0$ if we do not hit level-$i$ clusters or if we get a negative improvement). We call it a normal element. Then we insert $k_i^{(t)}-1$ zeros to $L_i$, make sure that every $L_i$'s has similar expected average value. We call them zero elements. Note that we insert $\log {n_D}$ to $2\log {n_D}$ elements to the lists by \cref{eqn:sum-k}. 
    \end{itemize}
    In the end, if a list has more than $|D|\log {n_D}$ elements, we truncate it so that it has exactly $|D|\log {n_D}$ elements. If a list has less than $|D|\log {n_D}$ elements, we insert special elements with infinite weights until it has $|D|\log {n_D}$ elements. Then, for each $L_i$, we split it into $2\log {n_D}$ segments $L_{i,j}$ ($1\le j\le 2\log {n_D}$), each of length $|D|/2$.

    \begin{claim}
        If the sum of the elements in each $L_{i,j}$ is at least $\frac{\eps'|D|}{\log {n_D}}$, then the total positive improvement to $\calC'$ is at least $\eps'|D|$.
    \end{claim}

    \begin{proof}
        If there are at least $\log {n_D}$ lists which do not contain any special elements, then the sum of the elements in these lists is a lower bound of the total positive improvement we get, and it is at least $\eps'|D|$. 
        Otherwise, we will have at most $\log {n_D}$ lists without any special elements. 
        Since for each $i$, there is at most one list $L_{i,j}$ containing both normal/zero elements and special elements, there will be at most $2\log {n_D}$ lists containing any normal/zero segments. 
        Therefore, there are less than $|D|\log {n_D}$ normal/zero elements in total, which means we have reached the first bullet above, then the total positive improvement is again at least $\eps'|D|$.
    \end{proof}

    Then we show that for each list, the sum is large enough with high probability.

    \begin{claim}
        The sum of the elements in each $L_{i,j}$ is at least $\frac{\eps'|D|}{\log {n_D}}$ with probability $p$ where 
        \[
            p = 1 - \exp\left(-\Omega\left(\frac{\sqrt{|D|}}{\log^3|D|}\right)\right).
        \]
    \end{claim}

    \begin{proof}
        Consider the $k$-th element $X_k$ in $L_{i,j}$. Suppose $X_k$ is a normal element followed by $l$ zero elements. Although the distribution of $X_k$ depends on previous elements in $L_{i,j}$, given any realization of them, we know
        \[\E[X_k]\ge\Omega\left(\frac{2^i}{i^2|D|}\right)\cdot(l+1)\Delta'^{(t)}=\Omega\left(\frac{2^i}{i^2}\cdot\frac{(l+1)}{\log {n_D}}\right),\]
        since:
        \begin{itemize}
            \item The total potential improvement from level-$i$ clusters in that round is at least $(l+1)\Delta'^{(t)}$;
            \item Each level-$i$ cluster will be hit with probability $\Omega\left(\frac{2^i}{i^2|D|}\right)$ (Lemma~\ref{lem:hit});
            \item The admissible neighborhood is correctly computed with constant probability;
            \item The estimated improvement is as desired with constant probability.
        \end{itemize}
        Therefore, if we look at a normal element $X_k$ and all following zero elements together, their average value in expectation is $\Omega\left(\frac{2^i}{i^2\log {n_D}}\right)$. %
        Without loss of generality, we treat special elements as $\Theta\left(\frac{2^i\Delta'^{(t)}}{i^2|D|}\right)$, which only decreases the sum.
        Then, since $L_{i,j}$ can only have at most $2\log {n_D}$ leading zero elements (coming from the end of $L_{i,j-1}$), we have
        \begin{equation}
        \label{eqn:EX}
            \E\left[X\right] \ge (|D|/2-2\log {n_D})\cdot\Omega\left(\frac{2^i}{i^2}\cdot\frac{1}{\log {n_D}}\right)=\Omega\left(\frac{2^i}{i^2}\cdot\frac{|D|}{\log {n_D}}\right)
        \end{equation}
        where $X=\sum_{k=1}^{|D|/2} X_k$.
        We will then look at bounding $X_k\le y_k$ for some value $y_k$. Note that a level-$i$ cluster in $\calC^*$ only contains vertices of degrees $O(2^i)$ by \cref{thm:preclustering}(6). Then each $y_k$ is at most $O(2^{2i})$ by \cref{lem:low-cost-cluster}, and it is also trivially bounded by $\Delta=O(|D|)$, so we have 
        \begin{equation}
        \label{eqn:Xk}
            y_k=O\left(\min\{2^{2i},|D|\}\right)=O\left(2^i\sqrt{|D|}\right).
        \end{equation}
        Comparing \cref{eqn:EX,eqn:Xk}, we can get
        \[y_k=O\left(\frac{\log^3|D|}{\sqrt{|D|}}\right)\cdot\E[X].\]
        Then the claim follows from a Chernoff bound.
    \end{proof}

    Combining these claims, we know that with probability $1-\exp\left(-\Omega\left(\sqrt{|D|}/\log^3|D|\right)\right)$, the total positive improvement to $\calC'$ is at least $\eps'|D|$. On the other hand, by \cref{lem:low-cost-cluster}, the negative improvement of hitting a level-$i$ cluster in $\calC^*$ is also bounded by $O(2^{2i})$. Then, by \cref{lem:est-imp}, the expected negative improvement is $O(2^{2i}e^{-2^i i^2})$. Following an analogous argument to above, the total negative improvement we get can be bounded by $\eps''|D|$ for an arbitrarily small constant $\eps''$, with probability $1-\exp\left(-\Omega\left(\sqrt{|D|}/\log^3|D|\right)\right)$. Setting $\eps''\ll\eps'$, the total improvement would be $(\eps'-\eps'')|D|=\Omega(|D|)$. The lemma then follows from a union bound over all iterations, as the number of iterations is a constant.
\end{proof}

Then we prove the time complexity of our algorithm.

\begin{lemma}
\label{lem:ls-time}
    The Local Search algorithm runs in $O(|D|)$ time in expectation and also with probability $1-\exp\left(-\Omega\left(\frac{\sqrt{|D|}}{\log^3|D|}\right)\right)$.
\end{lemma}

\begin{proof}
    Since the initial cost of our clustering is $|D|$ and in each iteration of the while-loop, the cost will be improved by $\Omega(|D|)$, the number of iterations is bounded by a constant. Consider each iteration:
    \begin{itemize}
        \item There are $|D|$ local search rounds. In each round, each singleton $v$ will be sampled with probability $\Theta\left(\frac{1}{|D|\log^2 d(v)}\right)$.
        If it is sampled, we spend $O(d(v)\log^2 d(v))$ time to compute the admissible neighborhood by \cref{thm:preclustering}(9), and spend $O(d(v)\log^2 d(v))$ time to generate the cluster by \cref{lem:est-imp}. Also, for a singleton $v$, we have $d(v)=d_D(v)$.

        Each atom $K$ will be sampled with probability $\Theta\left(\frac{|K|}{|D|\log^2|K|}\right)$ since any $v\in K$ satisfies $d(v)=\Theta(K)$.
        If it is sampled, we need to spend $O(d^{adm}(K))$ time to compute the admissible neighborhood by \cref{thm:preclustering}(7), and spend $O(d^{adm}(K)\log^2|K|)$ time to generate the cluster by \cref{lem:est-imp}.
        Note that $d^{adm}(K)=\frac{d_D(K)}{|K|}$ by \cref{thm:preclustering}(3).
            
        The expected time complexity is 
        \begin{align*}
            & \sum_{\text{singleton }v} \Theta\left(\frac{1}{|D|\log^2 d(v)}\right)\cdot O(d(v)\log^2 d(v)) \\ & + \sum_{\text{atom }K} \Theta\left(\frac{|K|}{|D|\log^2|K|}\right)\cdot O(d^{adm}(K)\log^2|K|) \\
            =\ & \sum_{\text{singleton }v} O\left(\frac{d_D(v)}{|D|}\right) + \sum_{\text{atom }K} O\left(\frac{d_D(K)}{|D|}\right) \\
            =\ & O\left(1\right).
        \end{align*}

        So the expected time complexity for $|D|$ rounds is $O(|D|)$.
        
        \item Since $\mathcal{C}$ is a collection of atoms and singletons which will never be broken in $\calC'$, it's straight-forward to compute $|D'|$ (and also $D'$ if $|D'|=O(|D|)$) from $\mathcal{C}'$ and $(\mathcal{C}, D)$ in $O(|D|)$ time. %
    \end{itemize} 
    So the total time complexity is $O(|D|)$ in expectation. For the high-probability bound, notice that:
    \begin{itemize}
        \item When we sample a singleton $v$ as the pivot, we only need to try to generate the cluster if $d(v)=O(\sqrt{|D|})$, so the time complexity in this single step is at most $O(\sqrt{|D|}\log^2|D|)$.
        \item When we sample (a vertex in) an atom $K$ as the pivot, we only need to try to generate the cluster if $d^{adm}(K)=O(\sqrt{|D|})$, so the time complexity in this single step is at most $O(\sqrt{|D|}\log^2|D|)$.
    \end{itemize}
    So the time complexity in each round can be upper bounded by $O(\sqrt{|D|}\log^2|D|)$, then high-probability claim follows from a Chernoff bound.
\end{proof}

Now let's consider how to detect singleton clusters. We make the following modification to our algorithm: In each round, with probability $\frac{1}{2}$, we do the original thing stated in \cref{alg:local-search-new}; Otherwise, we sample each singleton $v$ with probability $\frac{d(v)}{|D|\log^2d(v)}$, and spend $O(\log^2 d(v))$ time to estimate whether we should make it a singleton cluster. It's clear that inserting a singleton cluster $\{v\}$ will only change the cost by $O(d(v))$, so in the proof of \cref{lem:ls-correct}, we can define its level by its degree (instead of its size), then the same analysis also works with singleton clusters. For the time complexity, although the sampling probability increases from $\frac{1}{|D|\log^2d(v)}$ to $\frac{d(v)}{|D|\log^2d(v)}$, the cost when it is sampled decreases from $O(d(v)\log^2 d(v))$ to $O(\log^2 d(v))$, so the same analysis also works.

Putting things together, we can prove our main theorem for local search.

\begin{theorem}
    We can find an $O(\eps)$-good local optimum with probability $p$ in worst-case $O(|D|)$ time, where $p = 1 - \exp\left(-\Omega\left(\frac{\sqrt{|D|}}{\log^3|D|}\right)\right)$. 
\end{theorem}
\begin{proof}
    We run the Local Search algorithm, and directly terminate with the old clustering $(\calC,D)$ if we reach a time limit of $\Theta(|D|)$. By an union bound over \cref{lem:ls-correct,lem:ls-time}, we succeed with probability $1-\exp\left(-\Omega\left(\frac{\sqrt{|D|}}{\log^3|D|}\right)\right)$.
\end{proof}

Combining with \cref{thm:iter-ls,thm:dynamic}, we can get a $(2-2/13+O(\eps k))$-approximate dynamic algorithm with constant update time and with exponentially high probability, where $k$ is the approximation factor of the clustering from the last epoch.

The only remaining thing to prove \cref{thm:local-search} is to remove the factor $k$.
It can be removed by repeating the algorithm (including the preclustering part) if we significantly improve our solution. 
That is, for some small constant $\delta>\eps$, if $|D'|<\delta|D|$, then we repeat the algorithm with input $(\calC',D')$, and to have a higher probability guarantee, we multiply the number of rounds in the new repetition by $(|D|/|D'|)^{0.6}$. Note that the exponent of the error probability simply scales with the number of rounds, so the total error probability is still bounded by $\exp\left(-\Omega\left(\frac{\sqrt{|D|}}{\log^3|D|}\right)\right)$ by a union bound. On the other hand, the time complexity is still $O(D)$ since $\sum_{i=0}^{\infty}\delta^{0.4i}|D|=O(|D|)$.
Finally, the hidden constant in the $\Omega$ notation can be moved to the time complexity by having more rounds, then the error probability can be bounded by $\exp\left(-\frac{\sqrt{|D|}}{\log^3|D|}\right)$ in $O(|D|)$ time, finishing the proof of \cref{thm:local-search}.

%% file: clusterLP.tex
\section{Dynamic Cluster LP}
\label{sec:LP}

Recently, the authors of \cite{cao2025fastLP} showed a sublinear time algorithm with an expected approximation ratio of $\pureclusterlpratio$ for correlation clustering. 
They managed to solve the standard cluster LP approximately using a similar local search algorithm as \cite{CLMTYZ24} in sublinear time, and also improved the time of the two rounding algorithms in \cite{cao2024understanding} to sublinear due to the sparsity of the output of their LP solver.

The goal of this section is to implement the algorithm from \cite{cao2025fastLP} in $O(|D|)$ time in cluster representation. 
The original algorithm involves two parts, an approximate cluster LP solver and two rounding algorithms. 
The $O(|D|)$ time implementations of the LP solver and the two rounding algorithms are presented in \cref{subsec:solve-lp} and \cref{subsec:rounding}, respectively.

\subsection{Solving the cluster LP}
\label{subsec:solve-lp}

Given a graph $G = (V, E)$, the cluster LP is the natural relaxation of integer programming of the correlation clustering problem, where $z_S$ indicates whether the set $S$ is chosen as a cluster, and $x_{uv}$ indicates if the two vertices $u, v$ are in different clusters. 
\begin{definition}[cluster LP]
    \label{LP:clusterlp}
    \begin{align*}
        \min \quad & \obj(z) = \sum_{(u,v) \in E}x_{uv} + \sum_{(u,v) \notin E}(1 - x_{uv}) \\
        \text{s.t.} \quad & x_{uv} = 1 - \sum_{S \supseteq \{u, v\}}z_S & \forall (u, v) \in \binom{|V|}{2} \\
        \quad & \sum_{S \ni v}z_S = 1 & \forall v \in V \\
        & z_S \ge 0 & \forall S \subseteq V
    \end{align*}
\end{definition}

We follow the same algorithmic framework to solve the \hyperref[LP:clusterlp]{\color{black}cluster LP} as \cite{cao2025fastLP}. 
Followed by \cref{sec:preclustering}, we will assume a preclustering $(\mathcal{K} = \calC, E^{\adm})$ and the corresponding cluster representation $(\calC, D)$ of the graph are given as input, and the output will be an approximate fractional solution. 
Like \cite{cao2025fastLP}, we first turn the \hyperref[LP:clusterlp]{\color{black}cluster LP} into a covering LP such that the optimal solution remains the same and the objective value only increases by $2|D|$. 
Given the cluster representation $(\calC, D)$ of the graph with respect to the preclustering, 
for any vertex $v \in V$, let $\dc(v)$ be the number of edges in $D$ incident to $v$, and $\dc(C) = \sum_{v \in C} \dc(v)$ for any cluster $C \subseteq V$.
For any cluster $C$, we define the cost of it as
\[
    \cost(C) = \frac{1}{2}|\{(u, v) \in E: u \in C, v \in V \setminus C\}| + |\{(u, v) \notin E: u, v \in C\}| ~.
\]
Then for any clustering $\calC$, $\cost(\calC) = \sum_{C \in \calC}\cost(C)$. 
To transform a cluster LP into an equivalent covering LP, we add a vertex cost $\dc$ for each vertex, and set the new cost for each cluster $C$ as $\covers(C) = \cost(C) + \dc(C)$. 

\begin{definition}[covering cluster LP from \cite{cao2025fastLP}]
    \label{LP:coverclusterlp}
    \begin{align*}
        \min \quad & \obj(z) = \sum_{S \subseteq V}\covers(S) \cdot z_S \\
        \text{s.t.} \quad & \sum_{S \ni v}z_S \ge 1 & \forall v \in V \\
        & z_S \ge 0 & \forall S \subseteq V
    \end{align*}
\end{definition}
For any $(z_S)_{S \subseteq V}$, we denote the list of its non-zero entries as $\supp(z)$. 

The next lemma from \cite{cao2025fastLP} shows how to solve the original cluster LP via the covering cluster LP. 
\begin{lemma}[Lemma 11 from \cite{cao2025fastLP}]
\label{lem:cover-to-cluster}
Let $\epsilon, \covereps > 0$ be sufficiently small constants, and $c = \lceil \frac{1}{\covereps} \rceil$. Given a constant $\tmwu \in \mathbb{N}$ and a solution $z$ to the \hyperref[LP:coverclusterlp]{\color{black}covering cluster LP} where,
\begin{itemize}
    \item $z_S \geq \frac{1}{\tmwu}$ for each $S \in \supp(z)$,
    \item all $S \in \supp(z)$ do not split atoms,
    \item for each vertex $v$ the number of sets $S \in \supp(z)$ with $v \in S$ is at most $\tmwu$.
 \end{itemize}
We can then in $O(n)$ time find a solution $\widetilde z$ to the \hyperref[LP:clusterlp]{\color{black}cluster LP} where $\widetilde z_S \geq \frac{1}{c \tmwu}$ for all $S \in \supp(z)$ and $\covers(\widetilde z) \leq (1 + \covereps) \covers(z)$. 
\end{lemma}

To solve the \hyperref[LP:coverclusterlp]{\color{black}covering cluster LP} efficiently, \cite{cao2025fastLP} use the multiplicative-weights-update algorithm, and its pseudocode is given in \cref{alg:mw}.
We will follow this framework of their algorithm, and the only difference from their algorithm is the implementation of line~\ref{line:findPoint}, which is a task similar to the generateCluster procedure in the local search algorithm of \cite{CLMTYZ24}. 

In a multiplicative-weights-update algorithm for solving linear programming, we will maintain a weight $p$ over all the inequalities. 
Specific to \hyperref[LP:coverclusterlp]{\color{black}covering cluster LP}, all the inequalities are associated with vertices, therefore $p: V \to \mathrm{R}^+$ will be a weight for vertices. 
For every $S \subseteq V$, we define $p(S) = \sum_{v \in S}p(v)$ as the weight of the set $S$. 

\begin{algorithm}[htb]
	\caption{MWU algorithm for the \coverClusterLP from \cite{cao2025fastLP}}
	\label{alg:mw}
	\begin{algorithmic}[1]
    \State Initialize the weights $w_v^{(1)} = \dc(v)$ for each vertex $v \in V$.    \For{$t=1,\dots,\tmwu$}
        \State Normalize the weights $p^{(t)} = \frac{w^{(t)}}{\sum_v w_v^{(t)}}$. 
        \State Aggregate all constraints into single constraint: $\sum_{S} p^{(t)}(S) \cdot z_S = \sum_{v} p^{(t)}_v \left (\sum_{S: v \in S} z_S \right) \geq 1$. 
        \State Find the point $z^{(t)} \in \left[0, 1/{\covereps}\right]^{2^V}$ that, \label{line:findPoint}
        \begin{itemize}
            \setlength\itemsep{0.0em}
            \item satisfies the single constraint $\sum_{S} \pt(S) \cdot z^{(t)}_S \geq 1$,
            \item has objective value $\covers(z^{(t)}) \leq (1+5\covereps)~\covers(\opt) + \epsilon|D|$,
            \item does not split atoms (i.e., if $\zt_S > 0$ then $K(v) \subseteq S$ for all vertices $v \in S$), and
            \item has disjoint support (i.e., if $\zt_S, \zt_{T} > 0$ for two distinct $S,T \subseteq V$, then $S \cap T = \emptyset$).
        \end{itemize}
        \State The cost of a constraint corresponds to the margin by which it is satisfied or violated, $m_v^{(t)} = \sum_{S: v \in S} \zt_{S} - 1$.
        \State Update the weights $w^{(t+1)}_v =  w^{(t)}_v e^{-\covereps^3 m^{(t)}_v}$.\label{line:update-rule}
    \EndFor
    \State Let $\hat z$ be the average $\frac{1}{\tmwu} \sum_{t=1}^{\tmwu} \zt$.
    \For{each $v$ with $\sum_{S \supseteq K(v)} \hat z_S \leq 1 - 2\covereps$} \label{line:loop_feasible}
    \State Set the atom entry $\hat z_{K(v)}$ to 1.
    \EndFor
       \State $\zmwu \gets \frac{\hat z}{1-2\covereps}$.
 \State \Return $\zmwu$.
	\end{algorithmic}
\end{algorithm}

We will start with a variation of the polynomial time implementation of line~\ref{line:findPoint} (Algorithm 3 in \cite{cao2025fastLP}), and turn it into a $O(|D|)$ implementation using similar tools we developed for the local search algorithm in the \cref{sec:local-search}. 
The pseudocode for this algorithm is given in \cref{alg:disjointfamily-dynamic}. 
The implementation of line~\ref{line:findGoodCluster} is captured by the following lemma, modified from \cite{cao2025fastLP}. 

\begin{algorithm}[htb]
	\caption{Implementation of line~\ref{line:findPoint} in \cref{alg:mw}, modified from Algorithm 6 in \cite{cao2025fastLP}}
	\label{alg:disjointfamily-dynamic}
	\begin{algorithmic}[1]
        \State Let $R$ be the guess for $\covers(\opt)$ such that $R \in [\covers(\opt) , (1 +\covereps) \covers(\opt))$.
        \State $\hat p \gets p, \mathcal{F} \gets \emptyset, \hat{V} \gets V$
        \For{$v \in V$}
            \State If $\frac{\covers(K(v))}{p(K(v))} \leq (1 + 6\gamma)R$, add $K(v)$ to $\mathcal{F}$ and set $\hat p_w = 0$ for all $w \in K(v)$. \label{line:addlargepnodes}
            \State If $p_v \leq \frac{\covereps \dc(v)}{4\dc(V)}$, set $\hat p_w = 0$ for all $w \in K(v)$. \label{line:removesmallpnodes}
        \EndFor
        \For{$i=1,\dots,\Theta(n)$} \label{line:loop_family-poly}
            \State Sampling a vertex $v$ with probability $\Theta\left(\frac{1}{n\log^2 d(v)}\right)$
            \label{line:sample}
            \State Find a new small ratio cluster $C$ that $K(v) \subseteq C \subseteq \Ncand(v)$, with vertex weights $\hat p$ and target ratio $(1+3\covereps)R$, by \cref{lem:est-imp-LP} with $t = \Theta(\dadm(v)\log^2\dadm(v))$. \label{line:findGoodCluster}
            \State Skip the rest of this loop if we failed on the previous step
            \State Remove all vertices $u$'s in $C$ with $\hatp(u) = 0$
            \State Add $C$ to $\mathcal{F}$, set $\hat p_u$ to $0$ for all $u \in C$ and remove vertices in $C$ from $\hat{V}$. 
            \State Stop if $p(\mathcal{F}) > \covereps$. 
        \EndFor
        \If{$p(\mathcal{F}) > \covereps$}
            \State Set $(z_S)_{S \subseteq V} = \Vec{0}$
            \State For each $S \in \mathcal{F}$, set $z_S = \frac{1}{p(\mathcal{F})}$
            \State \Return $z$
        \Else
            \State \Return failure 
        \EndIf
	\end{algorithmic}
\end{algorithm}

\begin{lemma}[Improving Lemma 49 from \cite{cao2024understanding}]
    \label{lem:est-imp-LP}
    For any $t \ge d^{\adm}\log^2 d^{\adm}(v)$, vertex weights $\hat{p}$, target ratio $R$ and vertex $r \in V$, assume there exists a cluster $C$ such that
    \begin{itemize}
        \item $K(r) \subseteq C \subseteq \Ncand(r)$; 
        \footnote{Recall that $\Ncand(r) = K(r) \cup (\Nadm(r) \setminus (\cup_{K \text{ is an atom}}K))$, as introduced in \cref{sec:local-search}. }
        \item $\covers(C) + \covereps^2 d^{\adm}(C) \le R \cdot \hat{p}(C)$; 
        \item For all $v \in C$, $\covers(\{v\}) > R \cdot \hat{p}(v)$. 
    \end{itemize}
    Then in $O(t)$ time, we can generate a cluster $C$ that $K(r) \subseteq C \subseteq \Ncand(r)$ such that $\covers(C) \le R \cdot \hat{p}(C)$ with constant probability, and with probability at most $\exp(-t)$, we report a cluster that $\covers(C) > R \cdot \hat{p}(C) + \covereps^2 d^{\adm}(C)$. 
\end{lemma}

\begin{proof}
    Every part else is already contained in the original lemma, except for the bound of probability that a cluster $C$ with $\covers(C) > R \cdot \hatp(C) + \covereps^2 \dadm(C)$ is generated. 
    To achieve this, we can improve the concentration of estimating $\covers(C)$ using a similar approach to \cref{lem:est-imp}. 
\end{proof}

We need an invariant about the vertex weights $p$ throughout \cref{alg:mw}, formulated by the following lemma. 

\begin{lemma}[Similar to Lemma 34 in \cite{cao2025fastLP}]
    \label{lem:invariant}
    At any step of \cref{alg:mw}, for every vertex $v$, 
    \[
        \frac{\covereps\dc(v)}{16\dc(V)} \le p_v \le \frac{16\dc(v)}{\dc(V)}. 
    \]
\end{lemma}

\begin{proof}
    It is guaranteed by the property of multiplicative update method in \cref{alg:mw} and the special handling of vertices with large or small ratio, i.e., line~\ref{line:addlargepnodes} and \ref{line:removesmallpnodes} in \cref{alg:disjointfamily-dynamic}. 
    The proof is the same as the one of Lemma 34 in \cite{cao2025fastLP}, and we will not repeat it here. 
\end{proof}

Next, we will show that \cref{alg:disjointfamily-dynamic} succeeds with high probability and the expected running time is $O(|D|)$. 
The analysis is similar to the proof of \cref{lem:ls-correct} and \cref{lem:ls-time}. 

\begin{lemma}
    \cref{alg:disjointfamily-dynamic} runs in $O(|D|)$ time and succeeds with probability at least $1 - \exp\left(-\tilde{\Omega}(\sqrt{|D|})\right)$. 
\end{lemma}

\begin{proof}
    Let $\mathrm{Small}$ be the set of vertices whose $\hatp$ weight is set to 0 at line~\ref{line:removesmallpnodes}. 
    According to the proof of Lemma 29 in \cite{cao2025fastLP}, we have $p(\mathrm{Small}) \le \frac{\gamma}{4}$. 
    Note that at the beginning of each loop, $p(\mathcal{F}) + p(\mathrm{Small}) + \hatp(V) = 1$. 
    
    Next we will show that if the current $\mathcal{F}$ does not take up to $\covereps$ fraction of the total weight, there will be a large fraction of weight in the graph, such that we can succeed in finding a good ratio cluster if we hit this part. 
    \begin{claim}
        Let $\mathcal{F}' = \{C^*_i \in \opt: \covers(C^*_i) \le \hatp(C^*_i)(1 + 3\covereps)\covers(\opt)\}$. 
        If $p(\mathcal{F}) \le \covereps \le 1/9$, we have $\hatp(\mathcal{F}') \ge \covereps$. 
    \end{claim}
    \begin{proof}
        For any $C^*_i \in \opt \setminus \mathcal{F}'$, we have 
        \[
            \covers(C^*_i) > \hatp(C^*_i)(1 + 3\covereps)\covers(\opt). 
        \] 
        Summing up over all sets in $\opt \setminus \mathcal{F}'$, we have
        \[
            \covers(\opt) > \sum_{C^*_i \in \opt \setminus \mathcal{F}'}\covers(C^*_i) = (\hatp(\opt) - \hatp(\mathcal{F}'))(1 + 3\covereps)\covers(\opt).
        \]
        Since $\hatp(\opt) = 1 - p(\mathcal{F}) - p(\mathrm{Small}) \ge 1 - \frac{5\covereps}{4}$, we know $\hatp(\mathcal{F}') \ge \frac{3\covereps}{1 + 3\covereps} - \frac{5}{4}\covereps \ge \covereps$.
    \end{proof}
    
    Now we will analyze the decrease of $\hatp(V)$ in each step, which also equals to the increase of $p(\mathcal{F})$. 
    We first show a lower bound for hitting each cluster. 

    \begin{claim}
        \label{lem:hit-clusterLP}
        For any non-singleton cluster $C^*_i \in \opt$, it will be hit with probability $\Theta\left(\frac{|C^*_i|}{n\log^2 |C^*_i|}\right)$ at line~\ref{line:sample}. 
    \end{claim}
    
    \begin{proof}
        It directly follows from \cref{lem:hit}. 
    \end{proof}

    Next, we will show a lower bound of decrease of $\hatp(V)$ if a cluster $C^*_i \in \mathcal{F}'$ is hit. 
    
    \begin{claim}[Claim 33 from \cite{cao2025fastLP}]
        \label{clm:size}
        For any cluster $C \in \Ncand(v)$ such that 
        $\frac{\covers(C)}{\hat p(C)} \le (1+5\covereps) R$, then 
        \begin{itemize}
            \item either $|\{u \in C \mid \hat p_u > 0 \}| \ge \covereps^3 d(v)$, 
            \item or $\frac{\covers(\{u\})}{p_u} \le (1+6\covereps)R$ for some $u \in C$.
        \end{itemize}
    \end{claim}
    
    \begin{lemma}
        \label{lem:large-probability}
        Let $v$ be a vertex. 
        For any cluster $C$ such that $K(v) \subseteq C \subseteq \Ncand(v)$ and 
        $\frac{\covers(C)}{\hat p(C)} \le (1+5\covereps) R$, $\hatp(C) = \Theta(p(\Ncand(v)))$. 
    \end{lemma}

    \begin{proof}
        Let $C' = \{u \in C: \hatp(u) > 0\}$. 
        According to monotonicity of $\covers$, $\covers(C') \le \covers(C)$. 
        By definition, $\hatp(C') = \hatp(C)$, therefore
        \[
            \frac{\covers(C')}{\hatp(C')} \le \frac{\covers(C)}{\hatp(C)} \le (1 + 5\covereps) R. 
        \]
        According to \cref{clm:size}, $|C'| \ge \covereps^3 d(u)$. 

        Next we will show that $\dc(C') = \Theta(\dc(\Ncand(v)))$. 
        
        If $v$ is a singleton in the preclustering, then for any $u \in \Ncand(v)$, $\dc(u) = d(u)$ and $d(u) = \Theta(d(v))$ according to \cref{thm:preclustering}(5). 
        Therefore, 
        \[
            \dc(C') = \Theta(|C'| \cdot d(v)) = \Theta(d^2(v)) = \Theta(\dc(\Ncand(v))). 
        \]
        
        If $v$ is contained in an atom $K$, then for any $u \in \Ncand(v) \setminus K$, $|N(u) \cap K| \ge \frac{1}{3}|K|$ and $|K| = \Theta(d(v))$, therefore $\dc(C') \ge \dc(K) = \Theta(|\Ncand(v) \setminus K|\cdot d(v))$. 
        Since 
        \[
            \dc(\Ncand(v)) = \dc(\Ncand(v) \setminus K) + \dc(K)
        \]
        and
        \[
            \dc(\Ncand(v) \setminus K) = \sum_{u \in \Ncand(v) \setminus K}d(u) = \Theta(|\Ncand(u) \setminus K|\cdot d(v)). 
        \]
        Therefore, $\dc(C') = \Theta(\dc(\Ncand(v)))$. 
        
        According to \cref{lem:invariant}, we know for any $v \in C'$, $\hatp(v) = p_v = \Theta\left(\frac{\dc(v)}{\dc(V)}\right)$, therefore, 
        \[
            \hatp(C) = \hatp(C') = p(C') = \Theta\left(\frac{\dc(C')}{\dc(V)}\right) = \Theta\left(\frac{\dc(\Ncand(v))}{\dc(V)}\right) = \Theta(p(\Ncand(v))).
        \]
    \end{proof}

    By \cref{lem:est-imp-LP} and \cref{lem:large-probability}, assume $C^*_i \in \mathcal{F}'$ is hit in line~\ref{line:sample}, we can find a cluster $C$ with ratio at most $(1 + 5\covereps)R$ with constant probability, and if success, $\hatp(C) = \Theta(p(\Ncand(v))) = \Omega(\hatp(C^*_i))$. 
    In this case, $\hatp(V)$ will decrease by $\hatp(C) = \Omega(\hatp(C^*_i))$. 

    Let $\Delta_i = \hatp(C^*_i)$ if $C^*_i \in \mathcal{F}'$ and otherwise $\Delta_i = 0$. 
    Let $\Delta = \sum_i \Delta_i$. 
    We now summarize all conclusions we have so far. 
    \begin{itemize}
        \item If $p(\mathcal{F}) \le \covereps$, $\Delta \ge \covereps$; 
        \item When we hit a cluster $C^*_i$, we will increase $p(\mathcal{F})$ by $\Theta(\Delta_i)$ with constant probability, and the increment of $p(\mathcal{F})$ is $O(|C^*_i|^2/|D|)$ and takes value $O(\Delta_i)$ with constant probability;
        \item When we hit a cluster $C^*_i$, we will find a cluster $C$ that $\covers(C) - (1 + 3\covereps)R\cdot\hatp(C) = O(|C^*_i|^2)$ and $E[\covers(C) - (1 + 3\covereps)R\cdot\hatp(C)] = O(|C^*_i|^2/\exp(|C^*_i|))$; 
        \item A cluster $C^*_i$ is hit with probability $\Theta\left(\frac{|C^*_i|}{n\log^2|C^*_i|}\right)$. 
    \end{itemize}
    Using similar argument as we have in \cref{lem:ls-correct}, \cref{alg:disjointfamily-dynamic} succeeds with probability at least $1 - \exp(-\tilde{\Theta}(\sqrt{|D|}))$, and will output a family $\mathcal{F}$ that $\covers(\mathcal{F}) \le (1 + O(\covereps))R + \epsilon|D|$ for any $\epsilon > 0$. 
\end{proof}

And it's straightforward to check the expected running time of \cref{alg:disjointfamily-dynamic}. 
\begin{lemma}
    \cref{alg:disjointfamily-dynamic} runs in expected $O(|D|)$ time. 
\end{lemma}

Therefore, we can conclude the part of solving cluster LP in $O(|D|)$ time. 
\begin{theorem}
    \label{thm:fast-cluster-LP-dynamic}
    Given a preclustering $(\mathcal{K} = \mathcal{C}, E^{\adm})$, \cref{alg:mw} runs in expected $O(|D|)$ time.
    For any $\epsilon > 0$, with probability at least $1 - \exp(-\tilde{\Theta}(\sqrt{|D|}))$, \cref{alg:mw} return a $(1 + O(\covereps))$-approximate solution $z$ to cluster LP, stored by a list of non-zero entries, such that
    \begin{enumerate}
        \item $|\supp(z)| \le |V| / \delta$;
        \item For all $S \in \supp(z)$, $z_S \ge \delta$ for some constant $\delta$; 
        \item For each $v \in V$, there are at most $1/\delta$ sets $S$ that $v \in S$ and $S \in \supp(z)$; 
        \item All $S \in \supp(z)$ does not split or join atoms; 
        \item $\obj(z) \le \cost(\opt) + \epsilon|D|$. 
    \end{enumerate}
    where $\delta \in (0,1)$ is a constant. 
\end{theorem}

\subsection{Rounding the cluster LP}
\label{subsec:rounding}

In this subsection, we will show how to implement the two rounding algorithms from \cite{cao2024understanding} in $O(|D|)$ time in cluster representation. 
The input is the graph representation and a fractional solution to the \hyperref[LP:clusterlp]{\color{black}cluster LP} produced by \cref{thm:fast-cluster-LP-dynamic}, and the output will be the best of the two rounded clusterings and its corresponding cluster representation. 

First we look at the cluster-based rounding from \cite{cao2025fastLP}, restated in \cref{alg:sublinear-cluster-based}. 
Similar to \cref{alg:pivot}, we first need to contract all the inactive vertices in any active atoms, as they will always be in the same cluster for any $S \in \supp(z)$. 
After that, we just run \cref{alg:sublinear-cluster-based} directly, where line~\ref{lst:enumeration-cluster-based} can be implemented in $O(|D|)$ time in total as we only have $O(|D|)$ vertices after contraction and each of them is contained in at most $1/\delta$ sets in the $\supp(z)$ by \cref{thm:fast-cluster-LP-dynamic}(3). 
Therefore, the total running time of \cref{alg:sublinear-cluster-based} is $O(|D|)$, as desired. 

\begin{algorithm}[htb]
    \caption{Cluster-Based Rounding from \cite{cao2025fastLP}}
    \label{alg:sublinear-cluster-based}
    \begin{algorithmic}[1]
        \For{$S \in \supp(z)$}
         \State $k_S
    =  \lfloor\frac{n^c}{z_s}\log \frac{1}{p_s} \rfloor$, where $p_s$
    is uniformly chosen from $(0, 1)$ 
      \EndFor
        \For{$v \in V$}
        \State $k_v = \min \{ k_S \mid S \ni v\}$ \label{lst:enumeration-cluster-based}
        \EndFor
        \State Put all nodes with same $k_v$ value into same cluster and return
    \end{algorithmic}    
\end{algorithm}

Now we look at the pivot-based rounding from \cite{cao2025fastLP}, restated in \cref{alg:correlated-rounding}. 

Similar to \cref{alg:pivot}, we first contract all the cores, and implement line~\ref{lst:cluster-based-weighted-sampling} by weighted sampling without replacement. 
When we get a core at line~\ref{lst:cluster-based-weighted-sampling}, we just replace it by any vertex $u$ in the core. 
Since any $S \ni u$ in $\supp(z)$ contains the entire core, $x_{uv} = 0$ for any $v$ inside the core, therefore we will remove the entire core in the end of this step and cluster them together. 
By \cref{lem:weighted-sampling}, the total running time here is $O(|D|)$, if we do not need to break any core. 

Note that we might break a core and put vertices of a core into different clusters, but we will only break a core at line~\ref{lst:enumeration-pivot-based}. 
When some non-neighbors of $v$ from a core $X$ is selected, but not the whole core $X$, we have to break the core $X$. 
In this case, the cost of generated clustering increased by at least $O(|X|)$, as $X$ is fully connected in the graph. 
We can list out all remaining vertices in $O(|X|)$ time by enumerating over $X$ and remove the element for $X$ in the sets. 
For each remaining vertex, we insert it as a unit weight element into the sampling structure. 
After that, when we need samples, we will first sample if we should sample from the original elements or the inserted elements with probability proportional to the total weights, and generate the sampling accordingly in each part. 
The time cost here is charged to the increment of the cost. 
Therefore, the total extra running time for breaking cores is bounded by the cost of the final clustering. 

Let's consider the time complexity of iterating through all relevant vertices at line \ref{lst:enumeration-pivot-based}. 
For every $v \in V' \cap N^-(u)$ where $1 - x_{uv} > 0$, by \cref{thm:fast-cluster-LP-dynamic}, there exists $S \in \supp(z)$ such that $u, v \in S$ and $1 - x_{uv} \ge z_S \ge \delta$. 
Therefore, there will be at most $\obj(z)/\delta$ non-edges considered throughout the algorithm. 
And we assume $\obj(z) \le |D|$, otherwise $\calC$ will be a better solution than $z$. 
In this case, we will visit at most $|D|/\delta$ non-edges in total. 

What remains is to bound the total visits to edges. 
Here we use the same strategy as the one used in \cref{alg:pivot}. 
We enumerate the vertices in $C(v)$, which is the cluster of $v$ in $\calC$, and also the vertices that share an edge with $v$ in $D$. 
All the neighbors of $v$ are contained in this set, so we can just enumerate all neighbors of $v$ in $V'$ in this way. 
Similar to \cref{alg:pivot}, if all the neighbors of $v$ are removed after this step, then the total time cost here is $O(|D|)$. 
If a neighbor of $v$ does not go into the same cluster as $v$, the cost of the clustering will increase by at least one, and we will charge the next visit of this neighbor to the increment of the cost. 

Since the rounding algorithm can just stop when the current cost is greater than $|D|$ and report $(\calC, D)$, the total running time for visiting neighbors of $v$ are $O(|D|)$, therefore the whole algorithm can be implemented in $O(|D|)$ time. 

\begin{algorithm}[htb]
    \caption{Modified Pivot-Based Rounding with Threshold $1/3$ from \cite{cao2025fastLP}}
    \label{alg:correlated-rounding}
    \begin{algorithmic}[1]
        \State $\calC \gets \emptyset, V' \gets V$
        \While{$V' \neq \emptyset$}
            \State Randomly choose a pivot $u \in V'$ %
            \label{lst:cluster-based-weighted-sampling}
            \State $C \gets \{v \in V'\cap N^+(u): x_{uv} \leq \frac{1}{3}\}$
            \Comment{$N^+(u)$ is set of neighbors of $u$}
            \For{every $v \in V' \cap N^-(u)$}
            \Comment{$N^-(u)$ is set of non-neighbors of $u$}
            \State Independently add $v$ to $C$ with probability $1 - x_{uv}$ \label{lst:enumeration-pivot-based} \EndFor
            \State Randomly choose a set $S \ni u$, with probabilities $z_S$
            \State $C \gets C \cup (S \cap V' \cap N^+(u))$, $\calC \gets \calC \cup \{C\}$, $V' \gets V' \setminus C$
        \EndWhile
        \State \Return $\calC$
    \end{algorithmic}
    \label{algo:pivot}
\end{algorithm}

Combined with \cref{thm:fast-cluster-LP-dynamic}, we obtain the following result. 
\begin{theorem}
    \label{thm:dynamic-1.437}
    Given a graph represented by a clustering and violation pair $(\calC, D)$, in expected $O(|D|)$ time, one can generate a clustering $\calC'$ that $E[\cost(\calC')] \le \pureclusterlpratio \cdot \cost(\opt)$. 
\end{theorem}

%% file: improvedLocalSearch.tex
\section{Improved Local Search}
In this appendix, we are going to show how to improve the approximation ratio of the original approximation ratio of $1.847$ from the local search of \cite{CLMTYZ24} to an approximation ratio of $1.8$ in expectation. This is independent of the rest of the paper, and not needed for achieving the modifications of neither the local search nor solving the cluster LP. One of the central steps in this local search is a pivot step, where one creates a new solution by pivoting on three previous locally optimal solutions. We are going to replace this step with a randomized pivot inspired by the pivot of \cite{ACN08}. This can in addition also be used in the dynamic algorithm, though it does mean that we only get the approximation in expectation instead of with high probability. It is however possible to combine the two approaches, such that the final clustering has an approximation ratio of $1.8$ in expectation and $1.847$ with high probability. This is easily achievable due to the fact that both approaches construct a set of clustering where there exists at least one clustering with the desired property. So one can simply take the union of these sets.

To use the local search we need a notion for what a locally optimal cluster is. Furthermore, the local search does not actually output a locally optimal clustering, just one that is almost locally optimal.

The way we modify \cref{alg:flipping-local-search} is by modifying the computation of $\mathcal{C}_i''$:
\begin{theorem}
\label{thm:improved-approx-ratio}
  For every $0 \leq \alpha < \frac{1}{5}$
  there exists a positive integer $k$ and a real $\delta_0 > 0$
  such for every $0 < \delta \le \delta_0$, running the Iterated-flipping Local Search Algorithm on $(\mathcal{C},D)$ with the updated pivot 
  returns a $(2-\alpha + \delta)$-approximate clustering in expectation for the input $G$ if run on parameters $\eps = \delta/2k$, $\beta = 0.5$, and $s$, where $(\mathcal{C},D)$ is a clustering and symmetric difference, $k$ is the approximation factor of $(\mathcal{C},D)$.
\end{theorem}
This can then be used both with a preclustering as the input to get an approximation algorithm, but also with a previous clustering, as would be the case in the dynamic setting.

For the purpose of this statement, we will define
\[ \delta_0 = \frac{17}{36} \left(\frac{1}{5}-\alpha\right) \quad\mathrm{and}\quad s = 1 + \left\lceil \frac{2}{\frac{1}{5}-\alpha} \right\rceil. \]

Let $\hat{\alpha} := (\alpha + 1/5)/2$ and let $0 < \delta\le \delta_0$. We are going to aim for solutions of $(2-\hat{\alpha})\cost(\mathcal{C}^*)$ which together with the definition of $\delta_0$ is enough to get a $(2-\alpha)$-approximation. Much of the proof will be the same as \cite{CLMTYZ24}, though with a new way to perform the central pivot step.

We start with a statement that has essentially been proven before:
\begin{lemma}[Lemma 25 of \cite{CLMTYZ24}]\label{lem:iterated-w-increase}
Let $\mathcal{C}_1,\ldots,\mathcal{C}_\ell$ be clusterings
and 
\[w := w_0 + \sum_{i=1}^\ell \beta (E \setminus \mathcal{E(C}_i)). \]
(That is, for every $1 \leq i \leq \ell$, we add a weight of $\beta$ to every edge
 connecting two distinct clusters of $\mathcal{C}_i$.)
Let $\mathcal{C}$ be a $\frac{\delta}{2k}$-good local optimum of $G$ with weights $w$ with respect to the clustering $(\mathcal{C},D)$. Let $\mathcal{C}^\ast$ be the clustering guaranteed by \cref{thm:preclustering} to exist, where $(\mathcal{C},D)$ is a clustering and symmetric difference pair with approximation factor $k$.
If $\cost(\mathcal{C}) > (2-\hat{\alpha}) \cost(\mathcal{C}^\ast)$, then 
\begin{align*}
    \left(\hat{\alpha} + \delta\right) \cost(\mathcal{C}^\ast)) &+ 2\beta \left(\sum_{i=1}^\ell|E \setminus \mathcal{E(C}_i) \setminus \mathcal{E(C^\ast)}|\right) \ge\\
    &\beta\sum_{i=1}^\ell |E \setminus\mathcal{E(C)} \setminus \mathcal{E(C}_i)|\\
    &+|E \setminus\mathcal{E(C)} \setminus \mathcal{E(C^\ast)}| \\
    &+|\mathcal{E(C)}\setminus E|.
    \end{align*}
\end{lemma}

Now for $1\le i \le s$ and for all edges $(u,v)$ define the distance $d_i(u,v)$ to be the number of times $u$ and $v$ are separated in $\mathcal{C}_{i-1}'$, $\mathcal{C}_{i}$ and
$\mathcal{C}_{i}'$, so $d_i(u,v) \in [0, 1, 2, 3]$. It is clear that $d_i$ is a metric, that is, for any three nodes $u, v, w$, we have $d_i(u,v) \leq d_i(u,w) + d_i(v,w)$. We set up a budget function for each edge using $d_i$, define $b_i(u,v)$ and $x_i(u,v)$ as for edges
\begin{equation*}
\forall (u,v) \in E: b^{+}_i(u,v) = x_i(u,v) = \begin{cases}
0 & d_i(u,v) = 0\\
0 & d_i(u,v) = 1 \\
1 & d_i(u,v) = 2 \\
3 & d_i(u,v) = 3
\end{cases}
\end{equation*}
and for non-edges as
\begin{equation*}
\forall (u,v) \in \binom{V}{2}\setminus E: b^{-}_i(u,v) = 1 - x_i(u,v) = \begin{cases}
3 & d_i(u,v) = 0\\
2 & d_i(u,v) = 1 \\
1 & d_i(u,v) = 2 \\
0 & d_i(u,v) = 3 \\
\end{cases}
\end{equation*}
then we have
\begin{align}
    \sum_{(u,v) \in E^{+}} b^{+}_i(u,v) + \sum_{(u,v) \not\in E} b^{-}_i(u,v) \leq & \quad |\mathcal{E}(\mathcal{C}_{i-1}')\setminus E| + |\mathcal{E}(\mathcal{C}_{i})\setminus E| + |\mathcal{E}(\mathcal{C}_{i}')\setminus E| \nonumber\\
    &+ |E \setminus \mathcal{E}(\mathcal{C}_{i-1}') \setminus \mathcal{E}(\mathcal{C}_{i})| \nonumber\\ &+ |E \setminus \mathcal{E}(\mathcal{C}_{i-1}') \setminus \mathcal{E}(\mathcal{C}_{i}')| \nonumber\\ 
    &+ |E \setminus \mathcal{E}(\mathcal{C}_{i}) \setminus \mathcal{E}(\mathcal{C}_{i}')| \label{eq:budget-bound}
\end{align}
Now consider the following pivot algorithm: we randomly choose a pivot $u$ and then for any edge $(u,v) \in E$, we add $v$ to $u$'s cluster if $x(u,v) = 0$, add $v$ to $u$'s clustering with probability $1/4$ if $x(u,v) \geq 1$; for any edge $(u,v) \not\in E^{-}$, we will add $v$ to $u$'s cluster if $x(u,v) < 0$ and add $v$ to $u$'s cluster with probability $3/4$ if $x(u,v) = 0$. Then we remove the $u$'s cluster and repeat the whole process until all nodes have been clustered.\footnote{Given the symmetric difference $D$, it should be clear that this can be done in $O(|D|)$ time.} We can show that the cost of the pivot algorithm can be bounded by 
\begin{align}
\label{eq:localsearch}
    \cost(C_i'')\le 1.5\left(\sum_{(u,v) \in E} b^{+}_i(u,v) + \sum_{(u,v) \not\in E} b^{-}_i(u,v)\right)
\end{align}

One can use a standard triangle analysis to show \cref{eq:localsearch}. For any triangle $(u,v,w)$, let $\cost((u,v,w))$ and $\mathrm{lp}((u,v,w))$ be the cost of the algorithm and the standard linear programming. Then it can be to shown that $\cost(u,v,w) - 1.5\mathrm{lp}(u,v,w) \leq 0$ for all possible triangles $(u,v,w)$.

Combining \cref{eq:budget-bound} and \cref{eq:localsearch} we get
\begin{align}
    \frac{2}{3}\cost(C_i'')\leq & \quad |\mathcal{E}(\mathcal{C}_{i-1}')\setminus E| + |\mathcal{E}(\mathcal{C}_{i})\setminus E| + |\mathcal{E}(\mathcal{C}_{i}')\setminus E| \nonumber\\
    &+ |E \setminus \mathcal{E}(\mathcal{C}_{i-1}') \setminus \mathcal{E}(\mathcal{C}_{i})| \nonumber\\ &+ |E \setminus \mathcal{E}(\mathcal{C}_{i-1}') \setminus \mathcal{E}(\mathcal{C}_{i}')| \nonumber\\ 
    &+ |E \setminus \mathcal{E}(\mathcal{C}_{i}) \setminus \mathcal{E}(\mathcal{C}_{i}')| \label{eq:cpp-bound}
\end{align}

\begin{proof}[Proof of \cref{thm:improved-approx-ratio}]
    Using \cref{eq:localsearch} the proof is essentially equivalent to that of \cite{CLMTYZ24}.
    
    For the purpose of contradiction we assume that for all $\mathcal{C}\in \{\mathcal{C}_i,\mathcal{C}_i',\mathcal{C}_i''\}$, we have $\cost(\mathcal{C}) > (2-\hat{\alpha})\cost(\mathcal{C}^*$.
    
    First, we apply \cref{lem:iterated-w-increase} to $\mathcal{C}_0'$ to get
    \begin{equation}\label{eq:c0-bound}
        \left(\hat{\alpha} + \delta\right) \cost(\mathcal{C}^\ast))\ge
        |E \setminus\mathcal{E(C}_0') \setminus \mathcal{E(C^\ast)}|
        +|\mathcal{E(C}_0')\setminus E|.
    \end{equation}
    
    For any $1\le i \le s$, applying \cref{lem:iterated-w-increase} to $\mathcal{C}_i$ and using $\beta=\frac{1}{2}$ gives
    \begin{align}
        \left(\hat{\alpha} + \delta\right) \cost(\mathcal{C}^\ast)) &+ |E \setminus \mathcal{E(C}_{i-1}') \setminus \mathcal{E(C^\ast)}| \ge\nonumber\\
        &0.5|E \setminus\mathcal{E}(\mathcal{C}_i) \setminus \mathcal{E(C}_{i-1}')|
        +|E \setminus\mathcal{E}(\mathcal{C}_i) \setminus \mathcal{E(C^\ast)}|
        +|\mathcal{E}(\mathcal{C}_i)\setminus E|.\label{eq:ci-bound}
    \end{align}
    Similarly applying \cref{lem:iterated-w-increase} to $\mathcal{C}_i'$ we have:
    \begin{align}
        \left(\hat{\alpha} + \delta\right) \cost(\mathcal{C}^\ast)) &+ |E \setminus \mathcal{E(C}_{i-1}') \setminus \mathcal{E(C^\ast)}|+|E \setminus \mathcal{E(C}_i) \setminus \mathcal{E(C^\ast)}| \ge\nonumber\\
        &0.5 |E \setminus\mathcal{E}(\mathcal{C}_{i-1}') \setminus \mathcal{E(C}_i')|
        +0.5 |E \setminus\mathcal{E}(\mathcal{C}_i) \setminus \mathcal{E(C}_i')|
        +|E \setminus\mathcal{E}(\mathcal{C}_i') \setminus \mathcal{E(C^\ast)}|
        +|\mathcal{E}(\mathcal{C}_i')\setminus E|.\label{eq:cip-bound}
    \end{align}
    To bound the right-hand side of \cref{eq:cpp-bound} we use twice \cref{eq:ci-bound} and twice \cref{eq:cip-bound} to get
    \begin{align}
        \left(4\hat{\alpha} + 4\delta\right) \cost(\mathcal{C}^\ast)) &+ 4|E \setminus \mathcal{E(C}_{i-1}') \setminus \mathcal{E(C^\ast)}| \ge\nonumber\\
        &|E \setminus\mathcal{E}(\mathcal{C}_i) \setminus \mathcal{E(C}_{i-1}')|
        + |E \setminus\mathcal{E}(\mathcal{C}_{i-1}') \setminus \mathcal{E(C}_i')|
        + |E \setminus\mathcal{E}(\mathcal{C}_i) \setminus \mathcal{E(C}_i')|\nonumber\\
        &+2|E \setminus\mathcal{E}(\mathcal{C}_i') \setminus \mathcal{E(C^\ast)}|
        +2|\mathcal{E}(\mathcal{C}_i')\setminus E|+2|\mathcal{E}(\mathcal{C}_i)\setminus E|
    \end{align}
    Combining with \cref{eq:cpp-bound} we get 
    \begin{align}
        \left(4\hat{\alpha} +4\delta\right) \cost(\mathcal{C}^\ast)) &+|\mathcal{E}(\mathcal{C}_{i-1}')\setminus E|+ 4|E \setminus \mathcal{E(C}_{i-1}') \setminus \mathcal{E(C^\ast)}| \ge\nonumber\\
        &\frac{2}{3}\cost(C_i'')+2|E \setminus\mathcal{E}(\mathcal{C}_i') \setminus \mathcal{E(C^\ast)}|
        +|\mathcal{E}(\mathcal{C}_i')\setminus E|+|\mathcal{E}(\mathcal{C}_i)\setminus E|\label{eq:combined-it-bound}
    \end{align}
    Using the assumption that $(2-\hat{\alpha})\cost(\mathcal{C}^\ast)<\cost(C_i'')$ we rewrite \cref{eq:combined-it-bound} to
    \begin{align}
        \left(\frac{14}{3}\hat{\alpha} +4\delta-\frac{4}{3}\right) \cost(\mathcal{C}^\ast)) &+|\mathcal{E}(\mathcal{C}_{i-1}')\setminus E|+ 4|E \setminus \mathcal{E(C}_{i-1}') \setminus \mathcal{E(C^\ast)}| \ge\nonumber\\
        &2|E \setminus\mathcal{E}(\mathcal{C}_i') \setminus \mathcal{E(C^\ast)}|
        +|\mathcal{E}(\mathcal{C}_i')\setminus E|+|\mathcal{E}(\mathcal{C}_i)\setminus E|\label{eq:combined-it-bound-cost}
    \end{align}
    From this expression, we motivate defining 
    \[b_i := \frac{2|E \setminus\mathcal{E}(\mathcal{C}_i') \setminus \mathcal{E(C^\ast)}| + |\mathcal{E}(\mathcal{C}_i')\setminus E|+|\mathcal{E}(\mathcal{C}_i)\setminus E|}{\cost(\mathcal{C}^\ast))}\]
    Using this definition we can rewrite \cref{eq:combined-it-bound-cost} to
    \begin{equation}\label{eq:bim1-bi-relation}
        b_{i-1} + \frac{14}{3}\hat{\alpha} + 4\delta - \frac{4}{3} > b_i-b_{i-1}
    \end{equation}
    By defining $\mathcal{C}_0$ to be a clustering of only singletons combined with \cref{eq:c0-bound} we have 
    \begin{equation}\label{eq:b0-bound}
        b_0\le 2\alpha+2\delta \le 1
    \end{equation}
    Since every $b_i$ is nonnegative and by the choice of $s$ there exists an index $j$ such that $b_j > b_{j-1}-(\frac{1}{5}-\hat{\alpha})$. Let $j_0$ be the smallest such index. Now by \cref{eq:b0-bound} and the definition of $j_0$, for every $0\leq i< j_0$ we have
    \[b_i\leq 2\alpha + 2\delta\]
    Combining this with \cref{eq:bim1-bi-relation} we get the equation
    \[\frac{20}{3}\hat{\alpha} + 6\delta - \frac{4}{3} > \hat{\alpha}-\frac{1}{5}\]
    which can be rewritten to
    \[\hat{\alpha} + \frac{18}{17}\delta > \frac{1}{5}.\]
    By plugging in the definition of $\hat{\alpha}$ we get 
    \[\alpha + \frac{36}{17}\delta > \frac{1}{5}\]
    This is however a contradiction of $\alpha < \frac{1}{5}$ and the choice of $\delta < \delta_0$.
\end{proof}

%% file: main.bbl
\begin{thebibliography}{10}

\bibitem{agrawal2009generating}
Rakesh Agrawal, Alan Halverson, Krishnaram Kenthapadi, Nina Mishra, and Panayiotis Tsaparas.
\newblock Generating labels from clicks.
\newblock In {\em Proceedings of the Second ACM International Conference on Web Search and Data Mining}, pages 172--181, 2009.

\bibitem{ACN08}
Nir Ailon, Moses Charikar, and Alantha Newman.
\newblock Aggregating inconsistent information: {R}anking and clustering.
\newblock {\em Journal of the ACM}, 55(5):1--27, 2008.

\bibitem{arasu2009large}
Arvind Arasu, Christopher R{\'e}, and Dan Suciu.
\newblock Large-scale deduplication with constraints using dedupalog.
\newblock In {\em Proceedings of the 25th IEEE International Conference on Data Engineering (ICDE)}, pages 952--963, 2009.

\bibitem{assadi2025cc}
Sepehr Assadi, Sanjeev Khanna, and Aaron Putterman.
\newblock Correlation clustering and (de)sparsification: Graph sketches can match classical algorithms.
\newblock In Michal Kouck{\'{y}} and Nikhil Bansal, editors, {\em Proceedings of the 57th Annual {ACM} Symposium on Theory of Computing, {STOC} 2025, Prague, Czechia, June 23-27, 2025}, pages 417--428. {ACM}, 2025.
\newblock \href {https://doi.org/10.1145/3717823.3718194} {\path{doi:10.1145/3717823.3718194}}.

\bibitem{DBLP:conf/innovations/Assadi022}
Sepehr Assadi and Chen Wang.
\newblock Sublinear time and space algorithms for correlation clustering via sparse-dense decompositions.
\newblock In {\em Proceedings of the 13th Conference on Innovations in Theoretical Computer Science (ITCS)}, volume 215 of {\em LIPIcs}, pages 10:1--10:20, 2022.

\bibitem{BBC04}
Nikhil Bansal, Avrim Blum, and Shuchi Chawla.
\newblock Correlation clustering.
\newblock {\em Machine learning}, 56(1):89--113, 2004.

\bibitem{behnezhad2024fullydynamicCC}
Soheil Behnezhad, Moses Charikar, Vincent {Cohen-Addad}, Alma Ghafari, and Weiyun Ma.
\newblock Fully dynamic correlation clustering: Breaking 3-approximation, 2024.
\newblock URL: \url{https://arxiv.org/abs/2404.06797}, \href {https://arxiv.org/abs/2404.06797} {\path{arXiv:2404.06797}}.

\bibitem{Behnezhad2019FullyDM}
Soheil Behnezhad, Mahsa Derakhshan, Mohammad~Taghi Hajiaghayi, Clifford Stein, and Madhu Sudan.
\newblock Fully dynamic maximal independent set with polylogarithmic update time.
\newblock {\em 2019 IEEE 60th Annual Symposium on Foundations of Computer Science (FOCS)}, pages 382--405, 2019.
\newblock URL: \url{https://api.semanticscholar.org/CorpusID:202539552}.

\bibitem{10.1145/3519935.3520064}
Amos Beimel, Haim Kaplan, Yishay Mansour, Kobbi Nissim, Thatchaphol Saranurak, and Uri Stemmer.
\newblock Dynamic algorithms against an adaptive adversary: generic constructions and lower bounds.
\newblock In {\em Proceedings of the 54th Annual ACM SIGACT Symposium on Theory of Computing}, STOC 2022, page 1671–1684, New York, NY, USA, 2022. Association for Computing Machinery.
\newblock \href {https://doi.org/10.1145/3519935.3520064} {\path{doi:10.1145/3519935.3520064}}.

\bibitem{bonchi2013overlapping}
Francesco Bonchi, Aristides Gionis, and Antti Ukkonen.
\newblock Overlapping correlation clustering.
\newblock {\em Knowledge and Information Systems}, 35(1):1--32, 2013.

\bibitem{braverman2024fully}
Vladimir Braverman, Prathamesh Dharangutte, Shreyas Pai, Vihan Shah, and Chen Wang.
\newblock Fully dynamic adversarially robust correlation clustering in polylogarithmic update time.
\newblock {\em arXiv preprint arXiv:2411.09979}, 2024.

\bibitem{BravermanDPSW25}
Vladimir Braverman, Prathamesh Dharangutte, Shreyas Pai, Vihan Shah, and Chen Wang.
\newblock Fully dynamic adversarially robust correlation clustering in polylogarithmic update time.
\newblock In Yingzhen Li, Stephan Mandt, Shipra Agrawal, and Mohammad~Emtiyaz Khan, editors, {\em International Conference on Artificial Intelligence and Statistics, {AISTATS} 2025, Mai Khao, Thailand, 3-5 May 2025}, volume 258 of {\em Proceedings of Machine Learning Research}, pages 1477--1485. {PMLR}, 2025.
\newblock URL: \url{https://proceedings.mlr.press/v258/braverman25a.html}.

\bibitem{cao2025fastLP}
Nairen Cao, Vincent {Cohen-Addad}, Euiwoong Lee, Shi Li, David~Rasmussen Lolck, Alantha Newman, Mikkel Thorup, Lukas Vogl, Shuyi Yan, and Hanwen Zhang.
\newblock Solving the correlation cluster {LP} in sublinear time.
\newblock In {\em Proceedings of the 57th Annual ACM Symposium on Theory of Computing (STOC)}, pages 1154--1165, 2025.
\newblock The approximation factor of 1.437 was incorrect and fixed to 1.485 in \cite{cao2025fastLP:arxiv}.

\bibitem{cao2024understanding}
Nairen Cao, Vincent {Cohen-Addad}, Euiwoong Lee, Shi Li, Alantha Newman, and Lukas Vogl.
\newblock Understanding the cluster linear program for correlation clustering.
\newblock In {\em Proceedings of the 56th Annual ACM Symposium on Theory of Computing (STOC)}, pages 1605--1616, 2024.
\newblock The approximation factor of 1.437 was incorrect and fixed to 1.485 in \cite{cao2024understanding:arxiv}.

\bibitem{cao2024understanding:arxiv}
Nairen Cao, Vincent Cohen{-}Addad, Euiwoong Lee, Shi Li, Alantha Newman, and Lukas Vogl.
\newblock Understanding the cluster {LP} for correlation clustering.
\newblock {\em CoRR}, abs/2404.17509, 2024.
\newblock URL: \url{https://doi.org/10.48550/arXiv.2404.17509}, \href {https://arxiv.org/abs/2404.17509} {\path{arXiv:2404.17509}}, \href {https://doi.org/10.48550/ARXIV.2404.17509} {\path{doi:10.48550/ARXIV.2404.17509}}.

\bibitem{cao2025fastLP:arxiv}
Nairen Cao, Vincent Cohen{-}Addad, Shi Li, Euiwoong Lee, David~Rasmussen Lolck, Alantha Newman, Mikkel Thorup, Lukas Vogl, Shuyi Yan, and Hanwen Zhang.
\newblock Solving the correlation cluster {LP} in sublinear time.
\newblock {\em CoRR}, abs/2503.20883, 2025.
\newblock URL: \url{https://doi.org/10.48550/arXiv.2503.20883}, \href {https://arxiv.org/abs/2503.20883} {\path{arXiv:2503.20883}}, \href {https://doi.org/10.48550/ARXIV.2503.20883} {\path{doi:10.48550/ARXIV.2503.20883}}.

\bibitem{chakrabarti2008graph}
Deepayan Chakrabarti, Ravi Kumar, and Kunal Punera.
\newblock A graph-theoretic approach to webpage segmentation.
\newblock In {\em Proceedings of the 17th International conference on World Wide Web (WWW)}, pages 377--386, 2008.

\bibitem{ChakrabartyM-NIPS23}
Sayak Chakrabarty and Konstantin Makarychev.
\newblock Single-pass pivot algorithm for correlation clustering. keep it simple!
\newblock In {\em Advances in Neural Information Processing Systems (NeurIPS)}, 2023.
\newblock URL: \url{https://doi.org/10.48550/arXiv.2305.13560}, \href {https://arxiv.org/abs/2305.13560} {\path{arXiv:2305.13560}}, \href {https://doi.org/10.48550/ARXIV.2305.13560} {\path{doi:10.48550/ARXIV.2305.13560}}.

\bibitem{CGW05}
Moses Charikar, Venkatesan Guruswami, and Anthony Wirth.
\newblock Clustering with qualitative information.
\newblock {\em Journal of Computer and System Sciences}, 71(3):360--383, 2005.

\bibitem{CMSY15}
Shuchi Chawla, Konstantin Makarychev, Tselil Schramm, and Grigory Yaroslavtsev.
\newblock Near optimal {LP} rounding algorithm for correlation clustering on complete and complete $k$-partite graphs.
\newblock In {\em Proceedings of the 47th annual ACM Symposium on Theory of Computing (STOC)}, pages 219--228, 2015.

\bibitem{chen2012clustering}
Yudong Chen, Sujay Sanghavi, and Huan Xu.
\newblock Clustering sparse graphs.
\newblock In {\em Advances in Neural Information Processing Systems (Neurips)}, pages 2204--2212, 2012.

\bibitem{CLMNP21}
Vincent {Cohen{-}Addad}, Silvio Lattanzi, Slobodan Mitrovic, Ashkan Norouzi{-}Fard, Nikos Parotsidis, and Jakub Tarnawski.
\newblock Correlation clustering in constant many parallel rounds.
\newblock In {\em Proceedings of the 38th International Conference on Machine Learning (ICML)}, pages 2069--2078, 2021.

\bibitem{CLLN23+}
Vincent {Cohen{-}Addad}, Euiwoong Lee, Shi Li, and Alantha Newman.
\newblock Handling correlated rounding error via preclustering: A 1.73-approximation for correlation clustering.
\newblock In {\em Proceedings of 64th Annual IEEE Symposium on Foundations of Computer Science (FOCS)}, pages 1082--1104, 2023.

\bibitem{CLN22}
Vincent {Cohen-Addad}, Euiwoong Lee, and Alantha Newman.
\newblock Correlation clustering with {S}herali-{A}dams.
\newblock In {\em Proceedings of 63rd Annual IEEE Symposium on Foundations of Computer Science (FOCS)}, pages 651--661, 2022.

\bibitem{CLMTYZ24}
Vincent {Cohen-Addad}, David~Rasmussen Lolck, Marcin Pilipczuk, Mikkel Thorup, Shuyi Yan, and Hanwen Zhang.
\newblock Combinatorial correlation clustering.
\newblock In {\em Proceedings of the 56th Annual ACM Symposium on Theory of Computing (STOC)}, pages 1617--1628, 2024.
\newblock URL: \url{http://dx.doi.org/10.1145/3618260.3649712}, \href {https://doi.org/10.1145/3618260.3649712} {\path{doi:10.1145/3618260.3649712}}.

\bibitem{DMM2024prunedPivot}
Mina Dalirrooyfard, Konstantin Makarychev, and Slobodan Mitrović.
\newblock Pruned pivot: Correlation clustering algorithm for dynamic, parallel, and local computation models, 2024.
\newblock URL: \url{https://arxiv.org/abs/2402.15668}, \href {https://arxiv.org/abs/2402.15668} {\path{arXiv:2402.15668}}.

\bibitem{fischer2025faster}
Nick Fischer, Evangelos Kipouridis, Jonas Klausen, and Mikkel Thorup.
\newblock A faster algorithm for constrained correlation clustering.
\newblock {\em arXiv preprint arXiv:2501.03154}, 2025.
\newblock Conference version to appear at STACS'25.

\bibitem{kalashnikov2008web}
Dmitri~V. Kalashnikov, Zhaoqi Chen, Sharad Mehrotra, and Rabia Nuray-Turan.
\newblock Web people search via connection analysis.
\newblock {\em IEEE Transactions on Knowledge and Data Engineering}, 20(11):1550--1565, 2008.

\bibitem{PanPORRJ15}
Xinghao Pan, Dimitris~S. Papailiopoulos, Samet Oymak, Benjamin Recht, Kannan Ramchandran, and Michael~I. Jordan.
\newblock Parallel correlation clustering on big graphs.
\newblock In {\em Advances in Neural Information Processing Systems 28: Annual Conference on Neural Information Processing Systems 2015, December 7-12, 2015, Montreal, Quebec, Canada}, pages 82--90, 2015.

\bibitem{ShiDELM21}
Jessica Shi, Laxman Dhulipala, David Eisenstat, Jakub Lacki, and Vahab~S. Mirrokni.
\newblock Scalable community detection via parallel correlation clustering.
\newblock {\em Proc. {VLDB} Endow.}, 14(11):2305--2313, 2021.

\bibitem{DBLP:journals/mor/ZuylenW09}
Anke van Zuylen and David~P. Williamson.
\newblock Deterministic pivoting algorithms for constrained ranking and clustering problems.
\newblock {\em Math. Oper. Res.}, 34(3):594--620, 2009.
\newblock URL: \url{https://doi.org/10.1287/moor.1090.0385}, \href {https://doi.org/10.1287/MOOR.1090.0385} {\path{doi:10.1287/MOOR.1090.0385}}.

\bibitem{pmlr-v162-veldt22a}
Nate Veldt.
\newblock Correlation clustering via strong triadic closure labeling: Fast approximation algorithms and practical lower bounds.
\newblock In Kamalika Chaudhuri, Stefanie Jegelka, Le~Song, Csaba Szepesvari, Gang Niu, and Sivan Sabato, editors, {\em Proceedings of the 39th International Conference on Machine Learning}, volume 162 of {\em Proceedings of Machine Learning Research}, pages 22060--22083. PMLR, 17--23 Jul 2022.
\newblock URL: \url{https://proceedings.mlr.press/v162/veldt22a.html}.

\end{thebibliography}
